\documentclass[11pt,letterpaper]{article}

\usepackage[letterpaper,margin=1in,headheight=14pt,headsep=18pt]{geometry}
\usepackage[T1]{fontenc}
\usepackage{amsmath,amsthm,mathtools}
\usepackage{newtxtext}
\usepackage[varvw]{newtxmath}
\usepackage{microtype}
\usepackage{setspace}
\usepackage{titlesec}
\usepackage{fancyhdr}
\usepackage{needspace}
\usepackage{booktabs}
\usepackage{array}
\usepackage{enumitem}
\usepackage{xcolor}
\definecolor{headcol}{HTML}{1F3A5F}   
\usepackage[font=small,labelfont=bf,labelsep=period,skip=6pt]{caption}
\usepackage{subcaption}
\usepackage{tikz}
\usetikzlibrary{petri,positioning,arrows.meta,calc,fit,backgrounds,
  shapes.geometric,shapes.misc,decorations.pathreplacing,matrix}
\usepackage[most]{tcolorbox}
\usepackage[colorlinks=true,linkcolor=headcol,citecolor=headcol,urlcolor=headcol]{hyperref}
\usepackage{algorithm}
\usepackage[noend]{algpseudocode}
\usepackage{cleveref}
\hypersetup{pdftitle={Game Petri Nets: Compositional Modeling and Analysis of Large-Scale Strategic Systems},pdfauthor={Quanyan Zhu}}

\theoremstyle{definition}
\newtheorem{definition}{Definition}

\newtheorem{assumption}{Assumption}
\theoremstyle{plain}
\newtheorem{theorem}{Theorem}
\newtheorem{lemma}{Lemma}
\newtheorem{corollary}{Corollary}
\newtheorem{proposition}{Proposition}
\newtheorem{conjecture}{Conjecture}
\theoremstyle{remark}
\newtheorem{remark}{Remark}
\crefname{definition}{definition}{definitions}
\crefname{assumption}{assumption}{assumptions}
\crefname{theorem}{theorem}{theorems}
\crefname{lemma}{lemma}{lemmas}
\crefname{corollary}{corollary}{corollaries}
\crefname{proposition}{proposition}{propositions}
\crefname{conjecture}{conjecture}{conjectures}
\crefname{remark}{remark}{remarks}
\crefname{algorithm}{algorithm}{algorithms}
\Crefname{assumption}{Assumption}{Assumptions}
\Crefname{definition}{Definition}{Definitions}
\Crefname{theorem}{Theorem}{Theorems}
\Crefname{lemma}{Lemma}{Lemmas}
\Crefname{corollary}{Corollary}{Corollaries}
\Crefname{proposition}{Proposition}{Propositions}
\Crefname{remark}{Remark}{Remarks}

\titleformat{\section}{\needspace{10\baselineskip}\normalfont\Large\bfseries\color{headcol}}{\thesection}{0.8em}{}
\titleformat{\subsection}{\needspace{7\baselineskip}\normalfont\large\bfseries\color{headcol}}{\thesubsection}{0.8em}{}
\titleformat{\paragraph}[runin]{\normalfont\normalsize\bfseries}{}{}{}
\titlespacing*{\section}{0pt}{18pt plus 4pt minus 2pt}{8pt plus 2pt}
\titlespacing*{\subsection}{0pt}{12pt plus 3pt minus 2pt}{6pt plus 2pt}
\titlespacing*{\paragraph}{0pt}{6pt plus 2pt}{0.6em}

\fancypagestyle{plain}{\fancyhf{}\fancyfoot[C]{\small\thepage}}

\definecolor{gamecol}{HTML}{E8833A}   
\definecolor{rescol}{HTML}{3A7DC9}    
\definecolor{okcol}{HTML}{2E9E5B}     
\definecolor{badcol}{HTML}{C8423B}    
\definecolor{inkcol}{HTML}{333333}
\definecolor{softgray}{HTML}{F3F3F3}

\newtcolorbox{intuition}[1][Intuition]{enhanced,breakable,colback=blue!3,
  colframe=blue!40!black,fonttitle=\bfseries\small,title={#1},
  boxrule=0.6pt,arc=2pt,left=6pt,right=6pt,top=4pt,bottom=4pt,before upper={\setstretch{1.15}}}
\newtcolorbox{designrule}[1][Design rule]{enhanced,breakable,colback=orange!4,
  colframe=gamecol!80!black,fonttitle=\bfseries\small,title={#1},
  boxrule=0.6pt,arc=2pt,left=6pt,right=6pt,top=4pt,bottom=4pt,before upper={\setstretch{1.15}}}

\newcommand{\M}{\mathbb{M}}
\newcommand{\Mult}{\mathrm{Mult}_{\mathrm{fin}}}
\newcommand{\Pre}{\mathrm{Pre}}
\newcommand{\Post}{\mathrm{Post}}
\newcommand{\Out}{\mathrm{Out}}
\newcommand{\Obs}{\mathrm{Obs}}
\newcommand{\Reach}{\mathrm{Reach}}
\newcommand{\Spec}{\mathcal{S}}
\newcommand{\G}{\mathcal{G}}
\newcommand{\N}{\mathcal{N}}
\newcommand{\Ex}{\mathbb{E}}
\newcommand{\ind}{\mathbb{I}}
\newcommand{\tok}[1]{\langle #1\rangle}
\newcommand{\act}[1]{\mathrm{#1}}
\newcommand{\PA}{P_{\!A}}
\newcommand{\PD}{P_{\!D}}

\tikzset{
  >={Stealth[length=2mm,width=1.7mm]},
  every place/.style={draw=rescol!85!black,line width=0.8pt,fill=rescol!7,minimum size=7.5mm,inner sep=0pt},
  gplace/.style={place,draw=gamecol!90!black,fill=gamecol!14,double,double distance=0.9pt,minimum size=8.5mm},
  rplace/.style={place},
  every transition/.style={draw=none,fill=inkcol,minimum width=1.8mm,minimum height=7mm,inner sep=0pt},
  rtrans/.style={transition,fill=gamecol!90!black},
  htrans/.style={transition,minimum width=7mm,minimum height=1.8mm},
  arc/.style={->,line width=0.7pt,draw=inkcol},
  darc/.style={->,line width=0.7pt,dashed,draw=inkcol!60},
  ctok/.style={circle,fill=#1,draw=black!50,line width=0.4pt,inner sep=0pt,minimum size=2.2mm},
  every token/.style={fill=inkcol,minimum size=1.7mm},
  lbl/.style={font=\footnotesize,inner sep=1.5pt,text=inkcol},
  glbl/.style={font=\footnotesize\itshape,text=gamecol!65!black,inner sep=1.5pt,align=center},
  guard/.style={font=\scriptsize,inner sep=1.5pt,text=inkcol},
  cbox/.style={draw=gamecol!80!black,fill=gamecol!5,rounded corners=2pt,line width=0.7pt,inner sep=4pt,font=\footnotesize},
  sub/.style={draw=inkcol!50,dashed,line width=0.7pt,rounded corners=3pt,minimum width=1.6cm,minimum height=1.3cm,font=\footnotesize},
  stepbox/.style={draw=inkcol!70,fill=softgray,rounded corners=3pt,line width=0.7pt,align=center,font=\footnotesize,minimum height=14mm,text width=4.1cm},
}

\newcommand{\payoffgrid}[7]{%
  \begin{scope}[shift={#1}]
    \foreach \x/\y/\c/\n in {0/0/{#3}/a, 2.2/0/{#4}/b, 0/-1.3/{#5}/c, 2.2/-1.3/{#6}/d}{
      \node[draw=inkcol!40,fill=white,minimum width=1.8cm,minimum height=0.9cm,font=\small] (#2\n) at (\x,\y) {\c};
    }
    \node[font=\footnotesize,anchor=east] at ($(#2a.west)+(-0.55,0)$) {$\act{H}$};
    \node[font=\footnotesize,anchor=east] at ($(#2c.west)+(-0.55,0)$) {$\act{L}$};
    \node[font=\footnotesize,anchor=south] at ($(#2a.north)+(0,0.42)$) {$\act{H}$};
    \node[font=\footnotesize,anchor=south] at ($(#2b.north)+(0,0.42)$) {$\act{L}$};
    \node[font=\footnotesize\bfseries,align=center] at ($(#2a.north)!0.5!(#2b.north)+(0,1.05)$) {#7};
  \end{scope}%
}

\begin{document}
\thispagestyle{plain}

\begin{center}
  {\small\scshape\color{inkcol!80} Working paper \quad\textbullet\quad Draft of October 2026}\\[8pt]
  {\color{headcol}\rule{\textwidth}{0.9pt}}\\[14pt]
  {\LARGE\bfseries\color{headcol} Game Petri Nets: Compositional Modeling\\[4pt] and Analysis of Large-Scale Strategic Systems\par}
  \vspace{14pt}
  {\large Quanyan Zhu\par}
  \vspace{3pt}
  {\small Department of Electrical and Computer Engineering, Tandon School of Engineering\\ New York University, Brooklyn, NY 11201, USA\\[2pt] \href{mailto:qz494@nyu.edu}{\texttt{qz494@nyu.edu}}\par}
  \vspace{12pt}
  {\color{headcol}\rule{\textwidth}{0.4pt}}
\end{center}

\begin{center}
\begin{minipage}{0.88\textwidth}
\small
\noindent{\bfseries\scshape\color{headcol} Abstract.}\enspace Large-scale strategic systems, such as maintenance workflows, cyber campaigns, and multi-agent AI pipelines, consist of many interactions whose outcomes determine which later interactions become possible, who takes part, and with what resources. We introduce \emph{game Petri nets}, formalized as outcome-driven compositional game Petri nets (OCGPNs), to model and analyze such systems. Local games are attached to places of a colored Petri net. When a game resolves, a game-to-token interface converts the binding, the action profile, and the outcome into typed tokens, which enable, synchronize, or block subsequent games. The model separates a structural net, local games, and a solution layer, and associates with them the incidence matrix, the local payoff and outcome-probability matrices, and the routing matrix between games. For general, concurrent nets, we give an operational semantics and a structural theory based on invariants and Farkas certificates. These guarantee safety for every admissible strategy profile, and they show that guards can certify unreachability whereas incentive instruments can only change values. For sequential nets, in which one game is active at a time, the net induces a stochastic game on its game places. We prove that global equilibria are exactly collections of equilibria of local games augmented by continuation values. We characterize when downstream games change upstream behavior through successor-level pivotality, and we compute equilibrium values place by place by decomposing the place graph into strongly connected components, at a cost additive in the sizes of the local games. A composition calculus based on place fusion glues certificates across subnets and summarizes a downstream subnet by its entry values, also in the general-sum case. Two running examples illustrate the theory. In a maintenance workflow, graded routing reverses a myopic equilibrium, binary routing cannot, and a rework loop sustains cooperation only if rework is costly. In an intrusion campaign derived from multi-resolution cyber-defense games, deception lowers the attacker's value without removing reachability, and a newly discovered phishing scenario is composed with the existing campaign without re-solving it.

\smallskip
\noindent{\bfseries\scshape\color{headcol} Keywords.}\enspace Petri nets; colored tokens; cyber deception; multi-resolution games; stochastic games; game composition; strategic workflow; equilibrium selection; formal verification; agentic systems.
\end{minipage}
\end{center}
\vspace{6pt}

\doublespacing

\section{Introduction}
\label{sec:intro}

Many engineered systems are \emph{large-scale strategic systems}. They are built from many interactions among self-interested or adversarial agents, and the outcome of each interaction shapes what happens next. A maintenance job is a chain of negotiations, inspections, and authorizations among an operator, a provider, and a verifier. A cyber campaign is a sequence of engagements across network segments, each a contest between attack techniques and countermeasures. A multi-agent AI workflow is a pipeline of delegations, verifications, and hand-offs. In each case, an interaction does more than pay its players. Its outcome determines \emph{which interaction becomes possible next}, \emph{who takes part in it}, \emph{with what resources}, and \emph{with what information}. Interactions also run concurrently, compete for shared resources, and must sometimes synchronize before the system can proceed.

\paragraph{Limitations of existing formalisms.} Existing formalisms each capture part of this picture.
\begin{itemize}[nosep]
  \item \emph{Monolithic games.} An extensive-form game over the whole system grows multiplicatively with the number of stages. A stochastic game over global states can represent the same information, but it is not modular: the identity of the sub-interactions and of their interfaces is lost, and concurrency or synchronization must be encoded in the state.
  \item \emph{Petri nets.} Petri nets and colored Petri nets represent concurrency, resources, synchronization, and reachability directly, and they come with structural analysis by incidence matrices and invariants~\cite{murata1989,jensen2009}. But they have no strategic players.
  \item \emph{Petri games.} Petri games~\cite{finkbeiner2017} add players by identifying them with tokens, and they ask a qualitative question, the synthesis of winning strategies. They do not model payoffs, values, or equilibrium selection.
  \item \emph{Composed games.} Meta-games, games in games, and compositional game theory compose strategic interactions (\cref{sec:background}). They do not provide a structural layer with resources, concurrency, and reachability.
\end{itemize}
What is missing is a representation in which each interaction is a local game, outcomes become typed objects that govern what happens next, and analysis \emph{composes}: properties of the parts combine into properties of the whole. This leads to the central question of the paper:

\begin{center}
\emph{How can the outcome of a local game activate a transition to a new game while keeping genuine Petri-net semantics, and how can the resulting system be analyzed compositionally?}
\end{center}

\paragraph{Approach.} We attach a local game to a \emph{game place} of a colored Petri net. The completion of the game is a \emph{resolution transition}, and a \emph{game-to-token interface} $\Out_p(b,a,\omega)$ turns the binding, the action profile, and the outcome into colored tokens. These tokens then enable ordinary net transitions, including guarded branches, forks, and synchronizations. Exactly one successor game need not activate. One outcome can spawn several games, and several completed games may be jointly required before a successor is enabled. \Cref{fig:pipeline} summarizes the chain. We call the resulting model an \emph{outcome-driven compositional game Petri net} (OCGPN).

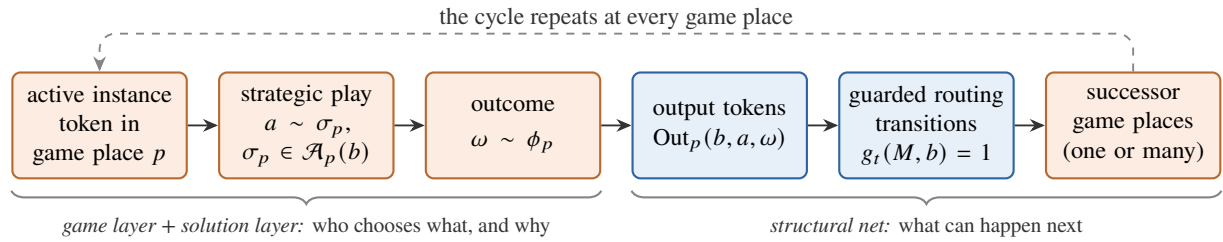
\begin{figure}[ht]
\centering
\begin{tikzpicture}[node distance=4mm,
  pbox/.style={draw,thick,rounded corners=3pt,align=center,font=\footnotesize,minimum height=14mm,text width=2.05cm},
  layer/.style={font=\scriptsize\itshape,text=inkcol}]
  \node[pbox,fill=gamecol!15,draw=gamecol!80!black] (a) {active instance\\token in\\game place $p$};
  \node[pbox,fill=gamecol!15,draw=gamecol!80!black,right=of a] (b) {strategic play\\$a\sim\sigma_p$,\\$\sigma_p\in\mathcal{A}_p(b)$};
  \node[pbox,fill=gamecol!15,draw=gamecol!80!black,right=of b] (c) {outcome\\[2pt]$\omega\sim\phi_p$};
  \node[pbox,fill=rescol!12,draw=rescol!80!black,right=of c] (d) {output tokens\\$\Out_p(b,a,\omega)$};
  \node[pbox,fill=rescol!12,draw=rescol!80!black,right=of d] (e) {guarded routing\\transitions\\$g_t(M,b)=1$};
  \node[pbox,fill=gamecol!15,draw=gamecol!80!black,right=of e] (f) {successor\\game places\\(one or many)};
  \foreach \u/\v in {a/b,b/c,c/d,d/e,e/f} \draw[arc] (\u) -- (\v);
  \draw[decorate,decoration={brace,mirror,amplitude=5pt},thick,inkcol!60]
    ($(a.south west)+(0,-0.15)$) -- ($(c.south east)+(0,-0.15)$)
    node[midway,below=6pt,layer] {game layer + solution layer: \emph{who chooses what, and why}};
  \draw[decorate,decoration={brace,mirror,amplitude=5pt},thick,inkcol!60]
    ($(d.south west)+(0,-0.15)$) -- ($(f.south east)+(0,-0.15)$)
    node[midway,below=6pt,layer] {structural net: \emph{what can happen next}};
  \draw[darc,rounded corners=6pt] (f.north) -- ++(0,0.5) -| node[pos=0.25,above,lbl]{the cycle repeats at every game place} (a.north);
\end{tikzpicture}
\caption{The basic chain of an OCGPN (symbols are defined in \cref{sec:model,sec:theory}). As a first intuition, the game decides and the net executes. In full, guards and the scheduler also shape which transitions are realized. Orange marks game-related objects and blue marks ordinary net objects. Every figure uses the same colors.}
\label{fig:pipeline}
\end{figure}

\begin{designrule}[Role of tokens]
A token is a typed \emph{object}: a task, entitlement, resource, observation, commitment, or certificate. A token need not be a strategic player. Players are defined by the game attached to a location. This differs from the established Petri-game tradition~\cite{finkbeiner2017}, where tokens usually represent distributed players and the main question is the synthesis of winning strategies.
\end{designrule}

The model separates three layers. The \emph{structural net} fixes what is possible, the \emph{local games} say what is at stake, and the \emph{solution layer} says what is played. Each layer carries one matrix object, and the analysis is built on these:
\begin{itemize}[nosep]
  \item the \emph{incidence matrix} $C$ of the unfolded net governs reachability, safety, and progress;
  \item the \emph{payoff} and \emph{outcome-probability matrices} $(R^p,\Phi^p_\omega)$ of each game place, aggregated into routing probabilities $\Psi^p_q$, govern incentives;
  \item the \emph{routing matrix} $P^\sigma$ between game places governs values and their computation.
\end{itemize}
These objects are previewed here informally; \cref{sec:model,sec:theory} define them, and \cref{tab:notation} collects the notation. The separation is deliberate. Guards act only on $C$, incentive instruments only on the local matrices, and equilibrium selection only on $P^\sigma$. This explains, for example, why deception can lower an attacker's value but can never certify that an outcome is unreachable, whereas a guard can.

\paragraph{Two levels of theory.} The theory is developed at two levels of generality. For \emph{general} OCGPNs, which allow arbitrary concurrency, we give an operational semantics on markings and a structural theory based on the incidence matrix: invariants, Farkas certificates of unreachability, and conditions for progress. For \emph{sequential} OCGPNs, in which one game is active at a time, the net induces a stochastic game on its game places. At this level, local and global games are precisely related. A stationary profile is a global equilibrium exactly when, at every place, it is an equilibrium of the local game \emph{augmented} by the continuation values that the profile itself generates. Equilibrium values can then be computed place by place, by solving the strongly connected components of the place graph in reverse topological order. The cost of this computation is additive, not multiplicative, in the sizes of the local games. This additivity is what makes large-scale systems tractable. In the 39-segment enterprise study of~\cite{yang2025crossechelon}, a decomposed model of about $10^4$ tree nodes replaces a monolithic one of roughly $10^{25}$.

\paragraph{Contributions.} We regard the following as the contributions of the paper.
\begin{enumerate}[label=(\arabic*),nosep]
  \item \textbf{A model and its semantics.} The game-to-token interface, in which action- and outcome-dependent token production determines the availability, composition, and synchronization of later interactions, with a precise operational semantics, bindings, scheduler, and unfolding (\cref{sec:model,sec:semantics}).
  \item \textbf{Structural certificates for strategic systems.} Invariants and Farkas certificates on the incidence matrix that hold for every equilibrium and every admissible profile. They are paired with a design rule separating guards, which can certify safety, from incentive instruments, which cannot (\cref{sec:structural}).
  \item \textbf{Interface design results.} Successor pivotality, a dominance criterion through the interface, conditions under which myopic play is globally consistent, and monotone comparative statics. These results say when and how downstream games can change upstream behavior (\cref{sec:strategic,sec:composition}).
  \item \textbf{A composition calculus.} All composition operators are built from place fusion. A fusion lemma glues invariants and certificates across subnets. A compositional-equilibrium theorem summarizes a downstream subnet by its entry values, also in the general-sum case. Further results give loop reduction, fork--join success probabilities, independence of parallel parts, and the local cost of carrying artifacts (\cref{sec:composition}).
  \item \textbf{Two running examples.} A cooperative maintenance workflow and an adversarial intrusion campaign illustrate the theory in the same order in which it is developed. They yield substantive conclusions. Graded routing can reverse a myopic equilibrium, whereas binary routing cannot. A forgiving rework loop destroys cooperation unless rework is costly. A newly discovered attack scenario can be composed with an existing campaign without re-solving it (\cref{sec:re1,sec:re2}).
\end{enumerate}
Several supporting results specialize standard theory to the place-level game: local--global consistency is the one-shot deviation principle, the place-level Shapley operator is Shapley's theorem, and backward induction is classical. We state and prove them because the rest of the analysis is built on them, and we identify them as such.

\paragraph{Organization.} \Cref{sec:background} motivates the composition of games and reviews meta-games, games in games, compositional game theory, and the related literature. \Cref{sec:primer} recalls Petri nets. \Cref{sec:model} defines the model, and \cref{sec:theory} develops the theory: semantics (\cref{sec:semantics}), local and global games (\cref{sec:games}), structural analysis (\cref{sec:structural}), strategic analysis (\cref{sec:strategic}), values and computation (\cref{sec:values}), and composition (\cref{sec:composition}). \Cref{sec:re1,sec:re2} present the running examples, and \cref{tab:roadmap} shows where each part of the theory is illustrated. \Cref{sec:discussion} concludes with lessons, limitations, and future directions.

\begin{table}[ht]
\centering
\caption{Roadmap: each part of the theory and where it is illustrated.}
\label{tab:roadmap}
\small
\begin{tabular}{@{}p{0.28\textwidth}p{0.32\textwidth}p{0.32\textwidth}@{}}
\toprule
\textbf{Theory (\cref{sec:theory})} & \textbf{Example I: maintenance (\cref{sec:re1})} & \textbf{Example II: intrusion (\cref{sec:re2})}\\\midrule
Semantics, incidence matrix (\cref{sec:semantics}) & net and execution trace (\cref{fig:examplenet,fig:trace}) & campaign net (\cref{fig:campaignnet})\\
Local and global games (\cref{sec:games}) & myopic versus global equilibrium (\cref{prop:reversal}); cyclic global game (\cref{prop:variantB}) & global campaign game (\cref{prop:deception})\\
Structural analysis (\cref{sec:structural}) & safety guarantees S1--S3, rework livelock (\cref{fig:safety}) & lockdown certificate versus decoys\\
Strategic analysis (\cref{sec:strategic}) & successor pivotality, mixed equilibria (\cref{fig:decomposition,fig:matrices,fig:phase}) & deception game, lateral game\\
Values and computation (\cref{sec:values}) & backward induction, rework component & component-wise solve, routing matrix, absorption\\
Composition (\cref{sec:composition}) & fork--join diagnostics (\cref{fig:forkjoin}) & credentials, telemetry fork, shared capacity; a new phishing scenario (\cref{sec:re2-new})\\
\bottomrule
\end{tabular}
\end{table}

\section{Background and motivation}
\label{sec:background}

\subsection{Motivation for composition}

The introduction argued that large-scale strategic systems should be modeled by composing local games rather than by one monolithic game. Composition requires specifying each interaction as a local game and, separately, how the outcomes of local games are connected. Any compositional framework must then answer three questions, which organize the rest of the paper:
\begin{enumerate}[label=(Q\arabic*),nosep]
  \item \textbf{Interface.} What does one game pass to the next: a payoff, a state, a parameter, or a set of typed objects?
  \item \textbf{Incentives.} How do downstream games change the incentives of upstream players, and when can the interface be designed to make good behavior an equilibrium?
  \item \textbf{Analysis.} Which properties (safety, termination, values) of the composite can be established from properties of the parts?
\end{enumerate}

\subsection{Three views of composed games}

Three established ideas address composition from different angles (\cref{fig:paradigms}).

\begin{figure}[ht]
\centering
\begin{subfigure}[t]{0.31\textwidth}\centering
\begin{tikzpicture}[x=1cm,y=1cm,
  mc/.style={draw=inkcol!55,line width=0.6pt,fill=white,minimum width=1.6cm,minimum height=1.1cm,inner sep=0pt},
  mini/.style={draw=gamecol!80!black,line width=0.5pt,fill=gamecol!12,minimum width=0.32cm,minimum height=0.22cm,inner sep=0pt}]
  \foreach \i/\x in {1/0,2/1.6} \foreach \j/\y in {1/0,2/-1.1} {
    \node[mc] (c\i\j) at (\x,\y) {};
    \foreach \u in {-0.17,0.17} \foreach \v in {0.12,-0.12} \node[mini] at ($(c\i\j)+(\u,\v)$) {};
  }
  \node[lbl,anchor=east] at (-0.85,0) {$\pi_1$};
  \node[lbl,anchor=east] at (-0.85,-1.1) {$\pi_2$};
  \node[lbl,anchor=south] at (0,0.6) {$\pi'_1$};
  \node[lbl,anchor=south] at (1.6,0.6) {$\pi'_2$};
  \node[lbl,align=center,anchor=north] at (0.8,-1.75) {meta-payoff $=$ value of the\\base game under $(\pi_i,\pi'_j)$};
\end{tikzpicture}
\caption{Meta-game: a game whose strategies select play in base games.}
\end{subfigure}\hfill
\begin{subfigure}[t]{0.31\textwidth}\centering
\begin{tikzpicture}[x=1cm,y=1cm]
  \node[draw=rescol!80!black,line width=0.8pt,fill=rescol!6,rounded corners=3pt,minimum width=3.6cm,minimum height=2.6cm] (out) at (0,0) {};
  \node[lbl,anchor=north west,font=\footnotesize\itshape,text=rescol!70!black] at ($(out.north west)+(0.08,-0.05)$) {outer game $G_{\mathrm{out}}(v)$};
  \node[draw=gamecol!85!black,line width=0.8pt,fill=gamecol!14,rounded corners=3pt,minimum width=2.0cm,minimum height=0.9cm,font=\footnotesize] (in) at (0,-0.45) {inner game $G_{\mathrm{in}}(\theta)$};
  \draw[arc] ($(in.north west)+(0.35,0.55)$) -- node[left,lbl]{$\theta$} ($(in.north west)+(0.35,0)$);
  \draw[arc] ($(in.north east)+(-0.35,0)$) -- node[right,lbl]{$v$} ($(in.north east)+(-0.35,0.55)$);
  \node[lbl,align=center,anchor=north] at (0,-1.75) {outer decisions set $\theta$;\\inner value $v$ enters outer payoff};
\end{tikzpicture}
\caption{Games in games: an inner game nested in an outer one, coupled by parameters and values.}
\end{subfigure}\hfill
\begin{subfigure}[t]{0.33\textwidth}\centering
\begin{tikzpicture}[x=1cm,y=1cm]
  \node[gplace,minimum size=7mm] (g1) at (0,0) {}; \node[ctok=badcol] at (g1) {};
  \node[rtrans,minimum height=6mm] (r) at (0.9,0) {};
  \node[place,minimum size=6mm] (q) at (1.75,0) {};
  \coordinate (fan) at (2.3,0);
  \node[transition,minimum height=6mm] (t1) at (2.85,0.75) {};
  \node[transition,minimum height=6mm] (t2) at (2.85,-0.75) {};
  \node[gplace,minimum size=7mm] (g2) at (3.65,0.75) {};
  \node[gplace,minimum size=7mm] (g3) at (3.65,-0.75) {};
  \node[place,minimum size=6mm,tokens=1] (res) at (0.9,-1.0) {};
  \draw[arc] (g1)--(r); \draw[arc] (r)--(q); \draw[arc] (res)--(r.south);
  \draw[arc] (q.east)--(fan)|-(t1); \draw[arc] (q.east)--(fan)|-(t2);
  \draw[arc] (t1)--(g2); \draw[arc] (t2)--(g3);
  \node[lbl,align=center,anchor=north] at (1.8,-1.75) {games at places; outcomes are\\tokens that enable other games};
\end{tikzpicture}
\caption{OCGPN: games at places, connected by typed tokens, guards, and resources.}
\end{subfigure}
\caption{Three ways to compose games. Meta-games compose \emph{vertically}, by abstraction. Games in games compose \emph{by nesting}, across layers. The OCGPN composes \emph{by flow}: outcomes become tokens that enable, synchronize, or block other games.}
\label{fig:paradigms}
\end{figure}
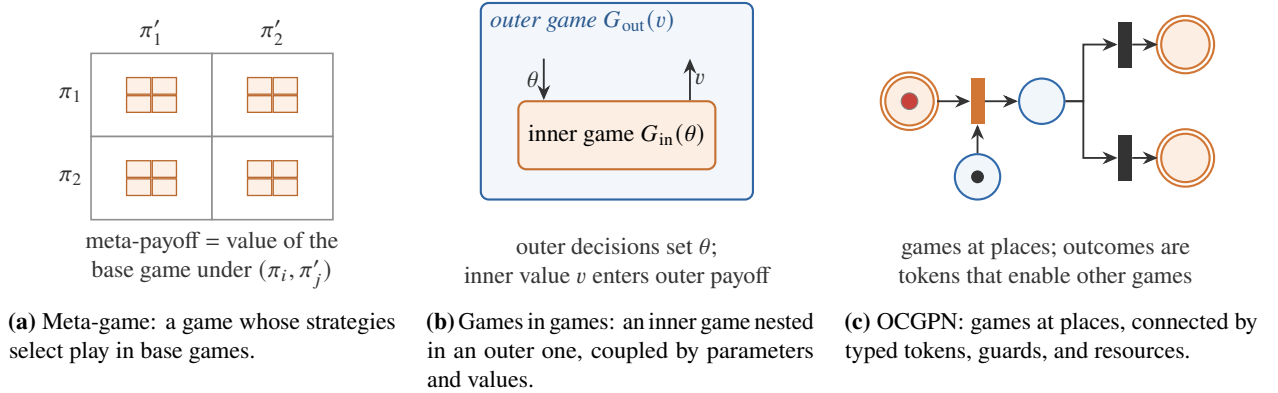

\paragraph{Meta-games.} Howard's theory of metagames~\cite{howard1971} lets players choose \emph{policies}, that is, conditional responses to other players' choices, rather than base actions. Equilibria of the resulting meta-game can explain cooperation that the base game does not support. The term is now also used for the empirical games of multi-agent learning. In empirical game-theoretic analysis~\cite{wellman2006} and policy-space response oracles~\cite{lanctot2017}, a meta-game is a normal-form game whose strategies are learned policies, and whose payoffs are estimated by simulating the base game. Hypergames~\cite{bennett1980,kovach2015} compose games in yet another sense: each player may perceive a different game, and the hypergame collects these perceptions. Meta-games compose \emph{vertically}. One game abstracts the play of another, but there is no notion of one game's outcome \emph{enabling} a different game.

\paragraph{Games in games.} The games-in-games principle~\cite{zhu2015gig} composes games \emph{across layers}. In cyber-physical resilience, a cyber-layer game between attacker and defender determines the parameters (for example, the mode or failure state) of a physical-layer control game. The value of the physical game feeds back into the payoffs of the cyber game. Equilibria are computed jointly, layer by layer. Multi-phase and multi-stage models of advanced persistent threats~\cite{zhu2018apt,huang2020apt} chain games across the phases of a kill chain, so that the outcome of one phase sets the initial state of the next. The multi-resolution games of~\cite{yang2025multires,yang2025general,yang2025crossechelon} nest technique-level extensive-form games inside a segment-level Markov game, and they move between resolutions by grounding and lifting. These works establish the central idea that the paper builds on: \emph{the equilibrium value of a downstream game enters the payoffs of the upstream game.} What they do not provide is a general structural layer. Nesting is fixed by design, concurrency and shared resources are outside the model, and there is no reachability analysis of the composite.

\paragraph{Compositional game theory.} Open games~\cite{ghani2018,hedges2016} give a categorical calculus in which games are morphisms that can be composed sequentially and in parallel, with string diagrams as syntax. Composition is algebraic and exact, and equilibria of composites are characterized from those of the parts. The calculus is about information and utility flow between decisions. It has no native notion of resources, synchronization over shared places, or reachability, which are the strengths of Petri nets.

\paragraph{Positioning.} The OCGPN takes the coupling principle of games in games, that downstream values enter upstream payoffs, and places it on the structural substrate of colored Petri nets. The interface (Q1) is a multiset of typed tokens. Incentives (Q2) are analyzed through outcome-indicator matrices and pivotality (\cref{sec:strategic}). Analysis (Q3) uses the incidence and routing matrices, with a calculus for each composition operator (\cref{sec:composition}). Meta-games and open games remain complementary. A meta-game can be solved at any game place, and open-game composition describes the information flow inside one.

\subsection{Related literature}

\paragraph{Petri nets and their analysis.} Petri nets~\cite{petri1962} and their colored extension~\cite{jensen1992,jensen2009} are standard models of concurrent systems. Structural analysis by incidence matrices, the state equation, and P- and T-invariants is classical~\cite{murata1989}. Generalized stochastic Petri nets add timing and probabilistic choice and reduce performance questions to Markov chains~\cite{marsan1995}. Workflow nets apply the formalism to business processes, with soundness as the central correctness notion~\cite{vanderaalst1998}. Supervisory control of discrete-event systems~\cite{ramadge1987}, including Petri-net supervisors built from place invariants~\cite{moody1998}, synthesizes guards that enforce specifications. Our guards play exactly this role, and \cref{prop:farkas} is in the same tradition of invariant-based certificates. None of these models has strategic players.

\paragraph{Games on Petri nets and games on graphs.} Petri games~\cite{finkbeiner2017} model distributed synthesis: tokens are players, partitioned into system and environment, and the question is whether the system has a winning strategy. Games on graphs and alternating-time temporal logic~\cite{alur2002} similarly address winning conditions in verification. These frameworks are qualitative (win or lose) and identify players with tokens or vertices. The OCGPN is quantitative (payoffs and values), and its players live in the games attached to places, not in the tokens.

\paragraph{Dynamic and structured games.} Stochastic games~\cite{shapley1953,filar1997}, recursive games~\cite{everett1957}, and dynamic noncooperative games~\cite{basar1999,fudenberg1991} provide the solution concepts and fixed-point results we use. \Cref{thm:contraction} is Shapley's theorem on the place graph. Multi-agent influence diagrams~\cite{koller2003} and graphical games~\cite{kearns2001} exploit the structure of a single game to compute equilibria efficiently. The SCC decomposition of \cref{thm:scc} exploits the structure \emph{between} games in an analogous way. Markov decision processes~\cite{puterman1994} provide the policy-evaluation identity \eqref{eq:neumann}.

\paragraph{Security games, attack models, and deception.} Game-theoretic security is surveyed in~\cite{alpcan2010,manshaei2013,tambe2011}, and defensive deception in~\cite{pawlick2019}. Interdependent security~\cite{kunreuther2003} shows how one party's protection changes another's incentives, a coupling the OCGPN expresses through shared places. Attack graphs~\cite{sheyner2002} and attack nets~\cite{mcdermott2000} model the structure of multi-step attacks without strategic defenders. Running example~II combines these strands. It takes the attack structure of attack graphs, the technique-level engagements of multi-resolution games~\cite{yang2025crossechelon}, and the deception mechanics of~\cite{pawlick2019}, and it places them in one net where reachability and values can be analyzed together.

\begin{table}[ht]
\centering
\caption{Positioning of the OCGPN. \checkmark: supported natively; $\circ$: partially or with extra machinery; --: not represented explicitly. ``Concurrency'' means explicit, partial-order concurrency; models over global states can encode it implicitly.}
\label{tab:positioning}
\footnotesize
\setlength{\tabcolsep}{3.5pt}
\begin{tabular}{@{}>{\raggedright\arraybackslash}p{4.6cm}cccccc@{}}
\toprule
 & strategic & concurrency, & typed & composition & structural & decomposed\\
 & players & synchronization & resources & operators & certificates & values\\\midrule
Stochastic / extensive-form games~\cite{shapley1953,fudenberg1991} & \checkmark & -- & -- & -- & -- & $\circ$\\
MAIDs, graphical games~\cite{koller2003,kearns2001} & \checkmark & -- & -- & $\circ$ & -- & $\circ$\\
Meta-games, EGTA~\cite{howard1971,wellman2006} & \checkmark & -- & -- & $\circ$ & -- & \checkmark\\
Games in games~\cite{zhu2015gig,zhu2018apt} & \checkmark & -- & -- & \checkmark & -- & \checkmark\\
Open games~\cite{ghani2018} & \checkmark & -- & -- & \checkmark & -- & $\circ$\\
Petri games~\cite{finkbeiner2017} & $\circ$ & \checkmark & $\circ$ & -- & \checkmark & --\\
Colored, stochastic, and workflow nets~\cite{jensen2009,marsan1995,vanderaalst1998} & -- & \checkmark & \checkmark & \checkmark & \checkmark & $\circ$\\
\textbf{OCGPN (this paper)} & \checkmark & \checkmark & \checkmark & \checkmark & \checkmark & \checkmark$^{\dagger}$\\
\bottomrule
\end{tabular}

\smallskip
{\footnotesize $^{\dagger}$Under sequential activation (\cref{ass:sequential}); over markings otherwise.}
\end{table}

\section{Primer: Petri nets and colored tokens}
\label{sec:primer}

For readers new to Petri nets, we recall the basic ``token game''. A Petri net is a bipartite graph of \emph{places} (circles) and \emph{transitions} (bars). A \emph{marking} puts tokens in places. A transition is \emph{enabled} if each input place holds enough tokens. \emph{Firing} it consumes those tokens and produces tokens in its output places (\cref{fig:firing}). Colored Petri nets~\cite{jensen2009} attach data, called a \emph{color}, to each token. A transition may also carry a \emph{guard}, a predicate on the colors it consumes.

\begin{figure}[ht]
\centering
\begin{subfigure}[b]{0.45\textwidth}\centering
\begin{tikzpicture}
  \node[place,tokens=2,label=above:{\footnotesize $p_1$}] (p1) at (0,0.8) {};
  \node[place,tokens=1,label=below:{\footnotesize $p_2$}] (p2) at (0,-0.8) {};
  \node[transition,label={[lbl]above right:$t$}] (t) at (1.6,0) {};
  \node[place,label=above:{\footnotesize $p_3$}] (p3) at (3.2,0) {};
  \draw[arc] (p1) -| (t); \draw[arc] (p2) -| (t); \draw[arc] (t) -- (p3);
  \node[guard,anchor=north,text=okcol!70!black] at (3.2,-0.6) {\footnotesize enabled};
\end{tikzpicture}
\caption{Before: $t$ is enabled.}
\end{subfigure}
\hfill
\begin{subfigure}[b]{0.45\textwidth}\centering
\begin{tikzpicture}
  \node[place,tokens=1,label=above:{\footnotesize $p_1$}] (p1) at (0,0.8) {};
  \node[place,label=below:{\footnotesize $p_2$}] (p2) at (0,-0.8) {};
  \node[transition,label={[lbl]above right:$t$}] (t) at (1.6,0) {};
  \node[place,tokens=1,label=above:{\footnotesize $p_3$}] (p3) at (3.2,0) {};
  \draw[arc] (p1) -| (t); \draw[arc] (p2) -| (t); \draw[arc] (t) -- (p3);
  \node[guard,anchor=north,text=badcol!80!black] at (3.2,-0.6) {\footnotesize disabled ($p_2$ empty)};
\end{tikzpicture}
\caption{After firing $t$ once.}
\end{subfigure}

\vspace{4mm}
\begin{subfigure}[b]{0.95\textwidth}\centering
\begin{tikzpicture}
  \node[place,label=above:{\footnotesize work orders}] (w) at (0,0) {};
  \node[ctok=okcol] at ($(w)+(-0.13,0.1)$) {};
  \node[ctok=badcol] at ($(w)+(0.13,-0.1)$) {};
  \node[transition,label=above:{\footnotesize $t_{\mathrm{exec}}$}] (t) at (2.6,0) {};
  \node[guard,below=1mm of t,align=center] {$[\mathrm{Safe}(\zeta)\wedge\mathrm{cert}(\zeta)]$};
  \node[place,label=above:{\footnotesize executed}] (e) at (5.2,0) {};
  \draw[arc] (w) -- node[above,lbl]{$\zeta$} (t); \draw[arc] (t) -- node[above,lbl]{$\zeta$} (e);
  \node[anchor=west,font=\footnotesize,align=left] at (6.1,0.15) {%
    \tikz\node[ctok=okcol]{}; $\zeta=\tok{\mathrm{id}{=}7,\ \mathrm{safe},\ \mathrm{cert}}$ can pass\\[2pt]
    \tikz\node[ctok=badcol]{}; $\zeta=\tok{\mathrm{id}{=}9,\ \mathrm{unsafe}}$ is blocked};
\end{tikzpicture}
\caption{Colored tokens carry data. The guard on $t_{\mathrm{exec}}$ decides which bindings may fire.}
\end{subfigure}
\caption{The Petri-net token game. Places are circles, transitions are bars, and tokens are dots. In a colored net each dot carries a typed record, and guards filter records.}
\label{fig:firing}
\end{figure}
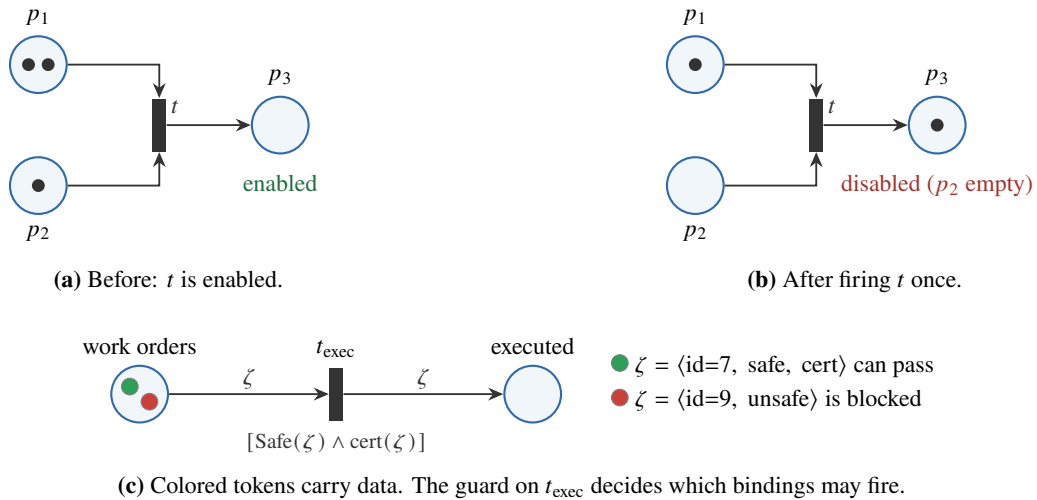

\section{The model}
\label{sec:model}

\subsection{Ingredients}

Let $P$ be a finite set of places and $T$ a finite set of transitions. Arcs are given by a relation $F\subseteq(P\times T)\cup(T\times P)$ with positive integer multiplicities. For a transition $t$, its \emph{input places} are those $p$ with $(p,t)\in F$, and its \emph{output places} those $p$ with $(t,p)\in F$; the multiset of input places, counted with multiplicity, is the \emph{preset} of $t$. A \emph{read arc} between a place and a transition tests for a token without consuming it; it abbreviates an input arc and an output arc that move the same token. Each transition $t$ carries a \emph{guard} $g_t$, a predicate on the marking and the binding defined below. Partition
\[
  P = P_G \,\cup\, P_R,\qquad P_G\cap P_R=\emptyset,
\]
where $P_G$ contains \emph{game places} and $P_R$ contains ordinary resource and control places. Disjointness is only a notational convenience: resource constraints on games are still encoded through transition arcs and token bindings. Each place $p$ has a \emph{color domain} $\mathrm{Col}_p$, the set of values (typed records) that a token in $p$ may carry.

\paragraph{Multisets and markings.} For a set $X$, a \emph{finite multiset} over $X$ is a function $m:X\to\mathbb{N}$ with finite support $\{x:m(x)>0\}$; $m(x)$ is the number of copies of $x$. We write $\Mult(X)$ for the set of finite multisets over $X$. Multisets are added, compared, and subtracted pointwise: $(m+m')(x)=m(x)+m'(x)$; $m\le m'$ if $m(x)\le m'(x)$ for all $x$; and $m'-m$ is defined when $m\le m'$. A \emph{marking} assigns to each place a finite multiset of tokens of its color domain:
\begin{equation}\label{eq:marking}
  M=(M_p)_{p\in P}\in\M := \prod_{p\in P}\Mult(\mathrm{Col}_p).
\end{equation}
Sums, differences, and the order $\le$ extend to markings place by place. These are the operations used in the firing rules below. The number of tokens in place $p$ is $|M_p|=\sum_{c}M_p(c)$, and we write $M(p)$ for $|M_p|$ when colors are irrelevant. For a finite set $X$, $\Delta(X)$ denotes the set of probability distributions on $X$, and $\ind\{\cdot\}$ is the indicator function.

\paragraph{Bindings.} A \emph{binding} of a transition $t$ selects the tokens it consumes: one token of the appropriate color from each input place, counted with arc multiplicity. For a resolution transition, which is defined below, we write a binding as $b=(b_0,b_R)$. Here $b_0$ is the \emph{active-instance token} in the game place, and $b_R$ is the tuple of resource tokens consumed from resource places (for example, a budget unit or an analyst). $B_t(M)$ denotes the set of bindings of $t$ available at $M$, and $\Pre(t,b)\in\M$ the multiset that the binding consumes.

\begin{definition}[Local game]\label{def:localgame}
For each game place $p\in P_G$ and binding $b$, the \emph{local game} is
\begin{equation}
  \G_p(b) = \Big(\mathcal{I}_p(b),\ \{A_{i,p}(b)\}_{i\in \mathcal{I}_p(b)},\ \{\mathsf{Inf}_{i,p}(b)\}_{i\in \mathcal{I}_p(b)},\ \{u_{i,p}\}_{i\in \mathcal{I}_p(b)},\ \Omega_p,\ \phi_p\Big).
\end{equation}
Here $\mathcal{I}_p(b)$ is the set of participants; $A_{i,p}(b)$ are the feasible actions of player $i$, and $a\in A_p(b):=\prod_{i\in\mathcal{I}_p(b)}A_{i,p}(b)$ is an action profile; $\mathsf{Inf}_{i,p}(b)$ is the information available to player $i$, including private token fields; $u_{i,p}(a,b,\omega)$ are local utilities; $\Omega_p$ is an outcome domain; and $\phi_p(\omega\mid a,b)$ is a probability kernel over outcomes.
\end{definition}

A deterministic game corresponds to a point-mass kernel, $\phi_p(\cdot\mid a,b)=\delta_{o_p(a,b)}$, where $\delta_\omega$ is the distribution that puts probability one on $\omega$ and $o_p$ is a deterministic outcome map. The kernel form also covers chance moves and noisy outcomes. \Cref{def:localgame} is deliberately general. The equilibrium theory of \cref{sec:games,sec:strategic,sec:values} specializes it to normal form with a publicly known binding (\cref{ass:normalform}), and \cref{def:nflocal} restates it in that form.

\paragraph{Resolution and the game-to-token interface.} A \emph{resolution} of $\G_p$ is a quadruple $(p,b,a,\omega)$ that records one completed instance: its binding, its action profile, and its realized outcome. Its token production is a finite multiset
\begin{equation}\label{eq:out}
  \Out_p(b,a,\omega)\in\M ,
\end{equation}
which may depend on the action profile as well as the outcome. The map $\Out_p$ is the \emph{game-to-token interface}. Its tokens are \emph{outcome tokens}: their colors record the outcome $\omega$ (and, where relevant, data such as an instance identifier). They can be placed in routing places or directly in successor game places. We model each resolution as a dedicated \emph{resolution transition} $r_p$. Firing $r_p$ consumes the active-instance token and the bound resources, $\Pre(r_p,b)$, and produces $\Out_p(b,a,\omega)$. Ordinary routing transitions then move these outputs on to successor games. \Cref{fig:anatomy} shows the full anatomy of one game place.

\begin{figure}[ht]
\centering
\begin{tikzpicture}[x=1cm,y=1cm]
  \fill[gamecol!5,rounded corners=4pt] (-1.1,-2.75) rectangle (3.75,4.15);
  \node[glbl,anchor=north west] at (-1.0,4.1) {strategic interaction};
  \fill[rescol!4,rounded corners=4pt] (4.05,-2.75) rectangle (11.3,4.15);
  \node[lbl,font=\footnotesize\itshape,text=rescol!70!black,anchor=north east] at (11.2,4.1) {outcome-guarded routing};
  \node[gplace,label={[glbl]below:game place\\[-1pt]$p\in P_G$}] (p) at (0,0) {};
  \node[ctok=gamecol] at (p) {};
  \node[lbl,anchor=south] at ($(p.north)+(0,0.05)$) {$\tok{\mathrm{id},b}$};
  \node[rtrans,label={[lbl]above left:$r_p$}] (r) at (2.2,0) {};
  \node[place,tokens=2,label={[lbl]left:resource}] (bud) at (2.2,-1.9) {};
  \node[place,label={[lbl,align=center]below:result place $q$\\$\tok{\mathrm{id},\omega,\mathrm{data}}$}] (q) at (5.3,0) {};
  \node[ctok=okcol] at (q) {};
  \coordinate (fan) at (6.6,0);
  \node[transition,label={[guard]above:$[\omega\in O_{p,t_1}]$}] (t1) at (7.7,1.3) {};
  \node[transition,label={[guard]above:$[\omega\in O_{p,t_2}]$}] (t2) at (7.7,-1.3) {};
  \node[lbl,anchor=north] at ($(t1.south)+(0,-0.05)$) {$t_1$};
  \node[lbl,anchor=north] at ($(t2.south)+(0,-0.05)$) {$t_2$};
  \node[gplace,label={[glbl]right:$\G_{q_1}$}] (g1) at (9.3,1.3) {};
  \node[gplace,label={[glbl]right:$\G_{q_2}$}] (g2) at (9.3,-1.3) {};
  \node[cbox,align=center] (cb) at (1.1,2.3) {%
    \textbf{local game} $\G_p(b)$\\[2pt]
    \begin{tabular}{c|cc}
      & $\act{H}$ & $\act{L}$\\\hline
      $\act{H}$ & $2,2$ & $-1,3$\\
      $\act{L}$ & $3,-1$ & $0,0$
    \end{tabular}\\[3pt]
    $a\sim\sigma_p\in\mathcal{A}_p(b)$,\ \ $\omega\sim\phi_p(\cdot\mid a,b)$};
  \draw[arc] (p) -- (r);
  \draw[arc] (bud) -- node[right,lbl]{cost} (r.south);
  \draw[arc] (r) -- node[above,lbl]{$\Out_p(b,a,\omega)$} (q);
  \draw[arc] (q.east) -- (fan) |- (t1);
  \draw[arc] (q.east) -- (fan) |- (t2);
  \draw[arc] (t1) -- (g1); \draw[arc] (t2) -- (g2);
  \draw[darc,draw=gamecol!80!black] (cb.south -| r) -- (r.north);
\end{tikzpicture}
\caption{Anatomy of a game place. An active-instance token waits in $p$. When the players' profile $a$ and the outcome $\omega$ are realized, the resolution transition $r_p$ (orange bar) consumes the instance token and the bound resources. It then emits the outcome tokens $\Out_p(b,a,\omega)$. Ordinary transitions $t_1,t_2$ with guards of the form \eqref{eq:guard} route these tokens to successor games.}
\label{fig:anatomy}
\end{figure}
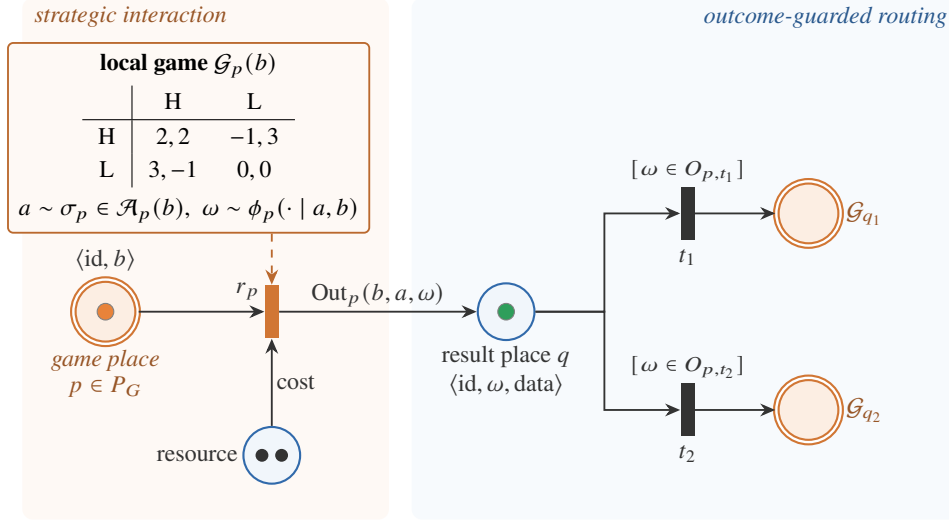

\subsection{General definition}

\begin{definition}[OCGPN]\label{def:ocgpn}
An \emph{outcome-driven compositional game Petri net} is a tuple
\begin{equation}
  \N = \big(P,\,T,\,F,\,\{g_t\},\,\mathrm{Col},\,M_0,\,\{\G_p\}_{p\in P_G},\,\{\mathcal{A}_p\}_{p\in P_G},\,\Pre,\,\Post,\,\kappa,\,\Obs,\,\Spec\big),
\end{equation}
where
\begin{itemize}[nosep,leftmargin=2em]
  \item $F$ is the arc relation with multiplicities, $\{g_t\}_{t\in T}$ are the guards, and $\mathrm{Col}=\{\mathrm{Col}_p\}_{p\in P}$ lists the color domains;
  \item $M_0$ is the initial marking;
  \item $\mathcal{A}_p(b)\subseteq\prod_{i\in\mathcal{I}_p(b)}\Delta(A_{i,p}(b))$ is a nonempty \emph{admissibility correspondence}: the local profiles that the modeler permits at $p$ (all profiles, or those satisfying policy or regulatory constraints);
  \item $\Post(t,b)$ is the output multiset of an ordinary transition, and $\Post(r_p,b,a,\omega):=\Out_p(b,a,\omega)$ that of a resolution transition;
  \item $\kappa$ is a scheduler (defined in \cref{sec:semantics});
  \item $\Obs$ specifies which parts of the marking and history each player observes;
  \item $\Spec$ is a set of global specifications: safety, progress, cost, or terminal objectives.
\end{itemize}
\end{definition}

The admissibility correspondence $\mathcal{A}_p$ is a modeling input. It is distinct from the \emph{equilibrium} correspondences derived in \cref{sec:games}, which depend on the continuation values of the whole net. Stage-game Nash play is one admissible choice, but \cref{sec:games} shows that it is generally not globally consistent.

\begin{figure}[ht]
\centering
\begin{tikzpicture}[
  plane/.style={draw=#1!70!black,fill=#1!7,thick,rounded corners=2pt,minimum width=10.5cm,minimum height=2.0cm},
  ptitle/.style={font=\small\bfseries,anchor=west}]
  \node[plane=okcol] (L3) at (0,4.8) {};
  \node[plane=gamecol] (L2) at (0,2.4) {};
  \node[plane=rescol] (L1) at (0,0) {};
  \node[ptitle,text=okcol!50!black] at (6.1,4.8) {(iii) solution layer};
  \node[font=\footnotesize,anchor=west,align=left] at (6.1,4.3) {$\mathcal{A}_p(b)$, $\mathcal{E}_p(V)$,\\ scheduler $\kappa$, selection};
  \node[ptitle,text=gamecol!60!black] at (6.1,2.4) {(ii) game layer};
  \node[font=\footnotesize,anchor=west,align=left] at (6.1,1.9) {$\G_p(b)$: players, actions,\\ information, payoffs, $\phi_p$};
  \node[ptitle,text=rescol!60!black] at (6.1,0) {(i) structural net};
  \node[font=\footnotesize,anchor=west,align=left] at (6.1,-0.5) {$(P,T,F,\mathrm{Col},M_0)$, guards,\\ $\Pre$, $\Post$};
  \node[place,minimum size=5mm] (a1) at (-3.6,0) {};
  \node[transition,minimum height=5mm] (a2) at (-2.6,0) {};
  \node[gplace,minimum size=6mm] (a3) at (-1.5,0) {};
  \node[rtrans,minimum height=5mm] (a4) at (-0.4,0) {};
  \node[place,minimum size=5mm] (a5) at (0.7,0) {};
  \node[transition,minimum height=5mm] (a6) at (1.7,0.45) {};
  \node[transition,minimum height=5mm] (a7) at (1.7,-0.45) {};
  \node[gplace,minimum size=6mm] (a8) at (2.8,0.45) {};
  \node[gplace,minimum size=6mm] (a9) at (2.8,-0.45) {};
  \foreach \u/\v in {a1/a2,a2/a3,a3/a4,a4/a5,a6/a8,a7/a9} \draw[arc,thin] (\u)--(\v);
  \draw[arc,thin] (a5) |- (a6); \draw[arc,thin] (a5) |- (a7);
  \node[cbox,font=\scriptsize] (gm) at (-1.0,2.4) {%
    \begin{tabular}{c|cc} & $\act{H}$ & $\act{L}$\\\hline $\act{H}$ & $2,2$ & $-1,3$\\ $\act{L}$ & $3,-1$ & $0,0$\end{tabular}};
  \node[draw=okcol!70!black,fill=white,rounded corners,font=\scriptsize,align=center] (sl) at (-1.0,4.8) {$\sigma_p\in\mathcal{E}_p(V)$\\ e.g.\ Nash / Stackelberg / BNE};
  \draw[darc,draw=gamecol!80!black] (a3.north) -- node[pos=0.25,left,lbl]{attach} (a3.north |- gm.south);
  \draw[darc,draw=okcol!70!black] (gm.north) -- node[pos=0.82,right,lbl]{solve} (sl.south);
  \draw[darc,draw=okcol!70!black] (sl.east) -| node[pos=0.6,right,lbl]{select $a,\omega$} (3.7,0.85) -| (a4.north);
\end{tikzpicture}
\caption{Three layers of an OCGPN. The structural net fixes \emph{what is possible}, the game layer says \emph{what is at stake}, and the solution layer says \emph{what is chosen}. \Cref{sec:theory} attaches one matrix object to each layer: the incidence matrix $C$, the local matrices $(R^p,\Phi^p_\omega)$, and the routing matrix $P^\sigma$. The equilibrium correspondence $\mathcal{E}_p(V)$ is defined in \cref{def:equilibria}.}
\label{fig:layers}
\end{figure}
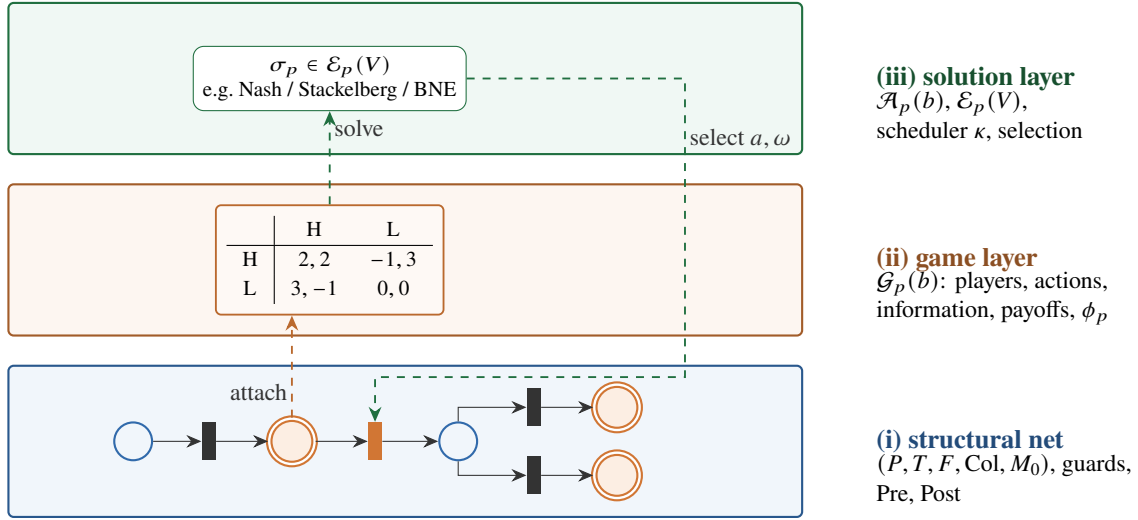

\paragraph{Well-formedness.} A well-formed OCGPN satisfies four conditions.
\begin{enumerate}[label=(W\arabic*),nosep]
  \item Every firing consumes only tokens that are actually present.
  \item Produced colors match the types of their destination places.
  \item Every kernel $\phi_p(\cdot\mid a,b)$ is a probability distribution.
  \item \emph{Exclusivity.} Each designated exclusive resource place $q$ carries a capacity $\mathrm{cap}(q)\in\mathbb{N}$. The weight vector $y$ that is $1$ on $q$ and on the places that hold the resource while it is in use, and $0$ elsewhere, is a P-invariant in the sense of \cref{prop:invariants}, with $y^\top M_0\le\mathrm{cap}(q)$ (markings viewed as vectors, \cref{sec:semantics}). Then no firing sequence duplicates the resource.
\end{enumerate}
If game instances may overlap, each instance's token color includes a unique instance identifier.

\section{General theory: semantics, analysis, and computation}
\label{sec:theory}

This section develops the theory in six steps, at two levels of generality.
\begin{itemize}[nosep]
  \item \textbf{General OCGPNs} (\cref{sec:semantics,sec:structural}, and the structural parts of \cref{sec:composition}): operational semantics, typed tokens, guards, incidence matrices, invariants, safety, and reachability. These results allow arbitrary concurrency.
  \item \textbf{Sequential OCGPNs} (\cref{sec:games,sec:strategic,sec:values}, and the value parts of \cref{sec:composition}): the class in which at most one game is active at a time (\cref{ass:sequential}). Here the net induces a stochastic game on its game places, and we obtain successor maps, augmented local games, Markov perfect equilibria, Shapley operators, and a decomposition into strongly connected components (SCCs).
\end{itemize}
Concurrent strategic interaction is covered by the general level and by the marking-level game of \cref{rem:beyond}. The place-level fixed-point results do not apply to it without state augmentation. The three matrix objects are separated deliberately. Guards act only on the incidence matrix, incentive instruments only on the local matrices, and equilibrium selection only on the routing matrix. \Cref{tab:notation} collects the notation. The running examples of \cref{sec:re1,sec:re2} illustrate the main results, in most cases through special instances.

\begin{table}[ht]
\centering
\caption{Principal notation.}
\label{tab:notation}
\footnotesize
\begin{tabular}{@{}>{\raggedright\arraybackslash}p{2.6cm}>{\raggedright\arraybackslash}p{5.2cm}|>{\raggedright\arraybackslash}p{2.4cm}>{\raggedright\arraybackslash}p{4.1cm}@{}}
\toprule
Symbol & Meaning & Symbol & Meaning\\\midrule
$P,T$; $P_G,P_R$ & places, transitions; game and resource places & $\mathcal{I}_p$, $\mathcal{I}$ & players at $p$; all players\\
$\mathrm{Col}_p$ & color domain of place $p$ & $A_i^p$, $A^p$ & actions of $i$ at $p$; profiles\\
$b=(b_0,b_R)$ & binding: instance token, resource tokens & $R_i^p$ & stage payoff of $i$ at $p$\\
$r_p$ & resolution transition of $p$ & $\Phi^p_\omega$ & outcome-probability matrix\\
$\Out_p(b,a,\omega)$ & tokens produced on resolution & $f$, $\Psi^p_q$ & successor map; routing probability to $q$\\
$\kappa$ & scheduler of enabled events & $Z$, $v_i$ & terminal places; terminal values\\
$\widetilde P,\widetilde T$ & unfolded places and transitions & $\mathcal{D}$ & place graph on $P_G$\\
$C=\Post-\Pre$ & incidence matrix of the unfolded net & $\overline V_i$ & values extended to $Z$\\
$n$ & firing-count vector & $Q_i^p(V)$ & augmented local payoff\\
$y$ & P-invariant or Farkas vector & $P^\sigma$, $\Xi^\sigma$, $\mathbf{r}_i^\sigma$ & routing, exit, reward\\
$\mathcal{A}_p$ & admissibility correspondence & $\mathcal{B}$ & Shapley (Bellman) operator\\
$\Spec$ & global specifications & $\mathcal{K}_1,\dots,\mathcal{K}_K$ & strongly connected components\\
\bottomrule
\end{tabular}
\end{table}

\subsection{Operational semantics and the incidence matrix}
\label{sec:semantics}

\paragraph{Enabling and firing.} A binding $b\in B_t(M)$ \emph{enables} an ordinary transition $t$ at marking $M$ if the input tokens are available and the guard $g_t$ holds:
\begin{equation}\label{eq:enable}
  M \ \ge\ \Pre(t,b),\qquad g_t(M,b)=1 .
\end{equation}
Firing gives $M'=M-\Pre(t,b)+\Post(t,b)$. A resolution transition $r_p$ is enabled by a binding $b=(b_0,b_R)$ if $b_0$ is an active-instance token in $p$, the resources $b_R$ are present, and its guard holds. Firing it proceeds in three steps: the players choose a profile $a\sim\sigma_p$ with $\sigma_p\in\mathcal{A}_p(b)$; an outcome $\omega\sim\phi_p(\cdot\mid a,b)$ is drawn; and the marking is updated by
\begin{equation}\label{eq:fire}
  M' = M - \Pre(r_p,b) + \Post(r_p,b,a,\omega),\qquad \Post(r_p,b,a,\omega)=\Out_p(b,a,\omega).
\end{equation}
\Cref{eq:fire} is a \emph{multiset} equation, not a subtraction of scalar counts. A resolution may add several tokens, including tokens that start parallel game instances.

\paragraph{Outcome-conditioned routing.} For a game place $p$, let $T_p^{\mathrm{route}}$ be the routing transitions whose input places receive the outcome tokens of $r_p$. For such a transition, the binding $b$ selects an outcome token, and we write $\omega(b)$ for the outcome recorded in its color. Guards can test for mutually exclusive outcomes, but exclusivity is not required:
\begin{equation}\label{eq:guard}
  g_t(M,b) = \ind\{\omega(b)\in O_{p,t}\},\qquad t\in T_p^{\mathrm{route}},
\end{equation}
where $O_{p,t}\subseteq\Omega_p$ is the set of outcomes that $t$ accepts. \Cref{fig:routing} shows the three resulting cases:
\begin{enumerate}[label=(\alph*),nosep]
  \item \textbf{Exclusive choice.} If the sets $O_{p,t}$ partition $\Omega_p$ and the routing transitions compete for one exclusive outcome token, exactly one branch fires.
  \item \textbf{Fork.} If several output tokens are generated, several routes may fire concurrently.
  \item \textbf{Conflict.} If guards overlap on the same exclusive token, the scheduler decides. Otherwise the model is nondeterministic.
\end{enumerate}

\begin{figure}[ht]
\centering
\begin{subfigure}[t]{0.32\textwidth}\centering
\begin{tikzpicture}[x=1cm,y=1cm]
  \node[rtrans,label={[lbl]below:$r_p$}] (r) at (0,0) {};
  \node[place] (q) at (1.2,0) {}; \node[ctok=okcol] at (q) {};
  \coordinate (fan) at (1.95,0);
  \node[transition,label={[guard]above:$[\omega{=}\mathrm{ok}]$}] (t1) at (2.7,0.95) {};
  \node[transition,label={[guard]below:$[\omega{=}\mathrm{fail}]$}] (t2) at (2.7,-0.95) {};
  \node[gplace] (g1) at (3.95,0.95) {}; \node[gplace] (g2) at (3.95,-0.95) {};
  \draw[arc] (r)--(q);
  \draw[arc] (q.east)--(fan)|-(t1); \draw[arc] (q.east)--(fan)|-(t2);
  \draw[arc,okcol!65!black,line width=1.1pt] (t1)--(g1); \draw[arc,inkcol!25] (t2)--(g2);
\end{tikzpicture}
\caption{Exclusive choice: one token, disjoint guards, so exactly one branch fires.}
\end{subfigure}\hfill
\begin{subfigure}[t]{0.32\textwidth}\centering
\begin{tikzpicture}[x=1cm,y=1cm]
  \node[rtrans,label={[lbl]below:$r_p$}] (r) at (0,0) {};
  \coordinate (fan) at (0.6,0);
  \node[place] (q1) at (1.45,0.95) {}; \node[ctok=okcol] at (q1) {};
  \node[place] (q2) at (1.45,-0.95) {}; \node[ctok=rescol] at (q2) {};
  \node[transition] (t1) at (2.7,0.95) {}; \node[transition] (t2) at (2.7,-0.95) {};
  \node[gplace] (g1) at (3.95,0.95) {}; \node[gplace] (g2) at (3.95,-0.95) {};
  \draw[arc] (r.east)--(fan)|-(q1); \draw[arc] (r.east)--(fan)|-(q2);
  \draw[arc] (q1)--(t1); \draw[arc] (q2)--(t2);
  \draw[arc,okcol!65!black,line width=1.1pt] (t1)--(g1); \draw[arc,okcol!65!black,line width=1.1pt] (t2)--(g2);
  \node[guard,anchor=south] at (1.45,1.4) {two tokens emitted};
\end{tikzpicture}
\caption{Fork: $\Out_p$ emits two tokens, so two successor games run concurrently.}
\end{subfigure}\hfill
\begin{subfigure}[t]{0.32\textwidth}\centering
\begin{tikzpicture}[x=1cm,y=1cm]
  \node[rtrans,label={[lbl]below:$r_p$}] (r) at (0,0) {};
  \node[place] (q) at (1.2,0) {}; \node[ctok=okcol] at (q) {};
  \coordinate (fan) at (1.95,0);
  \node[transition,label={[guard]above:$[\omega\in O_1]$}] (t1) at (2.7,0.95) {};
  \node[transition,label={[guard]below:$[\omega\in O_2]$}] (t2) at (2.7,-0.95) {};
  \node[gplace] (g1) at (3.95,0.95) {}; \node[gplace] (g2) at (3.95,-0.95) {};
  \draw[arc] (r)--(q);
  \draw[arc,badcol] (q.east)--(fan)|-(t1); \draw[arc,badcol] (q.east)--(fan)|-(t2);
  \draw[arc] (t1)--(g1); \draw[arc] (t2)--(g2);
  \node[font=\scriptsize,text=badcol,align=center,anchor=north] at (1.95,-1.85) {$O_1\cap O_2\neq\emptyset$: scheduler $\kappa$ decides};
\end{tikzpicture}
\caption{Conflict: overlapping guards compete for one token, so a scheduler is needed.}
\end{subfigure}
\caption{Three routing patterns after a game resolves. The figure shows that ``the outcome selects the next game'' is only the special case (a).}
\label{fig:routing}
\end{figure}
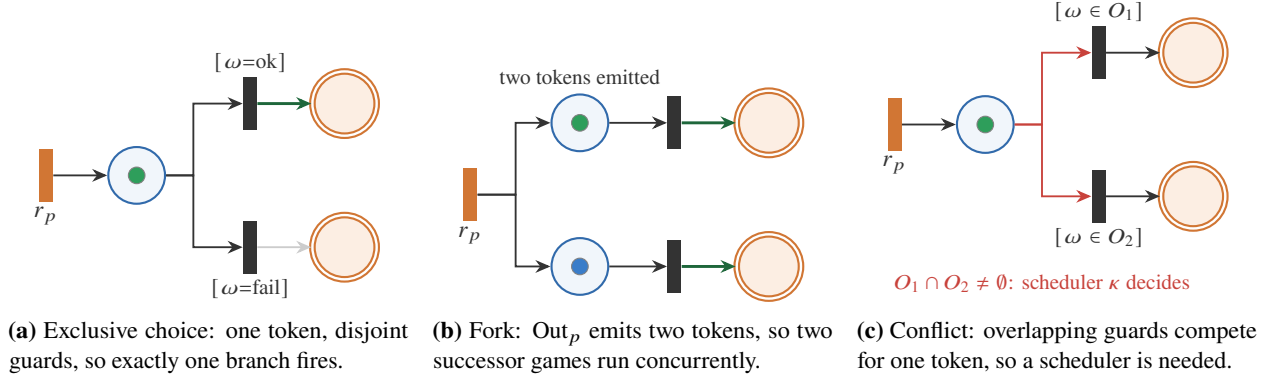

The familiar ``transition between games'' is the special case (a), written compactly with a \emph{successor map} $f$ (formalized in \cref{def:successor}):
\begin{equation}\label{eq:compact}
  \G_p \xrightarrow{\ \omega\ } \G_q,\qquad q = f(p,\omega).
\end{equation}
This notation hides the intermediate resource places and routing transitions. It should not replace the full marking semantics when parallel activities, synchronization, or shared resources matter.

\paragraph{Scheduler and global state.} An \emph{event} is a pair $\xi=(t,b)$ of a transition and a binding enabled at the current marking. The \emph{global state} $s$ consists of the marking together with any payoff-relevant private information, pending commitments, and the history summary needed to make future choices well defined. A \emph{scheduler} $\kappa(\xi\mid s)$ is a probability distribution over the events enabled at $s$. Throughout, $\kappa$ is \emph{exogenous and stationary}: it models engineering priority rules or timing races, and it is not a strategic player. An adversarial scheduler can be modeled by making it a player of a marking-level game (\cref{rem:beyond}). A \emph{strategy profile} $\pi$ assigns to each enabled resolution event $\xi=(r_p,b)$ a distribution $\pi(\cdot\mid s,\xi)$ over action profiles $A_p(b)$, and we write $\phi_\xi(\cdot\mid s,a):=\phi_p(\cdot\mid a,b)$. The induced transition kernel is
\begin{equation}\label{eq:kernel}
  \mathsf{P}^{\pi,\kappa}(s'\mid s) = \sum_{\xi,a,\omega}\kappa(\xi\mid s)\,\pi(a\mid s,\xi)\,\phi_\xi(\omega\mid s,a)\,\ind\{\mathrm{nx}_\xi(s,a,\omega)=s'\},
\end{equation}
where $\mathrm{nx}_\xi(s,a,\omega)$ is the next state after event $\xi$ fires with profile $a$ and outcome $\omega$, by \eqref{eq:fire}. For an ordinary event, the action and outcome factors are point masses.

\paragraph{Unfolding and the incidence matrix.} For matrix analysis we use the \emph{unfolded} net. Its places are $\widetilde P=\{(p,c):p\in P,\ c\in\mathrm{Col}_p\}$. Its transitions $\widetilde T$ are the pairs $(t,b)$ of an ordinary transition and a binding, together with the triples $(r_p,b,\chi)$ of a resolution transition, a binding, and an output class $\chi$, that is, a distinct value of $\Out_p(b,a,\omega)$. Resolutions with the same output share one column. A marking is identified with the vector in $\mathbb{N}^{|\widetilde P|}$ whose $(p,c)$ entry is $M_p(c)$, and $\mathbf{e}_{(p,c)}$ denotes the corresponding unit vector; unit vectors indexed by unfolded transitions are written $\mathbf{e}_{t}$ in the same way.

\begin{assumption}[Finite unfolding]\label{ass:finite}
The color domains, or at least the sets of colors and bindings that are reachable from $M_0$, are finite. Hence $\widetilde P$ and $\widetilde T$ are finite.
\end{assumption}

Under \cref{ass:finite}, define
\begin{equation}\label{eq:incidence}
  \Pre,\ \Post\in\mathbb{N}^{|\widetilde P|\times|\widetilde T|},\qquad C:=\Post-\Pre\in\mathbb{Z}^{|\widetilde P|\times|\widetilde T|}.
\end{equation}
Here the column of $\Pre$ (respectively $\Post$) indexed by an unfolded transition is the vector of the multiset it consumes (respectively produces); for $(r_p,b,\chi)$ the produced multiset is $\chi$. We use the same symbols $\Pre$ and $\Post$ for these matrices as for the multiset-valued maps; the meaning is clear from context.
\Cref{ass:finite} is substantive. A finite original net with infinite color domains, or with unboundedly many dynamic game parameters, need not have a finite unfolding. Let $n\in\mathbb{N}^{|\widetilde T|}$ count how often each column has fired along a firing sequence $M_0\to\cdots\to M$. Summing the firing rule over the sequence gives the \emph{state equation}
\begin{equation}\label{eq:state}
  M = M_0 + C\,n,\qquad n\ge 0 .
\end{equation}
The incidence matrix records the \emph{net change} produced by each unfolded transition. It does not, by itself, characterize which transitions are enabled or which markings are reachable. Guards that depend only on colors and bindings are compiled into the unfolding: a binding that a guard rejects has no column. Guards that depend on the rest of the marking disable columns at some markings without deleting them. The strategic layer enters only through \emph{which} output column fires, with probability determined by $\sigma_p$ and $\phi_p$.

\subsection{Sequential OCGPNs: local and global games}
\label{sec:games}

The semantics above say how a game outcome changes the marking. This subsection makes precise \emph{which game} the players are playing, both at one place (the local game) and over the whole net (the global game), and how the two are related. It applies to the class of \emph{sequential} OCGPNs defined by three standing assumptions, in addition to \cref{ass:finite}.

\begin{assumption}[Normal-form local games]\label{ass:normalform}
Each local game is a finite game in normal form with complete information. At a game place $p$, the binding $b$ is common knowledge among the players of $p$. Information sets $\mathsf{Inf}_{i,p}$ of \cref{def:localgame} are therefore trivial, and any private information must be encoded in $b$ and made public. The players of $p$ choose simultaneously. Different reachable bindings are treated as different game places of the unfolded net, so we write $\G_p$ for $\G_p(b)$.
\end{assumption}

\begin{assumption}[Sequential activation and exclusive routing]\label{ass:sequential}
At every reachable marking, at most one game place holds an active-instance token, and every other token needed for the analysis is folded into that token's color. When $\G_p$ resolves with outcome $\omega$, the emitted tokens enable exactly one routing path, so the scheduler has no choice to make. That path either activates exactly one game place or ends in exactly one \emph{terminal place}. Routing may depend on the action profile only through $\omega$, that is, $\Out_p(b,a,\omega)$ determines the same successor for all $a$.
\end{assumption}

\begin{assumption}[Discounting]\label{ass:discount}
The discount factor satisfies $\gamma\in(0,1)$. The value $\gamma=1$ is allowed only when the place graph defined below is acyclic.
\end{assumption}

\paragraph{The local game.} Under \cref{ass:normalform}, the local game of \cref{def:localgame} takes the following form.

\begin{definition}[Normal-form local game]\label{def:nflocal}
The local game at $p\in P_G$ is the tuple
\begin{equation}\label{eq:nflocal}
  \G_p=\big(\mathcal{I}_p,\ (A_i^p)_{i\in \mathcal{I}_p},\ (R_i^p)_{i\in \mathcal{I}_p},\ \Omega_p,\ \Phi^p\big),
\end{equation}
where:
\begin{itemize}[nosep,leftmargin=2em]
  \item $\mathcal{I}_p$ is a finite set of players;
  \item $A_i^p$ is a finite action set for player $i$, and $A^p=\prod_{i\in \mathcal{I}_p}A_i^p$ is the set of action profiles;
  \item $R_i^p:A^p\to\mathbb{R}$ is player $i$'s \emph{stage payoff}, the expected immediate reward $R_i^p(a)=\sum_{\omega}\phi_p(\omega\mid a)\,u_{i,p}(a,\omega)$;
  \item $\Omega_p$ is a finite outcome set;
  \item $\Phi^p=(\Phi^p_\omega)_{\omega\in\Omega_p}$ are the \emph{outcome-probability functions}, $\Phi^p_\omega(a)=\phi_p(\omega\mid a)$, so that $\sum_\omega\Phi^p_\omega(a)=1$ for every $a$.
\end{itemize}
A mixed strategy of player $i$ at $p$ is $\sigma_{i,p}\in\Delta(A_i^p)$, and $\sigma_p=(\sigma_{i,p})_{i\in \mathcal{I}_p}$ is a mixed profile. Any function $h:A^p\to\mathbb{R}$ extends multilinearly, $h(\sigma_p):=\sum_{a\in A^p}\prod_{i\in \mathcal{I}_p}\sigma_{i,p}(a_i)\,h(a)$.
\end{definition}

\begin{definition}[Outcome-probability matrices]\label{def:indicator}
When $|\mathcal{I}_p|=2$, $R_i^p$ and $\Phi^p_\omega$ are $|A_1^p|\times|A_2^p|$ matrices, $(\Phi^p_\omega)_{jk}=\phi_p(\omega\mid a_j,a_k)$, and $\sum_\omega\Phi^p_\omega=J$, the all-ones matrix. For mixed strategies $\alpha_1,\alpha_2$, the extension is bilinear: $h(\alpha_1,\alpha_2)=\alpha_1^\top h\,\alpha_2$. When $\phi_p$ is deterministic, each $\Phi^p_\omega$ is a $0$--$1$ indicator matrix. For more than two players, the same objects are tensors.
\end{definition}

\paragraph{The interface.} By \cref{ass:sequential}, every pair $(p,\omega)$ determines where play continues. Let $Z\subseteq P_R$ be the set of terminal places, with \emph{terminal values} $v_i:Z\to\mathbb{R}$.

\begin{definition}[Successor map, routing probabilities, place graph]\label{def:successor}
The \emph{successor map} $f:\{(p,\omega):p\in P_G,\ \omega\in\Omega_p\}\to P_G\cup Z$ sends $(p,\omega)$ to the game place activated, or to the terminal place reached, after $\G_p$ resolves with outcome $\omega$. For $q\in P_G\cup Z$, the \emph{routing probability} to $q$ is
\begin{equation}\label{eq:Psi}
  \Psi^p_q(a)=\sum_{\omega:\,f(p,\omega)=q}\Phi^p_\omega(a),\qquad \sum_q\Psi^p_q(a)=1 .
\end{equation}
The \emph{place graph} $\mathcal{D}$ has vertex set $P_G$ and an edge $p\to q$ whenever $\Psi^p_q\not\equiv0$ for some $q\in P_G$.
\end{definition}

The successor map is the compact notation \eqref{eq:compact}, computed from $\Out_p$ and the routing guards. Everything the players of $p$ need to know about the rest of the net passes through $\Psi^p$ and the values at its targets. Distinct outcomes with the same successor are strategically equivalent for continuation purposes. This is why the analysis below uses $\Psi$ rather than $\Phi$.

\paragraph{The global game.} The local games and the interface together define one dynamic game.

\begin{definition}[Global game]\label{def:global}
The \emph{global game} of a sequential OCGPN $\N$ is the stochastic game $\Gamma(\N)=\big(\mathcal{I},\ P_G\cup Z,\ (\G_p)_{p\in P_G},\ f,\ v,\ \gamma,\ p_0\big)$. The player set is the finite union $\mathcal{I}=\bigcup_{p}\mathcal{I}_p$. A player $i\notin\mathcal{I}_p$ has the single no-op action at $p$ and stage payoff $R_i^p\equiv0$. The initial game place $p_0$ is the one activated by $M_0$. Play proceeds in stages $\ell=0,1,\dots$ as follows:
\begin{enumerate}[label=(\arabic*),nosep]
  \item at stage $\ell$ the active place is $p_\ell$, and the players choose $a_\ell\in A^{p_\ell}$;
  \item each $i\in\mathcal{I}$ receives $R_i^{p_\ell}(a_\ell)$;
  \item an outcome $\omega_\ell\sim\phi_{p_\ell}(\cdot\mid a_\ell)$ is drawn, and $p_{\ell+1}=f(p_\ell,\omega_\ell)$;
  \item play stops at the first time $\tau$ with $p_\tau\in Z$, and each player receives the terminal value $v_i(p_\tau)$, discounted by $\gamma^\tau$.
\end{enumerate}
A \emph{behavior strategy} of player $i$ maps each history $(p_0,a_0,\dots,p_\ell)$ to a mixed action in $\Delta(A_i^{p_\ell})$. A \emph{stationary strategy} depends only on the current place: $\sigma_i=(\sigma_{i,p})_{p\in P_G}$. We write $\sigma=(\sigma_p)_{p\in P_G}$ for a stationary profile. Player $i$'s payoff from place $p$ is
\begin{equation}\label{eq:globalpayoff}
  U_i(\sigma;p)=\Ex^\sigma_p\Big[\sum_{\ell=0}^{\tau-1}\gamma^\ell R_i^{p_\ell}(a_\ell)+\gamma^\tau v_i(p_\tau)\,\ind\{\tau<\infty\}\Big].
\end{equation}
Since all stage payoffs and terminal values are finite in number, they are bounded, and \eqref{eq:globalpayoff} is well defined under \cref{ass:discount}.
\end{definition}

The global game is the place-level reading of the marking-level dynamics \eqref{eq:kernel}. Its states are the active game places, its stage games are the local games, and its transition law is $\Pr[p_{\ell+1}=q\mid p_\ell=p,a]=\Psi^p_q(a)$.

\paragraph{Policy evaluation.} For a stationary profile $\sigma$, define the \emph{routing matrix} $P^\sigma\in\mathbb{R}^{P_G\times P_G}$, the \emph{exit matrix} $\Xi^\sigma\in\mathbb{R}^{P_G\times Z}$, and the \emph{reward vectors} $\mathbf{r}_i^\sigma\in\mathbb{R}^{P_G}$ by
\begin{equation}\label{eq:routing}
  P^\sigma_{pq}=\Psi^p_q(\sigma_p),\qquad
  \Xi^\sigma_{pz}=\Psi^p_z(\sigma_p),\qquad
  \mathbf{r}_i^\sigma(p)=R_i^p(\sigma_p)+\gamma\sum_{z\in Z}\Xi^\sigma_{pz}\,v_i(z).
\end{equation}
Each row of $[P^\sigma\ \Xi^\sigma]$ sums to one.

\begin{lemma}[Policy evaluation]\label{lem:policy}
Under \cref{ass:sequential,ass:discount}, the value vector $V_i^\sigma:=U_i(\sigma;\cdot)\in\mathbb{R}^{P_G}$ is the unique solution of $V=\mathbf{r}_i^\sigma+\gamma P^\sigma V$, that is,
\begin{equation}\label{eq:neumann}
  V_i^\sigma=(I-\gamma P^\sigma)^{-1}\mathbf{r}_i^\sigma=\sum_{k\ge0}\gamma^k(P^\sigma)^k \mathbf{r}_i^\sigma .
\end{equation}
\end{lemma}
\begin{proof}
Condition \eqref{eq:globalpayoff} on the first stage. The stage reward is $R_i^p(\sigma_p)$. With probability $\Xi^\sigma_{pz}$, play stops at $z$ after one stage, worth $\gamma v_i(z)$. With probability $P^\sigma_{pq}$, play continues from $q$, worth $\gamma V_i^\sigma(q)$ by stationarity. This gives $V=\mathbf{r}_i^\sigma+\gamma P^\sigma V$. Since $P^\sigma$ is substochastic, its spectral radius (largest modulus of an eigenvalue) satisfies $\rho(\gamma P^\sigma)\le\gamma<1$, so $I-\gamma P^\sigma$ is invertible and the Neumann series converges. If $\gamma=1$, the place graph is acyclic by \cref{ass:discount}, so $P^\sigma$ is nilpotent (\cref{prop:nilpotent}) and the same conclusion holds.
\end{proof}

\paragraph{Local and global equilibria.} The continuation values of the global game enter each local game through the following construction.

\begin{definition}[Augmented local game]\label{def:augmented}
Let $V=(V_i)_{i\in\mathcal{I}}$ with $V_i:P_G\to\mathbb{R}$, and define the \emph{extended value}
\begin{equation}
  \overline V_i(q)=\begin{cases}V_i(q),&q\in P_G,\\ v_i(q),&q\in Z.\end{cases}
\end{equation}
The \emph{augmented local game} $\G_p[V]$ has the players and actions of $\G_p$ and the payoffs
\begin{equation}\label{eq:Qmatrix}
  Q_i^p(V)(a) \;=\; R_i^p(a) + \gamma\sum_{\omega\in\Omega_p}\Phi^p_\omega(a)\,\overline V_i\big(f(p,\omega)\big)
  \;=\; R_i^p(a) + \gamma\sum_{q\in P_G\cup Z}\Psi^p_q(a)\,\overline V_i(q),\qquad a\in A^p .
\end{equation}
In the two-player case, \eqref{eq:Qmatrix} is the matrix identity $Q_i^p(V)=R_i^p+\gamma\sum_q \overline V_i(q)\,\Psi^p_q$.
\end{definition}

\begin{definition}[Equilibrium notions]\label{def:equilibria}
Fix a stationary profile $\sigma$.
\begin{enumerate}[label=(\roman*),nosep]
  \item $\sigma_p$ is a \emph{myopic local equilibrium} if it is a Nash equilibrium of the stage game $(R_i^p)_{i\in \mathcal{I}_p}$.
  \item $\sigma_p$ is an \emph{augmented local equilibrium at $V$} if it is a Nash equilibrium of $\G_p[V]$. We write $\mathcal{E}_p(V)$ for the set of these.
  \item $\sigma$ is a \emph{global equilibrium}, that is, a stationary Markov perfect equilibrium of $\Gamma(\N)$, if for every player $i$, every place $p\in P_G$, and every behavior strategy $\sigma_i'$, we have $U_i(\sigma;p)\ge U_i(\sigma_i',\sigma_{-i};p)$.
\end{enumerate}
The augmented correspondence $\mathcal{E}_p(V)$ depends on the continuation values of the whole net, whereas the admissibility correspondence $\mathcal{A}_p$ of \cref{def:ocgpn} is a modeling input. When $\mathcal{A}_p$ is a strict restriction, the theory below applies to the game in which player $i$'s action set at $p$ is restricted accordingly.
\end{definition}

\begin{theorem}[Local--global consistency]\label{thm:consistency}
Under \cref{ass:finite,ass:normalform,ass:sequential,ass:discount}, a stationary profile $\sigma^\star$ is a global equilibrium if and only if $\sigma^\star_p\in\mathcal{E}_p(V^{\sigma^\star})$ for every game place $p$, where $V^{\sigma^\star}=(V_i^{\sigma^\star})_{i\in\mathcal{I}}$.
\end{theorem}
\begin{proof}
By \cref{lem:policy}, $V_i^{\sigma^\star}(p)=\mathbf{r}_i^{\sigma^\star}(p)+\gamma(P^{\sigma^\star}V_i^{\sigma^\star})(p)=Q_i^p(V^{\sigma^\star})(\sigma^\star_p)$.

($\Rightarrow$) Suppose that at some $p$ a player $i\in \mathcal{I}_p$ has a mixed action $\alpha$ with $Q_i^p(V^{\sigma^\star})(\alpha,\sigma^\star_{-i,p})>Q_i^p(V^{\sigma^\star})(\sigma^\star_p)=V_i^{\sigma^\star}(p)$. Consider the strategy that plays $\alpha$ at the first visit to $p$ and follows $\sigma^\star_i$ afterwards. Its payoff from $p$ is exactly $Q_i^p(V^{\sigma^\star})(\alpha,\sigma^\star_{-i,p})$, which exceeds $U_i(\sigma^\star;p)$. So $\sigma^\star$ is not a global equilibrium.

($\Leftarrow$) Fix $i$ and the others' stationary strategies $\sigma^\star_{-i}$. Player $i$ then faces a finite Markov decision process on $P_G$ with bounded rewards. Its Bellman optimality equation is $W(p)=\max_{\alpha\in\Delta(A_i^p)}Q_i^p(W)(\alpha,\sigma^\star_{-i,p})$. The right side is a $\gamma$-contraction in $W$ (by \cref{lem:affine} below), or a finite backward recursion when the place graph is acyclic. Hence it has a unique solution $W^\star$, which equals player $i$'s optimal payoff over all behavior strategies~\cite{puterman1994}. By hypothesis, $\sigma^\star_{i,p}$ attains the maximum at $W=V_i^{\sigma^\star}$ for every $p$. Hence $V_i^{\sigma^\star}$ solves the optimality equation, so $V_i^{\sigma^\star}=W^\star$. No deviation of player $i$ improves on $\sigma^\star$, from any place.
\end{proof}

\Cref{thm:consistency} specializes the one-shot deviation principle for stochastic games to the place-level game of a sequential OCGPN. A global equilibrium is a collection of local equilibria, one per place, of the games \emph{augmented} by the continuation values that the collection itself generates. Myopic local equilibria are a different object, and they coincide with global equilibria only in special cases (\cref{cor:myopic} below).

\begin{proposition}[Existence and uniqueness of values]\label{prop:existence}
Under \cref{ass:finite,ass:normalform,ass:sequential,ass:discount}, a global equilibrium exists. If $|\mathcal{I}_p|=2$ for every $p$, with a common pair of players, and $R_2^p=-R_1^p$, $v_2=-v_1$ (the zero-sum case), then all global equilibria give player~1 the same value vector $V^\star$.
\end{proposition}
\begin{proof}
For $\gamma<1$, the global game is a finite discounted stochastic game with a finite player set and bounded payoffs, so a stationary equilibrium exists by the theorem of Fink and Takahashi~\cite{fink1964,takahashi1964}. For $\gamma=1$ with an acyclic place graph, \cref{thm:backward} below constructs one. In the zero-sum case, \cref{thm:contraction} below shows that $V^\star$ is the unique fixed point of the Shapley operator, which every equilibrium value must satisfy by \cref{thm:consistency}.
\end{proof}

\begin{remark}[Extensions beyond sequential OCGPNs]\label{rem:beyond}
The standing assumptions can be relaxed in three directions.
\begin{itemize}[nosep,leftmargin=1.2em]
  \item \emph{Extensive-form local games.} If $\G_p$ is a finite extensive-form game with perfect recall, as in running example~II, the results hold with actions replaced by behavior strategies of $\G_p$. The local LP is then the sequence-form program~\cite{koller1996efficient}.
  \item \emph{Private information.} If token colors carry private types, $\G_p$ becomes a Bayesian game. Its information sets must not reveal hidden colors merely because the full marking contains them. Global consistency may then require perfect Bayesian or sequential equilibrium, and independently computed local Bayesian equilibria are generally not enough.
  \item \emph{Concurrency and strategic scheduling.} Without \cref{ass:sequential}, the state of the global game is the full global state $s$, with kernel \eqref{eq:kernel}. Let $R_i(s,\xi,a,\omega)$ denote player $i$'s immediate reward when event $\xi$ fires at $s$ with profile $a$ and outcome $\omega$. The augmented payoff becomes
  \begin{equation}\label{eq:Q}
    Q_i(s,\xi,a) = \Ex_{\omega\sim\phi_\xi(\cdot\mid s,a)}\Big[R_i(s,\xi,a,\omega) + \gamma\,V_i\big(\mathrm{nx}_\xi(s,a,\omega)\big)\Big],
  \end{equation}
  and \cref{thm:consistency} holds with places replaced by global states, at the cost of a larger state space. An adversarial scheduler becomes an additional player who chooses $\xi$. The resulting game is a finite stochastic game~\cite{shapley1953}. The OCGPN provides a compositional representation of it, with typed, localized state components; exact value computation still requires an adequate Markov state.
\end{itemize}
\end{remark}

\subsection{Structural analysis: reachability, safety, and progress}
\label{sec:structural}

The results of this subsection hold for general OCGPNs under \cref{ass:finite}. \Cref{eq:state} is necessary for reachability but not sufficient, because it ignores the order of firings and the enabling conditions at intermediate markings. This one-sidedness is what makes it useful for safety: if \eqref{eq:state} has no solution, the marking is certainly unreachable.

\begin{proposition}[Invariants]\label{prop:invariants}
\begin{enumerate}[label=(\roman*),nosep]
  \item \textbf{P-invariant.} If $y\in\mathbb{R}^{|\widetilde P|}$ satisfies $y^\top C=0$, then $y^\top M=y^\top M_0$ for every reachable $M$.
  \item \textbf{T-invariant.} If $n\in\mathbb{N}^{|\widetilde T|}\setminus\{0\}$ satisfies $Cn=0$, then any firing sequence with count vector $n$ that starts at $M$ returns to $M$. Every cycle of the reachability graph has a count vector that is a T-invariant. A T-invariant need not be \emph{realizable}: there may be no marking from which a sequence with count vector $n$ can actually fire.
\end{enumerate}
\end{proposition}
\begin{proof}
(i) Multiply \eqref{eq:state} on the left by $y^\top$. (ii) If a sequence with count vector $n$ fires from $M$, then \eqref{eq:state} gives $M'=M+Cn=M$. Conversely, a cycle returns to its start, so its count vector satisfies $Cn=0$. Realizability fails, for instance, when a transition in the support of $n$ is never enabled.
\end{proof}

\begin{proposition}[Farkas certificate of unreachability]\label{prop:farkas}
Let $B=\{M: M\ge m_B\}$ be a set of bad markings described by a lower bound $m_B$. Suppose there is a $y\ge0$ such that
\begin{equation}\label{eq:farkas}
  C^\top y\le 0\qquad\text{and}\qquad y^\top(m_B-M_0)>0 .
\end{equation}
Then no marking in $B$ is reachable. Conversely, if no such $y$ exists, the \emph{real} relaxation $\{n\in\mathbb{R}^{|\widetilde T|}:n\ge0,\ M_0+Cn\ge m_B\}$ is feasible.
\end{proposition}
\begin{proof}
Suppose $M\in B$ were reachable with count vector $n\ge0$. Then $y^\top m_B\le y^\top M=y^\top M_0+(C^\top y)^\top n\le y^\top M_0$, which contradicts \eqref{eq:farkas}. The converse is Farkas' lemma for the system $Cn\ge m_B-M_0$, $n\ge0$.
\end{proof}

The certificate is sound, but the method is incomplete in two ways. Real feasibility does not imply integer feasibility, and integer feasibility of the state equation does not imply reachability. A missing certificate therefore proves nothing. The vector $y$ is a \emph{weighted token count that no firing can increase}, while the bad set requires it to increase. Finding $y$ is a linear program, so a safety certificate can be computed rather than argued. Let $\Reach(M_0)$ denote the set of markings reachable from $M_0$ by finite firing sequences. More generally, any property $\mathrm{Inv}$ that holds at $M_0$ and is preserved by every firing holds at every reachable marking:
\begin{equation}\label{eq:invariant}
  \mathrm{Inv}(M_0)\ \wedge\ \big(\mathrm{Inv}(M)\wedge M\xrightarrow{\xi}M' \Rightarrow \mathrm{Inv}(M')\big)\ \Longrightarrow\ \forall M\in\Reach(M_0):\ \mathrm{Inv}(M).
\end{equation}
Because columns are indexed by outputs, not by equilibria, this argument covers every equilibrium at once, provided the unfolded net contains a column for every output that can actually occur.

\begin{designrule}[Structural and incentive instruments]
Structural certificates depend only on $C$. A binding-level guard deletes columns of the unfolding, so it can create a certificate. An incentive instrument (a payoff change, a decoy, a subsidy) changes only which columns are \emph{likely} to fire, so it can never produce one. Incentive instruments change values; guards change reachability. This also separates the \emph{existence} of a safe equilibrium from \emph{universal} safety: an equilibrium-selection rule may permit unsafe outcomes, whereas a guard excludes them for every profile.
\end{designrule}

\paragraph{Progress, deadlock, and boundedness.} A marking is a \emph{deadlock} if no event is enabled while unfinished work remains. An infinite execution that never completes the work is a \emph{livelock}. Outcome-dependent branching can cause deadlock through missing tokens, guards with no matching branch, resource starvation, or inconsistent certificates. Infinite executions require care. In a \emph{bounded} net, whose reachability set is finite, an infinite execution must revisit some marking. It then contains a cycle, whose count vector is a T-invariant (\cref{prop:invariants}(ii)). In an unbounded net this fails. A transition that consumes one token from $p$ and returns two has $C=[1]$ and no nonzero T-invariant, yet it fires forever through the markings $1,2,3,\dots$. For bounded nets, whether an execution repeats a cycle forever is a \emph{strategic} question: it depends on the probability that equilibrium play selects the output columns of that cycle. \Cref{sec:values} makes this quantitative. In particular, \emph{the existence of a Nash equilibrium implies neither deadlock freedom, nor eventual success, nor termination.}

\subsection{Strategic analysis of a game place}
\label{sec:strategic}

This subsection and the next concern sequential OCGPNs. By \cref{thm:consistency}, equilibrium play at a place $p$ is Nash play of the augmented game $\G_p[V]$, for the continuation values $V$ that the equilibrium generates. Here we study how the augmentation \eqref{eq:Qmatrix} reshapes incentives at $p$, for a \emph{given} value vector $V$ (\cref{fig:backward}).

\begin{figure}[ht]
\centering
\begin{tikzpicture}[x=1cm,y=1cm]
  \node[gplace,label={[glbl]left:$\G_P$}] (gp) at (0,0) {};
  \node[rtrans,label={[lbl]below:$r_P$}] (r) at (1.4,0) {};
  \node[place,label={[lbl]below:$q_P$}] (q) at (2.7,0) {};
  \coordinate (fan) at (3.45,0);
  \node[transition,label={[guard]above:$[\mathrm{ready}]$}] (ta) at (4.3,1.1) {};
  \node[transition,label={[guard]below:$[\mathrm{notReady}]$}] (tr) at (4.3,-1.1) {};
  \node[gplace,label={[glbl]above:$\G_A$}] (ga) at (5.6,1.1) {};
  \node[gplace,label={[glbl]below:$\G_R$}] (gr) at (5.6,-1.1) {};
  \node[draw=okcol!70!black,line width=0.7pt,fill=okcol!7,rounded corners=2pt,font=\footnotesize,inner sep=4pt] (va) at (8.7,1.1) {$V(\G_A)$: equilibrium value};
  \node[draw=badcol!70!black,line width=0.7pt,fill=badcol!5,rounded corners=2pt,font=\footnotesize,inner sep=4pt] (vr) at (8.7,-1.1) {$V(\G_R)$: equilibrium value};
  \foreach \u/\v in {gp/r,r/q,ta/ga,tr/gr} \draw[arc] (\u)--(\v);
  \draw[arc] (q.east)--(fan)|-(ta); \draw[arc] (q.east)--(fan)|-(tr);
  \draw[darc,draw=inkcol!40] (ga)--(va); \draw[darc,draw=inkcol!40] (gr)--(vr);
  \draw[->,line width=0.9pt,draw=okcol!65!black,dashed] (va.north) |- node[pos=0.75,above,lbl]{$+\gamma\,V(\G_A)$ if $\omega=\mathrm{ready}$} (0,2.35) -- (gp.north);
  \draw[->,line width=0.9pt,draw=badcol!80!black,dashed] (vr.south) |- node[pos=0.75,below,lbl]{$+\gamma\,V(\G_R)$ if $\omega=\mathrm{notReady}$} (0,-2.35) -- (gp.south);
  \node[cbox,font=\scriptsize,align=center,anchor=east] at (-0.9,1.3) {solve $\G_P$ with\\ $Q_i = R_i+\gamma V_i(\text{successor})$};
\end{tikzpicture}
\caption{Continuation values flow backward along routing paths. The successor games are solved first. Their equilibrium values are discounted and added to the payoffs at $\G_P$, weighted by the probability of being routed to each successor. This produces the augmented game $\G_P[V]$ of \cref{def:augmented}. The figure uses the planning game of running example~I (\cref{sec:re1}); $V(\G)$ denotes the equilibrium value of game $\G$. Players at $\G_P$ are, in effect, choosing which game they will play next.}
\label{fig:backward}
\end{figure}
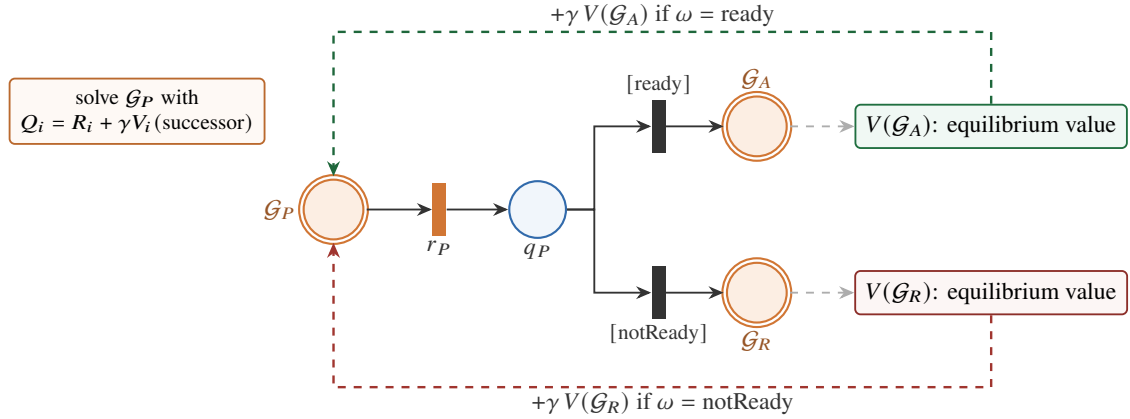

\begin{lemma}[Affinity and Lipschitz continuity]\label{lem:affine}
$Q_i^p(V)$ is affine in $V_i$ and depends on no other player's values. Moreover, $\max_{a}|Q_i^p(V)(a)-Q_i^p(V')(a)|\le\gamma\|V_i-V_i'\|_\infty$. Adding a constant $c$ to every extended value $\overline V_i(q)$, $q\in P_G\cup Z$, adds $\gamma c$ to every entry of $Q_i^p$, so it changes no best reply.
\end{lemma}
\begin{proof}
Affinity is immediate from \eqref{eq:Qmatrix}. For each $a$, $|Q_i^p(V)(a)-Q_i^p(V')(a)|\le\gamma\sum_q\Psi^p_q(a)\,\|V_i-V_i'\|_\infty=\gamma\|V_i-V_i'\|_\infty$, because the routing probabilities sum to one and $\overline V_i=\overline V_i'$ on $Z$. Replacing $\overline V_i$ by $\overline V_i+c$ adds $\gamma c\sum_q\Psi^p_q(a)=\gamma c$ to every entry.
\end{proof}

Incentives depend on payoff \emph{differences} between a player's own actions. For player $i\in\mathcal{I}_p$, two of its actions $a_i^+,a_i^-$, and any function $h:A^p\to\mathbb{R}$, define the \emph{difference operator}
\begin{equation}\label{eq:Dop}
  D_{a_{-i}}(h) := h(a_i^+,a_{-i})-h(a_i^-,a_{-i}),\qquad a_{-i}\in A_{-i}^p:=\textstyle\prod_{k\in \mathcal{I}_p\setminus\{i\}}A_k^p .
\end{equation}
For the row player of a bimatrix game, $a_{-i}=j$ indexes a column and $D_j$ is a row difference. For the column player, $D$ is a column difference. Because $\sum_q D_{a_{-i}}(\Psi^p_q)=D_{a_{-i}}(1)=0$, only differences between continuation values can matter.

\begin{definition}[Successor pivotality]\label{def:pivotal}
Successor $q\in P_G\cup Z$ is \emph{pivotal} for player $i$ against $a_{-i}$ if $D_{a_{-i}}(\Psi^p_q)\neq0$, that is, if switching from $a_i^-$ to $a_i^+$ against $a_{-i}$ changes the probability of being routed to $q$. The interface at $p$ is \emph{non-pivotal for player $i$} if no successor is pivotal against any $a_{-i}$, for any pair of $i$'s actions.
\end{definition}

Pivotality is defined for successors, not outcomes. If two outcomes lead to the same successor, a player may shift probability between them, so that each $\Phi^p_\omega$ changes, without changing $\Psi^p_q$ or any continuation payoff. Such outcome-level changes are strategically irrelevant for continuation.

\begin{theorem}[Dominance through the interface]\label{thm:dominance}
For every value vector $V$, action $a_i^+$ strictly dominates $a_i^-$ for player $i$ in $\G_p[V]$ if and only if
\begin{equation}\label{eq:dominance}
  D_{a_{-i}}(R_i^p) + \gamma\sum_{q\in P_G\cup Z} \overline V_i(q)\, D_{a_{-i}}(\Psi^p_q) \;>\; 0\qquad\text{for every } a_{-i}\in A^p_{-i} .
\end{equation}
In particular, if no successor is pivotal for $i$ against some $a_{-i}$, then the comparison between $a_i^+$ and $a_i^-$ against $a_{-i}$ is the same as in the stage game, whatever the continuation values.
\end{theorem}
\begin{proof}
Strict dominance means $D_{a_{-i}}(Q_i^p(V))>0$ for every $a_{-i}$. Since $D_{a_{-i}}$ is linear, apply it to the second form of \eqref{eq:Qmatrix}. If $D_{a_{-i}}(\Psi^p_q)=0$ for all $q$, the sum vanishes and $D_{a_{-i}}(Q_i^p)=D_{a_{-i}}(R_i^p)$.
\end{proof}

\begin{corollary}[Best-reply thresholds for exogenous values]\label{cor:thresholds}
Hold the continuation values $\overline V_i$ fixed as $\gamma$ varies. Then $D_{a_{-i}}(Q_i^p)$ is affine in $\gamma$. So for each $a_{-i}$ there is at most one threshold at which player $i$'s comparison of $a_i^+$ and $a_i^-$ against $a_{-i}$ switches, namely
\begin{equation}
  \gamma^\star_{a_{-i}}=-\frac{D_{a_{-i}}(R_i^p)}{\sum_q \overline V_i(q)\,D_{a_{-i}}(\Psi^p_q)}
\end{equation}
when the denominator is nonzero. When it is zero, no discount factor changes that comparison.
\end{corollary}

\begin{remark}[Endogenous values]\label{rem:endogenous}
\Cref{cor:thresholds} treats continuation values as exogenous. That is exact when the successors are terminal or their values do not depend on $\gamma$, as in variant A of running example~I (\cref{sec:re1}). In a genuine dynamic game, $\overline V_i=\overline V_i^\star(\gamma)$, and the comparison $D_{a_{-i}}(R_i^p)+\gamma\sum_q\overline V^\star_i(q;\gamma)D_{a_{-i}}(\Psi^p_q)$ need not be affine in $\gamma$. It may then have several switching points, which must be found by solving the global game for each $\gamma$, or by parametric analysis of the fixed point.
\end{remark}

\begin{corollary}[Global consistency of myopic play]\label{cor:myopic}
If the interface at $p$ is non-pivotal for every player of $p$, then $\G_p[V]$ and the stage game $\G_p$ have the same best-reply correspondences for every $V$. If this holds at every game place, then every profile of myopic local equilibria is a global equilibrium. Conversely, successor pivotality at $p$ is necessary for continuation values to change any best reply at $p$.
\end{corollary}
\begin{proof}
If $D_{a_{-i}}(\Psi^p_q)=0$ for all $q$, all $a_{-i}$, and all action pairs of $i$, then $Q_i^p(V)(a)-R_i^p(a)$ does not depend on $a_i$. So against any mixed $\sigma_{-i,p}$, the two payoff functions differ by a constant in $a_i$, and their best replies coincide. Global consistency then follows from \cref{thm:consistency}. The converse is the contrapositive.
\end{proof}

\begin{remark}[Interpretation for interface design]
\Cref{thm:dominance} turns interface design into a linear feasibility problem. To make $a_i^+$ dominant, choose the routing probabilities $\Psi^p_q$ and the successor values so that the left side of \eqref{eq:dominance} is positive for every $a_{-i}$. A necessary condition is that some successor is pivotal against every $a_{-i}$ with $D_{a_{-i}}(R_i^p)\le0$.
\end{remark}

\paragraph{Solving a local game.} In the zero-sum two-player case ($Q_2^p=-Q_1^p=:-Q$), the value and optimal strategies at $p$ come from one linear program:
\begin{equation}\label{eq:lp}
  \mathrm{val}(Q)=\max_{\alpha,\bar v}\ \bar v\quad\text{s.t.}\quad Q^\top \alpha\ge \bar v\,\mathbf{1},\ \ \mathbf{1}^\top \alpha=1,\ \ \alpha\ge0,
\end{equation}
where $\mathbf{1}$ is the all-ones vector, and its dual gives the column player's strategy. For general-sum bimatrix games, equilibria can be computed by support enumeration or by the Lemke--Howson algorithm. Since equilibria need not be unique, the solution layer must then include a selection rule.

\subsection{Global values and computation}
\label{sec:values}

\Cref{lem:policy} evaluates a \emph{fixed} stationary profile. This subsection computes \emph{equilibrium} values. In the zero-sum case, \cref{thm:consistency} reduces the global game to a fixed-point problem on the place graph, with one matrix game per place.

\begin{theorem}[Place-level Shapley operator]\label{thm:contraction}
Consider the zero-sum case with a common pair of players at every place and $\gamma<1$. Write $V\in\mathbb{R}^{P_G}$ for player~1's values (player~2's are $-V$), and define $\mathcal{B}:\mathbb{R}^{P_G}\to\mathbb{R}^{P_G}$ by $(\mathcal{B}V)(p)=\mathrm{val}\big(Q_1^p(V)\big)$. Then:
\begin{enumerate}[label=(\roman*),nosep]
  \item $\mathcal{B}$ is a $\gamma$-contraction in $\|\cdot\|_\infty$, with a unique fixed point $V^\star$;
  \item $V^\star(p)$ is the value of the global game $\Gamma(\N)$ started at $p$;
  \item a stationary profile $\sigma^\star$ is a global equilibrium if and only if, at every $p$, $\sigma^\star_p$ is a pair of optimal strategies of the matrix game $Q_1^p(V^\star)$;
  \item value iteration $V^{(k+1)}=\mathcal{B}V^{(k)}$ converges geometrically, and every iterate carries the computable error bound
  \begin{equation}\label{eq:certificate}
    \|V^{(k)}-V^\star\|_\infty\ \le\ \frac{\|\mathcal{B}V^{(k)}-V^{(k)}\|_\infty}{1-\gamma}.
  \end{equation}
\end{enumerate}
\end{theorem}
\begin{proof}
(i) The value of a matrix game is nonexpansive in the max norm: if $\max_a|Q(a)-Q'(a)|\le\epsilon$, then $Q(\alpha_1,\alpha_2)$ and $Q'(\alpha_1,\alpha_2)$ differ by at most $\epsilon$ for all mixed $\alpha_1,\alpha_2$, so their max--min values do too. Combined with \cref{lem:affine}, this gives $|(\mathcal{B}V)(p)-(\mathcal{B}V')(p)|\le\gamma\|V-V'\|_\infty$, and Banach's fixed-point theorem applies. (ii)--(iii) In a zero-sum matrix game, Nash equilibria are exactly pairs of optimal strategies, and they all yield the value. If $\sigma^\star$ is a global equilibrium, then by \cref{thm:consistency}, $V^{\sigma^\star}(p)=\mathrm{val}(Q_1^p(V^{\sigma^\star}))$ for all $p$, so $V^{\sigma^\star}$ is a fixed point of $\mathcal{B}$ and equals $V^\star$. Conversely, if $\sigma^\star_p$ is optimal in $Q_1^p(V^\star)$ at every $p$, then $V^\star$ solves the evaluation equation of \cref{lem:policy} for $\sigma^\star$, so $V^{\sigma^\star}=V^\star$, and \cref{thm:consistency} applies. This is Shapley's theorem~\cite{shapley1953} with game places as states. (iv) $\|V^{(k)}-V^\star\|\le\|V^{(k)}-\mathcal{B}V^{(k)}\|+\|\mathcal{B}V^{(k)}-\mathcal{B}V^\star\|\le\|\mathcal{B}V^{(k)}-V^{(k)}\|+\gamma\|V^{(k)}-V^\star\|$.
\end{proof}

\begin{corollary}[Cost of a coarse summary]\label{cor:coarse}
Suppose some game places use summary matrices, and let $\widehat{\mathcal{B}}$ be the resulting operator, also a $\gamma$-contraction. Define the \emph{local mismatch} $\varepsilon_p:=|(\widehat{\mathcal{B}}V^\star)(p)-(\mathcal{B}V^\star)(p)|$, which is zero at places that keep their exact matrices. Then the fixed point $\widehat V$ of $\widehat{\mathcal{B}}$ satisfies $\|\widehat V-V^\star\|_\infty\le\max_p\varepsilon_p/(1-\gamma)$.
\end{corollary}
\begin{proof}
$\|\widehat V-V^\star\|=\|\widehat{\mathcal{B}}\widehat V-\mathcal{B}V^\star\|\le\|\widehat{\mathcal{B}}\widehat V-\widehat{\mathcal{B}}V^\star\|+\|\widehat{\mathcal{B}}V^\star-\mathcal{B}V^\star\|\le\gamma\|\widehat V-V^\star\|+\max_p\varepsilon_p$. Rearrange. This is the multi-resolution bound of~\cite{yang2025general}.
\end{proof}

\begin{remark}[General-sum computation]\label{rem:generalsum}
Without the zero-sum structure, there is no contraction, and global equilibria can be multiple. \Cref{thm:consistency} still reduces the problem to place-level conditions. On an acyclic place graph, backward induction with a selection rule at each place constructs a global equilibrium (\cref{thm:backward}). On a cyclic graph, the components identified below can still be solved one at a time. Within a component, a fixed point must be found by a homotopy, learning, or a selection-guided iteration, and convergence is no longer guaranteed.
\end{remark}

\begin{proposition}[Backward induction as a finite Neumann series]\label{prop:nilpotent}
If the place graph $\mathcal{D}$ is acyclic, then ordering $P_G$ topologically makes $P^\sigma$ strictly upper triangular, hence nilpotent. The series \eqref{eq:neumann} then has at most $|P_G|$ nonzero terms, and solving \eqref{eq:neumann} by back-substitution is exactly backward induction. The procedure terminates even when $\gamma=1$.
\end{proposition}
\begin{proof}
Every edge goes from an earlier place to a later one, so $P^\sigma_{pq}\neq0$ only if $p$ precedes $q$. A strictly upper-triangular $|P_G|\times|P_G|$ matrix $N$ satisfies $N^{|P_G|}=0$. Back-substitution on the triangular system $(I-\gamma P^\sigma)V=\mathbf{r}^\sigma$ solves the last place first, which is the backward-induction order. Since $I-\gamma P^\sigma$ is unit triangular, it is invertible for every $\gamma$, including $\gamma=1$.
\end{proof}

\begin{theorem}[Backward induction on acyclic sequential OCGPNs]\label{thm:backward}
Let the place graph $\mathcal{D}$ be acyclic, and order $P_G$ topologically. Process the game places in reverse order. At each place $p$, the successors of $p$ have already been processed, so their extended values $\overline V(q)$ are fixed. Select any $\sigma_p\in\mathcal{E}_p(V)$ of the resulting augmented game $\G_p[V]$, which exists by Nash's theorem, and set $V_i(p):=Q_i^p(V)(\sigma_p)$. The resulting stationary profile $\sigma$ is a global equilibrium of $\Gamma(\N)$, for every $\gamma\in(0,1]$ and in the general-sum case. It is moreover subgame perfect: from every history, the continuation profile is a Nash equilibrium of the continuation game.
\end{theorem}
\begin{proof}
By construction and \cref{prop:nilpotent}, the values $V$ computed in reverse topological order solve the evaluation equation of \cref{lem:policy} for $\sigma$, so $V=V^\sigma$. Each $\sigma_p$ is a Nash equilibrium of $\G_p[V^\sigma]$, so $\sigma$ is a global equilibrium by \cref{thm:consistency}. For subgame perfection, the continuation game after any history ending at place $p_\ell$ is the global game started at $p_\ell$, because stage payoffs and transitions depend only on the current place. Since $\sigma$ is stationary and no player gains by deviating from any place, its continuation is a Nash equilibrium there.
\end{proof}

With rework loops, $\mathcal{D}$ has cycles, but only inside its strongly connected components (SCCs). For a component $\mathcal{K}$ and fixed values $V_{\mathrm{out}}$ on the places downstream of $\mathcal{K}$, define the \emph{restricted operator} $\mathcal{B}_{\mathcal{K}}(\cdot\,;V_{\mathrm{out}}):\mathbb{R}^{\mathcal{K}}\to\mathbb{R}^{\mathcal{K}}$ by $\big(\mathcal{B}_{\mathcal{K}}(W;V_{\mathrm{out}})\big)(p)=\mathrm{val}\big(Q_1^p(W\oplus V_{\mathrm{out}})\big)$ for $p\in\mathcal{K}$, where $W\oplus V_{\mathrm{out}}$ agrees with $W$ on $\mathcal{K}$ and with $V_{\mathrm{out}}$ downstream. This yields an exact decomposition.

\begin{theorem}[SCC decomposition]\label{thm:scc}
Let $\mathcal{K}_1,\dots,\mathcal{K}_K$ be the SCCs of $\mathcal{D}$, listed in reverse topological order of the condensation of $\mathcal{D}$, the acyclic graph obtained by contracting each SCC to a single vertex (downstream components first). Then $V^\star|_{\mathcal{K}_k}$ is the unique fixed point of $\mathcal{B}_{\mathcal{K}_k}\big(\cdot\,;V^\star|_{\mathcal{K}_1\cup\dots\cup\mathcal{K}_{k-1}}\big)$. Each component can therefore be solved once, in this order. Singleton components without self-loops need a single local solve.
\end{theorem}
\begin{proof}
By \eqref{eq:Qmatrix}, $Q^p$ for $p\in\mathcal{K}_k$ depends only on $V$ at places reachable from $p$. These lie in $\mathcal{K}_k$ or in earlier-listed components. Hence $\mathcal{B}_{\mathcal{K}_k}(\cdot\,;V_{\mathrm{out}})$ is well defined and, by the argument of \cref{thm:contraction}, a $\gamma$-contraction (\cref{thm:contraction}), whose fixed point coincides with $V^\star$ on $\mathcal{K}_k$. Induction on $k$ completes the proof. A singleton $\{p\}$ without a self-loop has $Q^p$ that is independent of $V(p)$, so one solve suffices.
\end{proof}

\begin{algorithm}[ht]
\caption{SCC-decomposed solution of an OCGPN (zero-sum, \cref{ass:sequential})}
\label{alg:scc}
\begin{algorithmic}[1]
\State \textbf{Input:} game places $P_G$ with matrices $\{R^p,\Psi^p_q\}$, terminal values $v$, $\gamma\in(0,1)$, tolerance $\varepsilon$
\State Build the place graph $\mathcal{D}$; compute its SCCs $\mathcal{K}_1,\dots,\mathcal{K}_K$ in reverse topological order; set $V\gets0$ on $P_G$
\For{$k=1,\dots,K$}
  \If{$\mathcal{K}_k=\{p\}$ and $p\not\to p$}
    \State $V(p)\gets\mathrm{val}\big(Q^p(V)\big)$ by the LP \eqref{eq:lp}\Comment{one local solve}
  \Else
    \Repeat\ \ $V|_{\mathcal{K}_k}\gets\mathcal{B}_{\mathcal{K}_k}\big(V|_{\mathcal{K}_k};V_{\mathrm{out}}\big)$\Comment{synchronous sweep}
    \Until{the last sweep changed $V|_{\mathcal{K}_k}$ by at most $\varepsilon(1-\gamma)$ in $\|\cdot\|_\infty$}\Comment{certificate \eqref{eq:certificate}}
  \EndIf
\EndFor
\State \Return $V$, optimal local strategies $\sigma_p$ of $Q_1^p(V)$, and the routing matrix $P^\sigma$ by \eqref{eq:routing}
\end{algorithmic}
\end{algorithm}

\begin{proposition}[Complexity]\label{prop:complexity}
Let $L(p)$ be the cost of one local solve at $p$, which is polynomial in $|A_1^p|\,|A_2^p|$ for \eqref{eq:lp}, or in the tree size for an extensive-form local game. Let $V^0=0$ be the initial vector of Algorithm~\ref{alg:scc}, and let $\mathrm{res}_k=\|\mathcal{B}_{\mathcal{K}_k}(V^0;V_{\mathrm{out}})-V^0\|_\infty$ be the initial residual on $\mathcal{K}_k$. If downstream components are solved exactly, Algorithm~\ref{alg:scc} returns, on each component, $V$ with $\|V-V^\star\|_\infty\le\varepsilon$ after at most
\begin{equation}
  N_k=\begin{cases}1,&\mathcal{K}_k\ \text{a singleton without self-loop},\\[2pt]
  \max\Big\{1,\ \Big\lceil \dfrac{\log\big(\varepsilon(1-\gamma)/\mathrm{res}_k\big)}{\log\gamma}\Big\rceil+1\Big\},&\text{otherwise},\end{cases}
\end{equation}
sweeps, at total cost $\sum_{k}N_k\sum_{p\in\mathcal{K}_k}L(p)$. This cost is additive in the local game sizes. It excludes the cost of constructing the unfolded net, which may itself be large when color or binding spaces are large (\cref{ass:finite}).
\end{proposition}
\begin{proof}
By \cref{thm:scc}, the components are solved in order. A singleton without a self-loop costs one solve. Inside any other component, a synchronous sweep applies the $\gamma$-contraction $\mathcal{B}_{\mathcal{K}_k}$, so after $N$ sweeps the change in the last sweep is at most $\gamma^{N-1}\mathrm{res}_k$. Requiring $\gamma^{N-1}\mathrm{res}_k\le\varepsilon(1-\gamma)$ and applying \eqref{eq:certificate} to $\mathcal{B}_{\mathcal{K}_k}$ gives $N_k$. In-place (Gauss--Seidel) sweeps are also $\gamma$-contractions in $\|\cdot\|_\infty$ for this operator~\cite{puterman1994}, and they typically converge faster, but the bound above is stated for the synchronous update that the algorithm prescribes.
\end{proof}

\paragraph{Termination and probabilities of reaching outcomes.} Undiscounted questions use the same routing matrix. The routing matrix $P^\sigma$ is the transient-to-transient block (game places), and the exit matrix $\Xi^\sigma$ is the transient-to-absorbing block (terminal places). Then
\begin{equation}\label{eq:absorb}
  \Pr[\text{absorb in } z\mid\text{start at } p] = \big[(I-P^\sigma)^{-1}\Xi^\sigma\big]_{pz},\qquad
  \mathbb{E}[\text{visits to } q\mid p] = \big[(I-P^\sigma)^{-1}\big]_{pq},
\end{equation}
whenever the spectral radius $\rho(P^\sigma)$ is less than one.

\begin{proposition}[Strategic termination]\label{prop:termination}
Under \cref{ass:sequential} and a fixed profile $\sigma$, play terminates with probability one from every game place if and only if $\rho(P^\sigma)<1$. If, from some game place, play follows a cycle $p_1\to\cdots\to p_m\to p_1$ of the place graph with probability one, then $\rho(P^\sigma)=1$, $I-P^\sigma$ is singular, and the net livelocks.
\end{proposition}
\begin{proof}
The first statement is the standard characterization of absorbing Markov chains: $\sum_k (P^\sigma)^k$ converges if and only if $\rho(P^\sigma)<1$, and it then counts expected visits. If a cycle $p_1\to\cdots\to p_m\to p_1$ is followed with probability one, the corresponding principal submatrix of $P^\sigma$ is a permutation matrix, which has eigenvalue one. If, in addition, the net is bounded and each traversal restores the marking, the cycle's firing-count vector is a T-invariant by \cref{prop:invariants}(ii). This must be checked on $C$, since a traversal may consume resources or recolor tokens.
\end{proof}

\subsection{Composition and its matrix calculus}
\label{sec:composition}

Large strategic systems are built from small ones. This subsection makes composition precise and shows that each operator acts on the three matrix objects in a simple, predictable way. The aim is a calculus: properties established for the parts (invariants, equilibria, values, termination) combine into properties of the whole, without rebuilding the composite from scratch.

\paragraph{Open nets and interfaces.} Composition needs a notion of boundary.

\begin{definition}[Open OCGPN]\label{def:open}
An \emph{open} OCGPN is an OCGPN $\N$ together with two sets of places, the \emph{entry places} $E(\N)\subseteq P$ and the \emph{exit places} $X(\N)\subseteq P$. Each boundary place carries its color type. Tokens enter the net only at entry places and leave it only at exit places. The \emph{interface signature} of $\N$ is the pair of typed sets $\big(E(\N),X(\N)\big)$.
\end{definition}

When an open net is analyzed on its own, its exit places act as the terminal places of \cref{def:successor}, with terminal values supplied by whatever lies downstream. This is what makes the values of \cref{sec:values} usable as interface data.

\paragraph{Place fusion.} All six operators are built from one primitive, the identification of places, possibly together with a few added transitions.

\begin{definition}[Place fusion]\label{def:fusion}
Let $\N_1,\N_2$ have places $P_1,P_2$ and transitions $T_1,T_2$, and let $\iota\subseteq P_1\times P_2$ be a set of type-compatible pairs to be identified. The \emph{fusion} $\N_1+_\iota\N_2$ has place set $(P_1\sqcup P_2)/\!\sim_\iota$, where $\sim_\iota$ is the equivalence generated by $\iota$, transition set $T_1\sqcup T_2$, and the inherited arcs. Its \emph{fusion matrix} $\Lambda\in\{0,1\}^{|(P_1\sqcup P_2)/\sim_\iota|\times(|P_1|+|P_2|)}$ has $\Lambda_{[q],q'}=1$ if and only if $q'\in[q]$. Self-fusion of a single net, $\N+_\iota$, is defined in the same way with $\iota\subseteq P\times P$.
\end{definition}

\begin{lemma}[Fusion lemma]\label{lem:fusion}
Let $\N=\N_1+_\iota\N_2$, possibly with an additional set $T_0$ of new transitions whose incidence columns form a matrix $C_0$. Then
\begin{equation}\label{eq:fusionC}
  C=\Big[\,\Lambda\begin{pmatrix}C_1&0\\0&C_2\end{pmatrix}\ \ \Big|\ \ C_0\,\Big].
\end{equation}
Moreover:
\begin{enumerate}[label=(\roman*),nosep]
  \item a vector $y$ is a P-invariant of $\N$ if and only if $\Lambda^\top y=(y_1,y_2)$, where $y_1,y_2$ are P-invariants of $\N_1,\N_2$ that agree on every fused pair, and in addition $C_0^\top y=0$. The same holds with ``P-invariant'' replaced by ``Farkas vector'' ($C^\top y\le0$) and equalities by inequalities;
  \item if $n_1,n_2$ are T-invariants of $\N_1,\N_2$, then $(n_1,n_2,0)$ is a T-invariant of $\N$. The converse fails: $\N$ may have T-invariants $(n_1,n_2,n_0)$ whose parts are not T-invariants of the parts, because imbalances of the parts cancel on fused rows. Whether composition creates such invariants, and whether they are realizable, must be checked on $C$.
\end{enumerate}
\end{lemma}
\begin{proof}
A column of $C_k$ records token changes on places of $P_k$. After fusion, the change on a class $[q]$ is the sum of the changes on its members, which is the action of $\Lambda$. New transitions contribute their own columns, giving \eqref{eq:fusionC}. For (i), $y^\top C=\big((\Lambda^\top y)^\top\mathrm{diag}(C_1,C_2),\ y^\top C_0\big)$, and $\Lambda^\top y$ assigns to each original place the weight of its class. So the vector $\Lambda^\top y$ automatically agrees on fused pairs, and $y^\top C=0$ if and only if $y_1^\top C_1=0$, $y_2^\top C_2=0$, and $y^\top C_0=0$. The Farkas case is identical with inequalities. For (ii), $C(n_1,n_2,0)=\Lambda(C_1n_1,C_2n_2)=0$. Conversely, a T-invariant of $\N$ only requires $\Lambda(C_1n_1,C_2n_2)+C_0n_0=0$, which allows $C_kn_k\neq0$.
\end{proof}

\Cref{fig:fusion} shows the block structure for the most common case, in which an exit place of one net is fused with an entry place of another.

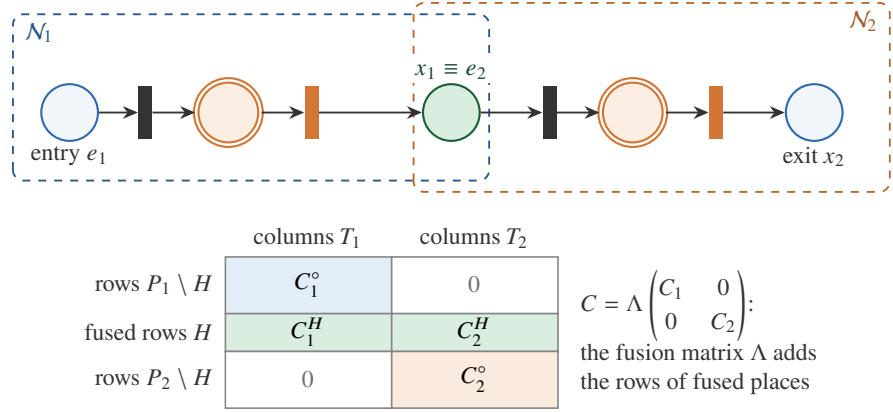
\begin{figure}[ht]
\centering
\begin{tikzpicture}[x=1cm,y=1cm,
  blk/.style={draw=inkcol!60,line width=0.6pt,font=\footnotesize,inner sep=0pt,anchor=north west}]
  \draw[line width=0.7pt,dashed,draw=rescol!70!black,rounded corners=4pt] (-0.2,-0.9) rectangle (6.1,1.3);
  \draw[line width=0.7pt,dashed,draw=gamecol!80!black,rounded corners=4pt] (5.1,-1.05) rectangle (11.4,1.45);
  \node[lbl,anchor=north west,text=rescol!70!black] at (-0.1,1.25) {$\N_1$};
  \node[lbl,anchor=north east,text=gamecol!70!black] at (11.3,1.4) {$\N_2$};
  \node[place,label={[lbl]below:entry $e_1$}] (e1) at (0.55,0) {};
  \node[transition] (t1) at (1.55,0) {};
  \node[gplace] (g1) at (2.65,0) {};
  \node[rtrans] (r1) at (3.75,0) {};
  \node[place,fill=okcol!18,draw=okcol!60!black,label={[lbl,text=okcol!45!black,fill=white,inner sep=1.5pt]above:$x_1\equiv e_2$}] (s) at (5.6,0) {};
  \node[transition] (t2) at (6.9,0) {};
  \node[gplace] (g2) at (8.0,0) {};
  \node[rtrans] (r2) at (9.1,0) {};
  \node[place,label={[lbl]below:exit $x_2$}] (x2) at (10.4,0) {};
  \foreach \u/\v in {e1/t1,t1/g1,g1/r1,r1/s,s/t2,t2/g2,g2/r2,r2/x2} \draw[arc] (\u)--(\v);
  \begin{scope}[shift={(2.6,-1.9)}]
    \def\cw{2.2}
    \node[lbl,anchor=south] at (1.1,0.05) {columns $T_1$};
    \node[lbl,anchor=south] at (3.3,0.05) {columns $T_2$};
    \node[blk,minimum width=2.2cm,minimum height=0.75cm,fill=rescol!14] at (0,0) {};
    \node[font=\footnotesize] at (1.1,-0.375) {$C_1^{\circ}$};
    \node[blk,minimum width=2.2cm,minimum height=0.75cm,fill=white] at (2.2,0) {};
    \node[font=\footnotesize,text=inkcol!70] at (3.3,-0.375) {$0$};
    \node[blk,minimum width=2.2cm,minimum height=0.5cm,fill=okcol!18] at (0,-0.75) {};
    \node[font=\footnotesize] at (1.1,-1.0) {$C_1^{H}$};
    \node[blk,minimum width=2.2cm,minimum height=0.5cm,fill=okcol!18] at (2.2,-0.75) {};
    \node[font=\footnotesize] at (3.3,-1.0) {$C_2^{H}$};
    \node[blk,minimum width=2.2cm,minimum height=0.75cm,fill=white] at (0,-1.25) {};
    \node[font=\footnotesize,text=inkcol!70] at (1.1,-1.625) {$0$};
    \node[blk,minimum width=2.2cm,minimum height=0.75cm,fill=gamecol!16] at (2.2,-1.25) {};
    \node[font=\footnotesize] at (3.3,-1.625) {$C_2^{\circ}$};
    \node[lbl,anchor=east] at (-0.1,-0.375) {rows $P_1\setminus H$};
    \node[lbl,anchor=east] at (-0.1,-1.0) {fused rows $H$};
    \node[lbl,anchor=east] at (-0.1,-1.625) {rows $P_2\setminus H$};
    \node[lbl,anchor=west,align=left] at (4.65,-1.0) {$C=\Lambda\begin{pmatrix}C_1&0\\0&C_2\end{pmatrix}$:\\the fusion matrix $\Lambda$ adds\\the rows of fused places};
  \end{scope}
\end{tikzpicture}
\caption{Composition by place fusion. The exit place $x_1$ of $\N_1$ is identified with the entry place $e_2$ of $\N_2$. The composite incidence matrix keeps the blocks of the parts and adds their rows on the fused place (green), as in \eqref{eq:fusionC}. Invariants of the parts that agree on the fused place therefore glue into invariants of the composite (\cref{lem:fusion}).}
\label{fig:fusion}
\end{figure}

\paragraph{The six operators.} With fusion available, the operators of \cref{fig:composition} have the following formal definitions.

\begin{definition}[Composition operators]\label{def:operators}
Let $\N,\N_1,\N_2,\N_3$ be open OCGPNs.
\begin{enumerate}[label=(\alph*),nosep]
  \item \textbf{Sequential composition} $\N_1;\N_2:=\N_1+_\iota\N_2$, where $\iota$ is a type-preserving bijection between $X(\N_1)$ and $E(\N_2)$; we write $\iota(x)$ for the entry place fused with the exit place $x$. Then $E(\N_1;\N_2)=E(\N_1)$ and $X(\N_1;\N_2)=X(\N_2)$.
  \item \textbf{Conditional choice} $\N_1\oplus_\omega\N_2$: a result place $q$ and two routing transitions $t_1,t_2$ with disjoint guards $[\omega\in O_1]$, $[\omega\in O_2]$, $O_1\cap O_2=\emptyset$, leading from $q$ to $E(\N_1)$ and $E(\N_2)$. The entry is $q$, and the exit is $X(\N_1)\cup X(\N_2)$.
  \item \textbf{Parallel composition} $\N_1\parallel_{H}\N_2:=\N_1+_\iota\N_2$, where $\iota$ identifies only a set $H$ of shared resource places. Boundaries are unions.
  \item \textbf{Synchronization} $\mathrm{join}(\N_1,\N_2;\N_3)$: a new transition $j$ whose inputs are one exit place of each of $\N_1,\N_2$ and whose output is an entry place of $\N_3$, with a guard that checks matching identifiers.
  \item \textbf{Feedback} $\mathrm{fb}_{x\to e}(\N):=\N+_\iota$ with $\iota=\{(x,e)\}$, $x\in X(\N)$, $e\in E(\N)$: a failed result re-enters an earlier game.
  \item \textbf{Dynamic instantiation}: a game place $p$ whose local game is drawn from a template family $\{\G_{p,\vartheta}\}_{\vartheta\in\Theta_p}$, with $\vartheta$ read from a field of the active token's color. In the unfolded net, $p$ is replaced by the instances $\{p_\vartheta\}_{\vartheta\in\Theta_p}$.
\end{enumerate}
\end{definition}

\begin{figure}[ht]
\centering
\begin{subfigure}[t]{0.48\textwidth}\centering
\begin{tikzpicture}
  \node[sub] (n1) at (0,0) {$\N_1$};
  \node[place,label={[lbl]above:interface}] (i) at (2.4,0) {};
  \node[ctok=okcol] at (i) {};
  \node[sub] (n2) at (4.8,0) {$\N_2$};
  \draw[arc] (n1)--node[above,lbl]{exit}(i); \draw[arc] (i)--node[above,lbl]{entry}(n2);
  \node[lbl,below=6mm of i] {colors must match: $\mathrm{Col}_{\mathrm{exit}}=\mathrm{Col}_{\mathrm{entry}}$};
\end{tikzpicture}
\caption{Sequential $\N_1;\N_2$}
\end{subfigure}\hfill
\begin{subfigure}[t]{0.48\textwidth}\centering
\begin{tikzpicture}[x=1cm,y=1cm]
  \node[place] (q) at (0,0) {}; \node[ctok=okcol] at (q) {};
  \coordinate (fan) at (0.7,0);
  \node[transition,label={[guard]above:$[\omega\in O_1]$}] (t1) at (1.5,0.8) {};
  \node[transition,label={[guard]below:$[\omega\in O_2]$}] (t2) at (1.5,-0.8) {};
  \node[sub,minimum width=1.2cm,minimum height=0.9cm] (n1) at (3.1,0.8) {$\N_1$};
  \node[sub,minimum width=1.2cm,minimum height=0.9cm] (n2) at (3.1,-0.8) {$\N_2$};
  \draw[arc] (q.east)--(fan)|-(t1); \draw[arc] (q.east)--(fan)|-(t2);
  \draw[arc] (t1)--(n1); \draw[arc] (t2)--(n2);
\end{tikzpicture}
\caption{Choice $\N_1\oplus_\omega\N_2$}
\end{subfigure}

\vspace{5mm}
\begin{subfigure}[t]{0.48\textwidth}\centering
\begin{tikzpicture}
  \node[gplace] (g1) at (-0.5,1.0) {}; \node[rtrans] (r1) at (1.3,1.0) {}; \node[place] (o1) at (2.6,1.0) {};
  \node[gplace] (g2) at (-0.5,-1.0) {}; \node[rtrans] (r2) at (1.3,-1.0) {}; \node[place] (o2) at (2.6,-1.0) {};
  \node[place,tokens=3,label={[lbl]right:shared resource}] (sr) at (1.3,0) {};
  \draw[arc] (g1)--(r1); \draw[arc] (r1)--(o1); \draw[arc] (g2)--(r2); \draw[arc] (r2)--(o2);
  \draw[arc] (sr)--(r1.south); \draw[arc] (sr)--(r2.north);
  \node[glbl] at (-1.4,1.0) {$\N_1$}; \node[glbl] at (-1.4,-1.0) {$\N_2$};
\end{tikzpicture}
\caption{Parallel $\N_1\parallel_H\N_2$ with a shared place}
\end{subfigure}\hfill
\begin{subfigure}[t]{0.48\textwidth}\centering
\begin{tikzpicture}[x=1cm,y=1cm]
  \node[place,label={[lbl]left:result 1}] (a) at (0,0.8) {}; \node[ctok=okcol] at (a) {};
  \node[place,label={[lbl]left:result 2}] (b) at (0,-0.8) {}; \node[ctok=okcol] at (b) {};
  \node[transition] (j) at (1.6,0) {};
  \node[gplace,label={[glbl]right:joint game}] (g) at (3.9,0) {};
  \draw[arc] (a)-|(j); \draw[arc] (b)-|(j);
  \draw[arc] (j) -- node[below,guard]{$[\mathrm{id}_1{=}\mathrm{id}_2]$} (g);
\end{tikzpicture}
\caption{Synchronization (join)}
\end{subfigure}

\vspace{5mm}
\begin{subfigure}[t]{0.48\textwidth}\centering
\begin{tikzpicture}
  \node[gplace,label={[glbl]below:$\G_1$}] (g) at (0,0) {};
  \node[rtrans] (r) at (1.3,0) {};
  \node[place] (o) at (2.6,0) {};
  \node[transition,label={[guard]right:$[\mathrm{fail}]$}] (f) at (2.6,1.3) {};
  \node[transition,label={[guard]above:$[\mathrm{pass}]$}] (s) at (3.9,0) {};
  \draw[arc] (g)--(r); \draw[arc] (r)--(o); \draw[arc] (o)--(s); \draw[arc] (o)--(f);
  \draw[arc,badcol] (f.north) -- ++(0,0.35) -| node[above,lbl,pos=0.25,text=badcol]{revision token} (g.north);
  \draw[arc] (s) -- ++(0.8,0);
\end{tikzpicture}
\caption{Feedback / repair loop}
\end{subfigure}\hfill
\begin{subfigure}[t]{0.48\textwidth}\centering
\begin{tikzpicture}
  \node[transition] (t) at (0,0) {};
  \node[gplace,label={[glbl]below:template $\G_\vartheta$}] (g) at (2.8,0) {};
  \node[ctok=gamecol] at (g) {};
  \draw[arc] (t) -- node[above,lbl]{$\tok{\mathrm{type},\vartheta,\mathrm{id}}$} (g);
  \node[cbox,font=\scriptsize,align=left] (c) at (2.8,1.4) {$\vartheta$ sets players,\\ actions, payoffs};
  \draw[darc,draw=gamecol!80!black] (c)--(g);
\end{tikzpicture}
\caption{Dynamic game instantiation}
\end{subfigure}
\caption{The six composition operators of \cref{def:operators}. Panels (a)--(b) are standard workflow constructs. Panels (c)--(d) are where Petri nets do more than game trees: true concurrency and joins over shared resources. Panels (e)--(f) add loops and outcome-parameterized games.}
\label{fig:composition}
\end{figure}
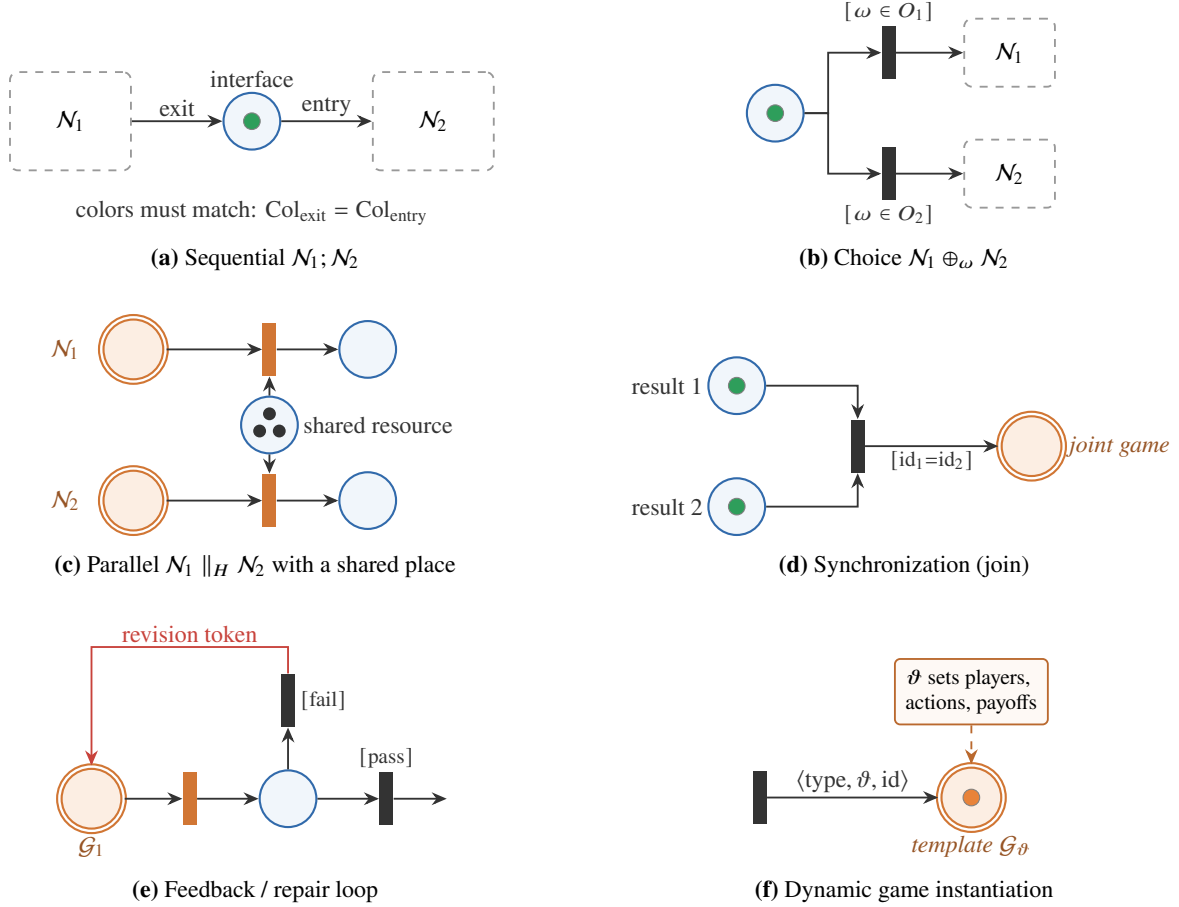

\Cref{tab:calculus} summarizes how each operator acts on the three matrix objects. The rest of the subsection proves these entries and draws their consequences.

\begin{table}[ht]
\centering
\caption{The composition calculus at a glance: the effect of each operator on the incidence matrix, the local matrices, and the place graph and routing matrix.}
\label{tab:calculus}
\footnotesize
\begin{tabular}{@{}>{\raggedright\arraybackslash}p{2.1cm}>{\raggedright\arraybackslash}p{3.7cm}>{\raggedright\arraybackslash}p{3.6cm}>{\raggedright\arraybackslash}p{4.6cm}@{}}
\toprule
Operator & Incidence matrix $C$ & Local matrices $(R^p,\Phi^p)$ & Place graph and routing $P^\sigma$\\\midrule
sequential & fusion of exit and entry rows & unchanged & block upper triangular; downstream enters through entry values\\
choice & two new guarded columns & one indicator term per branch at the choosing place & a fork in the place graph\\
parallel & block diagonal, shared rows added & unchanged unless shared places are read & product of the parts (concurrent state)\\
synchronization & one column with several inputs & unchanged & join of two branches; success probability multiplies\\
feedback & self-fusion; possible new T-invariants, to be checked & unchanged & a back edge, enlarging an SCC\\
instantiation & one row/column set per parameter & one local game per parameter & instances replace the template place\\
\bottomrule
\end{tabular}
\end{table}

\begin{theorem}[Matrix calculus of composition]\label{thm:composition}
Let $\N_1,\N_2$ be open OCGPNs with incidence matrices $C_1,C_2$, local matrices $\{R^p,\Phi^p_\omega\}$, and, for stationary profiles, routing matrices $P_1^\sigma,P_2^\sigma$.
\begin{enumerate}[label=(\roman*),nosep]
  \item \textbf{Parallel.} Order the rows of $\N_1\parallel_H\N_2$ as $(P_1\setminus H,\ H,\ P_2\setminus H)$, unfolded by color, and write $C_k^{\circ}$ and $C_k^{H}$ for the rows of $C_k$ on the private places $P_k\setminus H$ and on the shared places $H$. Then
  \begin{equation}\label{eq:parallelC}
    C=\begin{pmatrix} C_1^{\,\circ} & 0\\ C_1^{\,H} & C_2^{\,H}\\ 0 & C_2^{\,\circ}\end{pmatrix}.
  \end{equation}
  If $y_1,y_2$ are P-invariants (respectively, Farkas vectors) of the parts that agree on $H$, then the glued vector $y$ is a P-invariant (respectively, satisfies $C^\top y\le0$) of the composition.
  \item \textbf{Sequential.} In $\N_1;\N_2$, with game places ordered $(\N_1,\N_2)$ and blocks indexed by subnet in brackets,
  \begin{equation}\label{eq:seqP}
    P^\sigma=\begin{pmatrix}P_{[11]}&P_{[12]}\\0&P_{[22]}\end{pmatrix},\qquad
    V_{[2]}^\sigma=(I-\gamma P_{[22]})^{-1}\mathbf{r}_{[2]},\qquad
    V_{[1]}^\sigma=(I-\gamma P_{[11]})^{-1}\big(\mathbf{r}_{[1]}+\gamma P_{[12]}V_{[2]}^\sigma\big),
  \end{equation}
  Hence $\N_2$ affects $\N_1$ only through its \emph{entry values} $P_{[12]}V_{[2]}^\sigma$.
  \item \textbf{Choice.} At the choosing place $p$, the branches enter \eqref{eq:Qmatrix} additively, one term $\gamma V(f(p,\omega))\Phi^p_\omega$ per outcome. \Cref{thm:dominance} applies branch by branch.
  \item \textbf{Synchronization.} A join is a single column with several input rows. A P-invariant must assign the output place the sum of the weights of the joined places.
  \item \textbf{Feedback.} A feedback fusion adds an edge to the place graph and may merge components into one SCC. \Cref{prop:nilpotent} no longer applies to that component, but \cref{thm:contraction,thm:scc} do. Feedback creates a cyclic \emph{dependency}. It does not by itself create a T-invariant: a traversal may consume resources or recolor tokens, and a T-invariant must be verified on the composite incidence matrix (\cref{lem:fusion}(ii)). Termination is decided by the routing matrix (\cref{prop:termination}).
\end{enumerate}
\end{theorem}
\begin{proof}
(i) This is \cref{lem:fusion} with $\iota$ identifying $H$ and no new transitions: the row blocks of $\Lambda\,\mathrm{diag}(C_1,C_2)$ are those of \eqref{eq:parallelC}, and the gluing statement is \cref{lem:fusion}(i). (ii) Exit places of $\N_1$ are entry places of $\N_2$, and no arc leads back. So no game place of $\N_2$ has a successor in $\N_1$, and the lower-left block of $P^\sigma$ vanishes. Block back-substitution in $(I-\gamma P^\sigma)V=\mathbf{r}$ gives \eqref{eq:seqP}: the inverse in the second equation multiplies the whole right-hand side, including the downstream term. (iii) This is linearity of \eqref{eq:Qmatrix} in the sum over $\omega$. (iv) If the join column $c$ consumes from $q_1,\dots,q_\ell$ and produces in $q$, then $y^\top c=y_q-\sum_k y_{q_k}$, which must vanish. (v) The fused pair $(x,e)$ creates an edge from every game place that can reach $x$ to every game place reachable from $e$, so every place on a resulting cycle of $\mathcal{D}$ lies in one SCC. The remaining statements are \cref{lem:fusion}(ii) and \cref{prop:termination}.
\end{proof}

\paragraph{Sequential composition and equilibria.} \Cref{thm:composition}(ii) concerns a fixed profile. The following result shows that equilibria themselves compose, also in the general-sum case.

\begin{theorem}[Compositional equilibrium]\label{thm:seqeq}
Let $\N=\N_1;\N_2$ satisfy \cref{ass:normalform,ass:sequential,ass:discount}. Fix a global equilibrium $\sigma_{[2]}$ of $\N_2$, with value vectors $V^{\sigma_{[2]}}$, and let $\N_1[\sigma_{[2]}]$ be the net $\N_1$ in which each exit place $z$ is terminal with value $v_i(z):=V_i^{\sigma_{[2]}}(\iota(z))$, the value of the entry place of $\N_2$ fused with $z$. Then $(\sigma_{[1]},\sigma_{[2]})$ is a global equilibrium of $\N$ if and only if $\sigma_{[1]}$ is a global equilibrium of $\N_1[\sigma_{[2]}]$.
\end{theorem}
\begin{proof}
By \cref{thm:consistency}, $(\sigma_{[1]},\sigma_{[2]})$ is a global equilibrium of $\N$ if and only if, at every game place, the local profile is an equilibrium of the augmented game at $V^{(\sigma_{[1]},\sigma_{[2]})}$. At places of $\N_2$, the augmented games involve only values at places of $\N_2$ and its own exits, because no edge leads back to $\N_1$. So these conditions hold exactly because $\sigma_{[2]}$ is a global equilibrium of $\N_2$, and $V^{(\sigma_{[1]},\sigma_{[2]})}$ coincides with $V^{\sigma_{[2]}}$ there. At places of $\N_1$, the augmented games involve values at places of $\N_1$ and at the entry places of $\N_2$, and the latter equal the terminal values of $\N_1[\sigma_{[2]}]$. By \eqref{eq:seqP}, the values at places of $\N_1$ also coincide with those of $\N_1[\sigma_{[2]}]$ under $\sigma_{[1]}$. So the conditions at places of $\N_1$ are exactly the conditions of \cref{thm:consistency} for $\sigma_{[1]}$ in $\N_1[\sigma_{[2]}]$.
\end{proof}

\Cref{thm:seqeq} is an \emph{interface contract}. A subnet is summarized, for everything upstream, by the equilibrium values it offers at its entry places. Downstream equilibrium selection is a parameter of the upstream problem, which makes the dependence of upstream behavior on downstream selection explicit. Variant B of running example~I (\cref{sec:re1}) is an instance: the planning game sees $\G_A$ and $\G_R$ only through $v_A$ and $v_R$. This is also the formal counterpart of grounding in multi-resolution games~\cite{yang2025general}.

\paragraph{Conditional choice and comparative statics.} Recall the routing probability $\Psi^p_q$ of \eqref{eq:Psi}.

\begin{proposition}[Monotone comparative statics of the interface]\label{prop:monotone}
Let $V'$ differ from $V$ only in player $i$'s value at one successor $q$ of $p$, with $V_i'(q)\ge V_i(q)$. If $D_{a_{-i}}(\Psi^p_q)\ge0$ for every $a_{-i}$, that is, if $a_i^+$ never lowers the probability of reaching $q$, then
\begin{equation}
  D_{a_{-i}}\big(Q_i^p(V')\big)-D_{a_{-i}}\big(Q_i^p(V)\big)=\gamma\big(V_i'(q)-V_i(q)\big)\,D_{a_{-i}}(\Psi^p_q)\ \ge\ 0 .
\end{equation}
In particular, if $a_i^+$ strictly dominates $a_i^-$ in $\G_p[V]$, it still does in $\G_p[V']$.
\end{proposition}
\begin{proof}
Subtract \eqref{eq:Qmatrix} for $V$ from \eqref{eq:Qmatrix} for $V'$. Only the outcomes routed to $q$ contribute, and $D_{a_{-i}}$ is linear.
\end{proof}

\Cref{prop:monotone} says when improving a branch helps. Raising the value of a successor encourages exactly the actions that make that successor more likely. If an action has zero pivotality for $q$ against some $a_{-i}$, as in variant B of running example~I against $\act{L}$, improving $q$ cannot change that comparison.

\paragraph{Feedback and loop reduction.} A feedback loop turns a backward recursion into a fixed point. For a component entered through one place, the fixed point reduces to a scalar formula.

\begin{proposition}[Loop reduction]\label{prop:loop}
Fix a stationary profile $\sigma$ and a game place $p$, for example the entry of a feedback loop. Let $\tau^+\ge1$ be the first time after $0$ at which play returns to $p$ or leaves the component of $p$. Define
\begin{equation}
  \mathcal{U}_i(p)=\Ex^\sigma_p\Big[\sum_{\ell=0}^{\tau^+-1}\gamma^\ell R_i^{p_\ell}(a_\ell)+\gamma^{\tau^+} \overline V_i^\sigma(p_{\tau^+})\ind\{p_{\tau^+}\neq p\}\Big],\qquad
  \mathcal{R}(p)=\Ex^\sigma_p\big[\gamma^{\tau^+}\ind\{p_{\tau^+}=p\}\big].
\end{equation}
Then $\mathcal{R}(p)\le\gamma<1$ and
\begin{equation}\label{eq:loop}
  V_i^\sigma(p)=\frac{\mathcal{U}_i(p)}{1-\mathcal{R}(p)} .
\end{equation}
For a single place with a self-loop, $\mathcal{R}(p)=\gamma P^\sigma_{pp}$ and $\mathcal{U}_i(p)=\mathbf{r}_i^\sigma(p)+\gamma\sum_{q\neq p}P^\sigma_{pq}V_i^\sigma(q)$.
\end{proposition}
\begin{proof}
Decompose the payoff \eqref{eq:globalpayoff} at time $\tau^+$. If play has left the component of $p$, the remaining payoff is $\overline V_i^\sigma(p_{\tau^+})$ by stationarity. If play has returned to $p$, the remaining payoff is $V_i^\sigma(p)$, discounted by $\gamma^{\tau^+}$. Taking expectations gives $V_i^\sigma(p)=\mathcal{U}_i(p)+\mathcal{R}(p)V_i^\sigma(p)$. Since $\tau^+\ge1$, $\mathcal{R}(p)\le\gamma<1$, so we may solve for $V_i^\sigma(p)$. The single-place case is $\tau^+=1$.
\end{proof}

\Cref{eq:loop} is the renewal form of \eqref{eq:neumann}. Its denominator $1-\mathcal{R}(p)$ measures how strongly the loop discounts repeated attempts. As $\mathcal{R}(p)\to1$, the value is dominated by repetition, which is the quantitative face of the livelock in \cref{prop:termination}. The value of the lateral game in running example~II, \eqref{eq:VL}, is the single-place case.

\paragraph{Synchronization.} A join turns independent branches into a conjunction.

\begin{proposition}[Fork--join probabilities]\label{prop:forkjoin}
Consider $\mathrm{join}(\N_1,\N_2;\N_3)$ preceded by a fork that starts $\N_1$ and $\N_2$ together. Assume the branches share no places, and their games are resolved independently under stationary profiles. Let $\theta_k$ be the probability that branch $k$ ends in the pass place feeding the join, computed by \eqref{eq:absorb} on the branch's own routing matrix. Then the join fires with probability $\theta_1\theta_2$. Exactly one pass token is left without its partner with probability $\theta_1(1-\theta_2)+\theta_2(1-\theta_1)$.
\end{proposition}
\begin{proof}
By independence, the pair of absorbing outcomes of the two branches has the product distribution. The join fires if and only if both branches end in their pass places. Exactly one pass token waits forever if and only if exactly one branch passes.
\end{proof}

The second probability quantifies the structural defect found in running example~I (\cref{sec:re1-composition}): a pass certificate stranded by its partner's failure. It also shows that synchronization is a strategic bottleneck. If each branch's players choose effort, a branch's contribution to joint success is scaled by the other branch's pass probability. This is a complementarity that a single-branch analysis misses.

\paragraph{Parallel composition.} Without shared places, parallel subnets can be analyzed separately.

\begin{proposition}[Independent parallel composition]\label{prop:parallel}
Let $\N=\N_1\parallel_\emptyset\N_2$, with additively separable payoffs, where each subnet discounts its own stages (separate clocks), and with any scheduler that eventually fires every enabled transition. If $\sigma_k$ is a global equilibrium of $\N_k$ for $k=1,2$, then playing $\sigma_1$ in $\N_1$ and $\sigma_2$ in $\N_2$ is an equilibrium of $\N$. Each player's payoff is $U_i=U_i^{(1)}+U_i^{(2)}$, and the probabilities of the pairs of terminal places factorize.
\end{proposition}
\begin{proof}
With no shared places, the firings of $\N_1$ neither enable nor disable those of $\N_2$. The scheduler changes only the interleaving, not the sequence of stages within each subnet. Under separate clocks, each subnet's payoff stream is therefore distributed as in isolation. A player who deviates in one subnet changes only that subnet's payoff, and gains nothing there because $\sigma_k$ is an equilibrium. Factorization follows from independence.
\end{proof}

With shared places, \cref{thm:composition}(i) still glues structural invariants. Strategically, however, the game played at one place may depend on the marking of the shared places, as in the coupling through security-operations-center (SOC) capacity in running example~II (\cref{sec:re2}). \Cref{ass:sequential} then fails, and values must be computed on the product state (\cref{rem:beyond}).

\paragraph{Dynamic instantiation and state size.} Instantiation is the operator through which earlier outcomes reshape later games, such as the defect count in running example~I or a stolen credential in running example~II.

\begin{proposition}[Local cost of instantiation]\label{prop:instantiation}
Suppose an artifact with values in $\Theta$ is produced at a game place $p_{\mathrm{prod}}$ and read at a game place $p_{\mathrm{use}}$. Let the \emph{artifact window} $\mathcal{W}\subseteq P_G$ be the set of game places that lie on a path of the place graph from $p_{\mathrm{prod}}$ to $p_{\mathrm{use}}$, excluding $p_{\mathrm{prod}}$. Then the unfolded place graph of the Markov-state analysis has
\begin{equation}
  |P_G|+(|\Theta|-1)\,|\mathcal{W}|
\end{equation}
game places. By contrast, an encoding that carries the artifact in the state at every game place has $|P_G|\,|\Theta|$ states.
\end{proposition}
\begin{proof}
Under \cref{ass:sequential}, the artifact token is folded into the active token's color from its production until its use. So exactly the places in $\mathcal{W}$ must distinguish its $|\Theta|$ values, and every other place keeps a single instance.
\end{proof}

In running example~II (\cref{sec:re2}), the credential is produced at $p_L$ and read at the adjacent $p_A$, so $\mathcal{W}=\{p_A\}$, and only the asset game doubles. The artifact is still part of the Markov state where it matters: colored tokens do not remove the information an exact analysis needs, but they \emph{localize} it to the artifact window.

\paragraph{A compositional analysis procedure.} Taken together, these results give a procedure for analyzing a composite net from its parts.
\begin{enumerate}[label=(\arabic*),nosep]
  \item Analyze each part structurally. Compute its P-invariants and Farkas certificates, and glue them across fused places by \cref{lem:fusion}.
  \item Order the parts by the condensation of the place graph. Solve the downstream parts first, and pass their entry values upstream (\cref{thm:seqeq}).
  \item Within a part, solve loops by \cref{thm:contraction} or \cref{prop:loop}, joins by \cref{prop:forkjoin}, and independent parallel parts separately (\cref{prop:parallel}).
  \item Use \cref{prop:monotone} and \cref{thm:dominance} to redesign interfaces where the selected equilibrium is undesirable, and \cref{prop:instantiation} to control the cost of carrying artifacts.
\end{enumerate}

\section{Running example I: strategic maintenance}
\label{sec:re1}

The first running example is cooperative in structure: all players are on the system's side, but their incentives conflict. It illustrates each part of \cref{sec:theory} in turn. The central lesson is that the \emph{shape} of the outcome map, not the size of downstream rewards, decides whether incentives can be fixed.

\subsection{Model}
\label{sec:re1-model}

A manufacturing machine shows an anomalous vibration signature. A plant operator $O$ and a contracted maintenance provider $S$ meet at a \emph{planning game} $\G_P$. Each chooses high preparation $\act{H}$ (full diagnostic effort, accurate telemetry, a reserved production window) or low preparation $\act{L}$ (cutting effort). High preparation is costly, but joint high preparation improves the readiness of the repair. The stage payoffs $(R_O^P,R_S^P)$ to the operator and the provider are those in \cref{tab:stage}. For these payoffs alone, $\act{L}$ strictly dominates $\act{H}$ for each player: against $\act{H}$, low effort yields $3>2$; against $\act{L}$, it yields $0>-1$. So $(\act{L},\act{L})$ is the unique \emph{myopic} equilibrium, a prisoner's dilemma.

\begin{table}[ht]
\centering
\caption{Immediate (stage) payoffs $R^P$ of the planning game $\G_P$. The numbers are illustrative, not empirically estimated.}
\label{tab:stage}
\begin{tabular}{c|cc}
\toprule
Operator $\backslash$ Provider & $\act{H}$ & $\act{L}$\\\midrule
$\act{H}$ & $(2,2)$ & $(-1,3)$\\
$\act{L}$ & $(3,-1)$ & $(0,0)$\\\bottomrule
\end{tabular}
\end{table}

We consider two variants that share this stage game and differ only in \emph{what the outcome of $\G_P$ unlocks downstream}.

\paragraph{Variant A: graded routing.} Let $k\in\{0,1,2\}$ be the number of players who choose $\act{L}$. The outcome token is $\tok{\mathrm{job},\mathrm{evidence},k}$. If $k=0$, transition $t_V$ routes it to a joint verification game $\G_V$. Otherwise $t_R$ routes it to a corrective-recovery game instantiated with the defect count $k$, which we denote $\G_R^{(k)}$ (\cref{tab:routingA}). This is dynamic game instantiation with parameter $\vartheta=k$. After verification, an approval transition $t_a$ leads to the guarded execution stage. In $\G_V$, both parties choose \emph{Complete} ($d_i=0$) or \emph{Delay} ($d_i=1$). In $\G_R$, they choose \emph{Remediate} ($d_i=0$) or \emph{Stall} ($d_i=1$). As an illustrative specification, let
\begin{equation}\label{eq:downstreamA}
  u_i^{V}(d) = 4 - 3d_i - d_{-i},\qquad
  u_i^{R}(d;k) = -2k - 3d_i - d_{-i}.
\end{equation}

\begin{table}[ht]
\centering
\caption{Outcome-guarded routing in variant A.}
\label{tab:routingA}
\begin{tabular}{lll}
\toprule
Current game & Outcome guard & Transition / next game\\\midrule
$\G_P$ & $(\act{H},\act{H})$, $k=0$ & $t_V\to\G_V$; token $\tok{\mathrm{job},\mathrm{evidence},k{=}0}$\\
$\G_P$ & $(\act{H},\act{L})$ or $(\act{L},\act{H})$, $k=1$ & $t_R\to\G_R$; token $\tok{\mathrm{job},\mathrm{incomplete},k{=}1}$\\
$\G_P$ & $(\act{L},\act{L})$, $k=2$ & $t_R\to\G_R$; token $\tok{\mathrm{job},\mathrm{incomplete},k{=}2}$\\
$\G_V$ & verification complete & $t_a\to$ approval gate; execution only if safe\\
$\G_R$ & remediation complete & terminal (rework ends the episode)\\\bottomrule
\end{tabular}
\end{table}

\paragraph{Variant B: binary routing with a rework loop.} Now the outcome is $\omega_P=\mathrm{ready}$ only under $(\act{H},\act{H})$; all other profiles give $\omega_P=\mathrm{notReady}$. A ready token activates an \emph{authorization game} $\G_A$; a not-ready token activates a recovery game $\G_R$, whose repair token returns to $\G_P$ (\cref{fig:examplenet}). In $\G_A$, the operator and an independent verifier have distinct objectives. The operator may request execution. The verifier authorizes only if a machine-state predicate $\mathrm{Safe}(\zeta)$ and the required provenance certificate are both satisfied. Execution is controlled by two separate mechanisms. The authorization token produced by $r_A$ is in the \emph{preset} of the execution transition $t_E$, so execution needs an available authorization. The evidence predicate $\mathrm{Safe}(\zeta)\wedge\mathrm{cert}(\zeta)$, where $\zeta$ is the evidence record bound to $t_E$, is a \emph{guard} on $t_E$, so execution also needs valid evidence; figures abbreviate it as $[\mathrm{Safe}\wedge\mathrm{cert}]$. The verifier's incentives determine whether authorization is granted, but they cannot bypass either mechanism. In $\G_R$, the operator and provider renegotiate corrective actions, additional diagnostics, or task reassignment. \Cref{sec:re1-strategic} first compares incentives at $\G_P$ for \emph{prescribed} continuation values $v_A$ and $v_R$, which is a local illustration. \Cref{sec:re1-values} then solves the cyclic global game with $\gamma<1$, as \cref{ass:discount} requires.

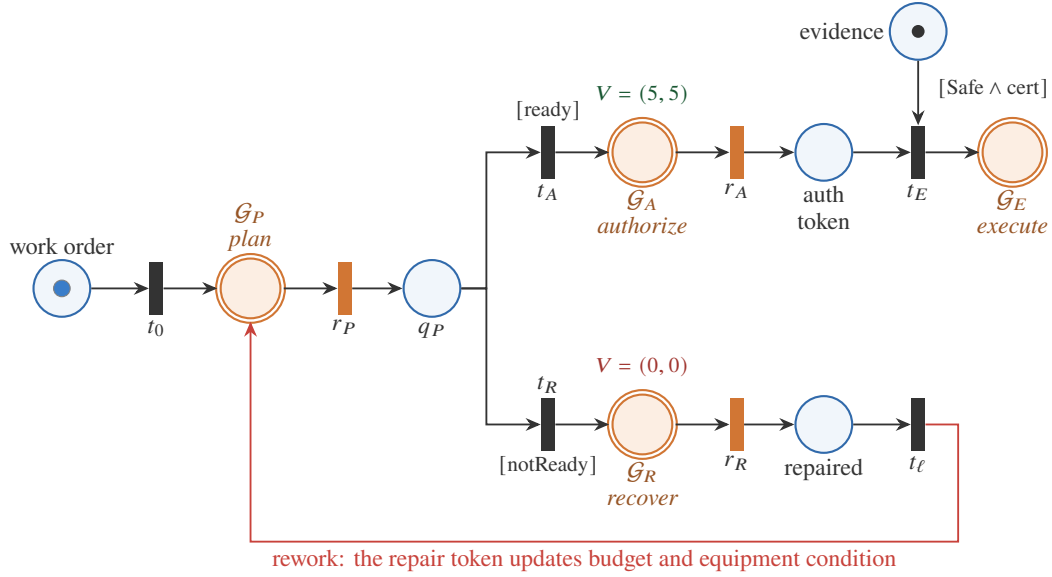
\begin{figure}[ht]
\centering
\begin{tikzpicture}[x=0.96cm,y=1cm]
  \node[place,label={[lbl]above:work order}] (s0) at (0,0) {}; \node[ctok=rescol] at (s0) {};
  \node[transition,label={[lbl]below:$t_0$}] (t0) at (1.3,0) {};
  \node[gplace,label={[glbl]above:$\G_P$\\[-1pt]plan}] (gp) at (2.6,0) {};
  \node[rtrans,label={[lbl]below:$r_P$}] (rp) at (3.9,0) {};
  \node[place,label={[lbl]below:$q_P$}] (op) at (5.1,0) {};
  \coordinate (fan) at (5.85,0);
  \node[transition,label={[guard]above:$[\mathrm{ready}]$},label={[lbl]below:$t_A$}] (ta) at (6.7,1.8) {};
  \node[gplace,label={[glbl]below:$\G_A$\\[-1pt]authorize}] (ga) at (8.0,1.8) {};
  \node[rtrans,label={[lbl]below:$r_A$}] (ra) at (9.3,1.8) {};
  \node[place,label={[lbl,align=center]below:auth\\[-1pt]token}] (au) at (10.5,1.8) {};
  \node[transition,label={[lbl]below:$t_E$}] (te) at (11.8,1.8) {};
  \node[gplace,label={[glbl]below:$\G_E$\\[-1pt]execute}] (ge) at (13.1,1.8) {};
  \node[place,tokens=1,label={[lbl]left:evidence}] (ev) at (11.8,3.4) {};
  \node[guard,anchor=west] at (12.0,2.65) {$[\mathrm{Safe}\wedge\mathrm{cert}]$};
  \node[font=\scriptsize,text=okcol!55!black,anchor=south] at (8.0,2.3) {$V=(5,5)$};
  \node[transition,label={[guard]below:$[\mathrm{notReady}]$},label={[lbl]above:$t_R$}] (tr) at (6.7,-1.8) {};
  \node[gplace,label={[glbl]below:$\G_R$\\[-1pt]recover}] (gr) at (8.0,-1.8) {};
  \node[rtrans,label={[lbl]below:$r_R$}] (rr) at (9.3,-1.8) {};
  \node[place,label={[lbl]below:repaired}] (rep) at (10.5,-1.8) {};
  \node[transition,label={[lbl]below:$t_\ell$}] (tl) at (11.8,-1.8) {};
  \node[font=\scriptsize,text=badcol!80!black,anchor=south] at (8.0,-1.3) {$V=(0,0)$};
  \foreach \u/\v in {s0/t0,t0/gp,gp/rp,rp/op,ta/ga,ga/ra,ra/au,au/te,te/ge,tr/gr,gr/rr,rr/rep,rep/tl}
    \draw[arc] (\u)--(\v);
  \draw[arc] (op.east) -- (fan) |- (ta);
  \draw[arc] (op.east) -- (fan) |- (tr);
  \draw[arc] (ev) -- (te);
  \draw[arc,badcol] (tl.east) -- ++(0.45,0) |- (2.6,-3.35) -- (gp.south);
  \node[lbl,text=badcol,anchor=north] at (7.2,-3.4) {rework: the repair token updates budget and equipment condition};
\end{tikzpicture}
\caption{The OCGPN of variant B. The success path (top) needs both the \emph{ready} outcome and a guard that checks independent evidence, so incentives alone cannot fire $t_E$. The recovery path (bottom) closes a rework loop back to $\G_P$.}
\label{fig:examplenet}
\end{figure}

\paragraph{Semantics (\cref{sec:semantics}).} \Cref{fig:trace} traces one successful execution of variant B through the five kinds of object that the semantics keeps apart: (a) \emph{choices}, (b) \emph{outcomes}, (c) \emph{token production}, (d) \emph{enabled transitions}, and (e) \emph{realized game instances}. Keeping them distinct avoids the misleading claim that an outcome alone fires an arbitrary transition even when the required resources are absent.

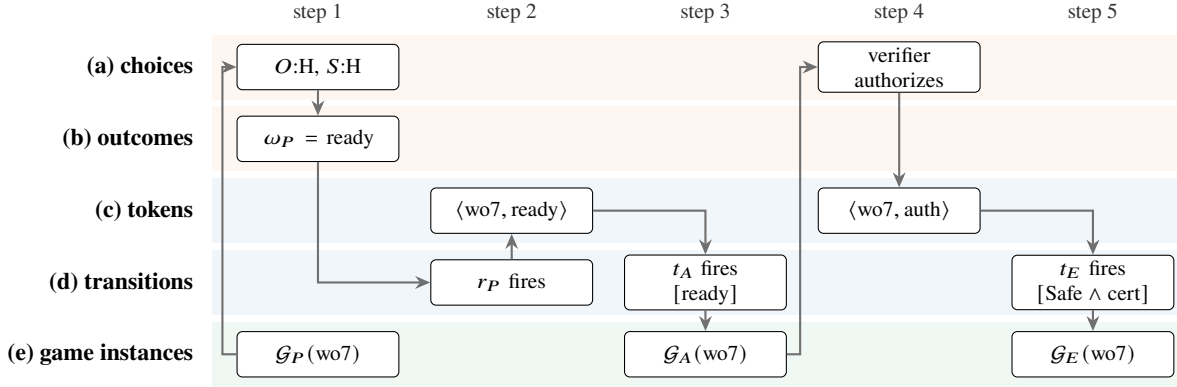
\begin{figure}[ht]
\centering
\begin{tikzpicture}[x=2.33cm,y=0.95cm,
  ev/.style={draw,rounded corners=2pt,fill=white,font=\scriptsize,align=center,inner sep=2pt,minimum height=6mm,text width=2.0cm}]
  \foreach \y/\n/\c in {0/{(a) choices}/gamecol,-1/{(b) outcomes}/gamecol,-2/{(c) tokens}/rescol,-3/{(d) transitions}/rescol,-4/{(e) game instances}/okcol}{
    \fill[\c!7] (-0.05,\y-0.45) rectangle (5.45,\y+0.45);
    \node[font=\footnotesize\bfseries,anchor=east] at (-0.1,\y) {\n};
  }
  \foreach \x/\s in {0.55/1,1.65/2,2.75/3,3.85/4,4.95/5} \node[font=\scriptsize,text=inkcol] at (\x,0.75) {step \s};
  \node[ev] (a1) at (0.55,0) {$O{:}\act{H}$,\ $S{:}\act{H}$};
  \node[ev] (b1) at (0.55,-1) {$\omega_P=\mathrm{ready}$};
  \node[ev] (c2) at (1.65,-2) {$\tok{\mathrm{wo}7,\mathrm{ready}}$};
  \node[ev] (d2) at (1.65,-3) {$r_P$ fires};
  \node[ev] (d3) at (2.75,-3) {$t_A$ fires\\ $[\mathrm{ready}]$};
  \node[ev] (e3) at (2.75,-4) {$\G_A(\mathrm{wo}7)$};
  \node[ev] (a4) at (3.85,0) {verifier\\authorizes};
  \node[ev] (c4) at (3.85,-2) {$\tok{\mathrm{wo}7,\mathrm{auth}}$};
  \node[ev] (d5) at (4.95,-3) {$t_E$ fires\\ $[\mathrm{Safe}\wedge\mathrm{cert}]$};
  \node[ev] (e5) at (4.95,-4) {$\G_E(\mathrm{wo}7)$};
  \node[ev] (e1) at (0.55,-4) {$\G_P(\mathrm{wo}7)$};
  \foreach \u/\v in {a1/b1,d2/c2,d3/e3,a4/c4,d5/e5} \draw[->,thick,inkcol!70] (\u)--(\v);
  \draw[->,thick,inkcol!70] (b1.south) |- (d2.west); \draw[->,thick,inkcol!70] (c2.east) -| (d3.north); \draw[->,thick,inkcol!70] (c4.east) -| (d5.north);
  \draw[->,thick,inkcol!70] (e1.west) -- ++(-0.08,0) |- (a1.west);
  \draw[->,thick,inkcol!70] (e3.east) -| ($(a4.west)+(-0.1,0)$) -- (a4.west);
\end{tikzpicture}
\caption{Swim-lane trace of one successful execution of variant B. Each arrow is a causal dependency. Choices produce outcomes, outcomes produce tokens, tokens (together with guards) enable transitions, and transitions create the next game instance. The work-order identifier $\mathrm{wo}7$ in the token color ties the instances together.}
\label{fig:trace}
\end{figure}

\subsection{Structural analysis}
\label{sec:re1-structural}

Unfold the net of \cref{fig:examplenet} by splitting the result place $q_P$ into $q_P^{\mathrm{r}}$ (ready) and $q_P^{\mathrm{n}}$ (not ready), and $r_P$ into the columns $r_P^{\mathrm{r}}$ and $r_P^{\mathrm{n}}$. The routing transitions are $t_A$ (guard $[\mathrm{ready}]$) and $t_R$ (guard $[\mathrm{notReady}]$), and $t_\ell$ returns the repair token to $\G_P$.

\paragraph{Conservation.} Let $y$ be $1$ on every place that can hold the work-order token and $0$ on the evidence place. Every column moves the work-order token from one such place to another, so $y^\top C=0$. By \cref{prop:invariants}(i), there is always exactly one work-order token.

\paragraph{Safety.} The net provides three distinct guarantees, each independent of the players' strategies.
\begin{enumerate}[label=(S\arabic*),nosep]
  \item \emph{Execution requires an authorization token.} The authorization place is in the preset of $t_E$, and only $r_A$ produces into it. The rows of the state equation \eqref{eq:state} for the authorization and execution places read
  \begin{equation}
    M(\mathrm{auth}) = n_{r_A}-n_{t_E}\ge 0,\qquad M(\G_E)=n_{t_E},
  \end{equation}
  so every reachable marking satisfies $n_{t_E}\le n_{r_A}$: executions never outnumber authorizations.
  \item \emph{Execution requires valid evidence.} The guard $[\mathrm{Safe}(\zeta)\wedge\mathrm{cert}(\zeta)]$ depends only on the binding, so the unfolding contains no column of $t_E$ for evidence that fails it. A marking with an execution token bound to invalid evidence therefore has a zero row in $C$ and no producing column, and the certificate $y=\mathbf{e}_{(\G_E,\,\mathrm{invalid})}$ of \cref{prop:farkas} proves it unreachable.
  \item \emph{Authorizations are not duplicated.} The authorization token carries the work-order identifier, so the authorization place is one of the work-order places of the conservation invariant above. Hence $M(\mathrm{auth})\le1$ at every reachable marking.
\end{enumerate}
\Cref{fig:safety} shows the corresponding reachability picture.

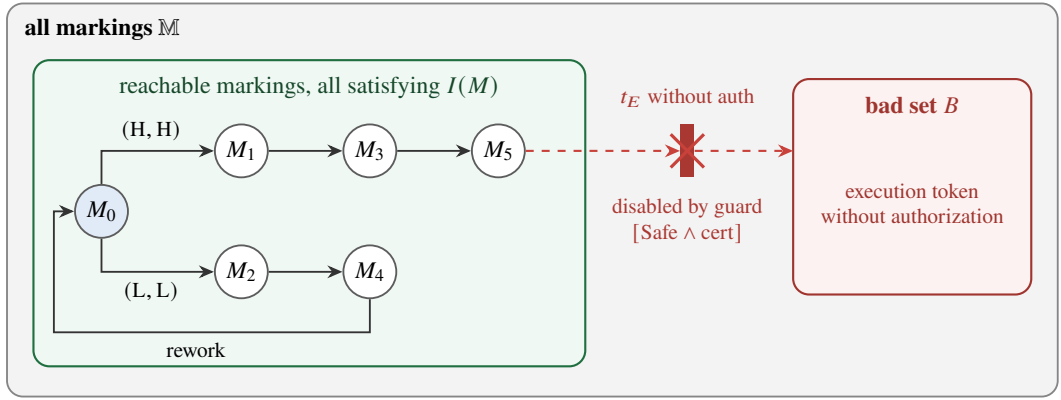
\begin{figure}[ht]
\centering
\begin{tikzpicture}[x=1cm,y=1cm,
  mk/.style={circle,draw=inkcol!80,line width=0.7pt,fill=white,font=\footnotesize,inner sep=0pt,minimum size=7mm},
  e/.style={->,line width=0.7pt,draw=inkcol}]
  \draw[line width=0.7pt,draw=inkcol!60,fill=softgray,rounded corners=6pt] (-0.5,-2.45) rectangle (13.4,2.75);
  \node[font=\footnotesize\bfseries,anchor=north west] at (-0.4,2.68) {all markings $\M$};
  \draw[line width=0.8pt,draw=okcol!70!black,fill=okcol!7,rounded corners=6pt] (-0.15,-2.05) rectangle (7.15,2.0);
  \node[font=\footnotesize,text=okcol!45!black,anchor=north] at (3.5,1.92) {reachable markings, all satisfying $I(M)$};
  \draw[line width=0.8pt,draw=badcol!80!black,fill=badcol!7,rounded corners=6pt] (9.9,-1.1) rectangle (13.05,1.75);
  \node[font=\footnotesize\bfseries,text=badcol!80!black,anchor=north] at (11.475,1.65) {bad set $B$};
  \node[font=\scriptsize,text=badcol!80!black,align=center] at (11.475,0.1) {execution token\\without authorization};
  \node[mk,fill=rescol!15] (m0) at (0.75,0) {$M_0$};
  \node[mk] (m1) at (2.6,0.8) {$M_1$};
  \node[mk] (m3) at (4.3,0.8) {$M_3$};
  \node[mk] (m5) at (6.0,0.8) {$M_5$};
  \node[mk] (m2) at (2.6,-0.8) {$M_2$};
  \node[mk] (m4) at (4.3,-0.8) {$M_4$};
  \draw[e] (m0.north) |- node[pos=0.72,above,font=\scriptsize]{$(\act{H},\act{H})$} (m1);
  \draw[e] (m0.south) |- node[pos=0.72,below,font=\scriptsize]{$(\act{L},\act{L})$} (m2);
  \draw[e] (m1) -- (m3); \draw[e] (m3) -- (m5); \draw[e] (m2) -- (m4);
  \draw[e] (m4.south) -- (4.3,-1.6) -- node[below,font=\scriptsize,pos=0.55]{rework} (0.12,-1.6) -- (0.12,0) -- (m0.west);
  \node[transition,fill=badcol!85!black] (tE) at (8.5,0.8) {};
  \draw[->,line width=0.8pt,dashed,draw=badcol] (m5) -- (tE);
  \draw[->,line width=0.8pt,dashed,draw=badcol] (tE) -- (9.9,0.8);
  \draw[line width=1.2pt,draw=badcol] ($(tE)+(-0.22,-0.22)$) -- ($(tE)+(0.22,0.22)$);
  \draw[line width=1.2pt,draw=badcol] ($(tE)+(-0.22,0.22)$) -- ($(tE)+(0.22,-0.22)$);
  \node[font=\scriptsize,text=badcol!85!black,anchor=south] at (8.5,1.25) {$t_E$ without auth};
  \node[font=\scriptsize,text=badcol!85!black,anchor=north,align=center] at (8.5,0.3) {disabled by guard\\$[\mathrm{Safe}\wedge\mathrm{cert}]$};
\end{tikzpicture}
\caption{Safety by construction in variant B. Every reachable marking, under \emph{every} admissible equilibrium (both the $(\act{H},\act{H})$ and $(\act{L},\act{L})$ branches), satisfies the invariant. An execution without authorization is excluded because $t_E$ needs the authorization token in its preset (S1), and an execution with invalid evidence is excluded by its guard (S2). The argument never mentions payoffs, so it is independent of equilibrium selection.}
\label{fig:safety}
\end{figure}

\paragraph{Progress.} By conservation, the net is bounded. The vector
\begin{equation}\label{eq:tinv}
  n^{\circlearrowleft}=\mathbf{e}_{r_P^{\mathrm{n}}}+\mathbf{e}_{t_R}+\mathbf{e}_{r_R}+\mathbf{e}_{t_\ell}
\end{equation}
satisfies $Cn^{\circlearrowleft}=0$ and is realized by the rework cycle. Since the net is bounded, an infinite execution must repeat a cycle of the reachability graph (\cref{sec:structural}), and the rework cycle is the only one. Whether it does depends on how often equilibrium play selects the column $r_P^{\mathrm{n}}$, which \cref{sec:re1-values} computes.

\subsection{Strategic analysis}
\label{sec:re1-strategic}

\Cref{thm:dominance} analyzes both variants with one computation. Outcomes of $\G_P$ lead to distinct successors in both variants, so the routing probabilities $\Psi$ coincide with the outcome-probability matrices $\Phi$, which here are $0$--$1$ indicators. Write $E_{XY}$ for the $2\times2$ indicator of the profile $(X,Y)$, take the operator as the row player, and set $a^+=\act{H}$, $a^-=\act{L}$. In both variants, $D_{\act{H}}(R)=2-3=-1$ and $D_{\act{L}}(R)=-1-0=-1$.

\paragraph{Variant A.} In the downstream games \eqref{eq:downstreamA}, the compliant action $d_i=0$ strictly dominates, so each has the unique equilibrium $(0,0)$, with values $V(\G_V)=4$ and $V(\G_R^{(k)})=-2k$. The outcome indicators of $\G_P$ are $\Phi_V=E_{HH}$, $\Phi_{R,1}=E_{HL}+E_{LH}$, and $\Phi_{R,2}=E_{LL}$. By \eqref{eq:Qmatrix}, the augmented matrix decomposes as
\begin{equation}\label{eq:decompA}
  Q(\gamma) = \begin{pmatrix}2&-1\\3&0\end{pmatrix}
  + \gamma\left[\,4\begin{pmatrix}1&0\\0&0\end{pmatrix}
  -2\begin{pmatrix}0&1\\1&0\end{pmatrix}
  -4\begin{pmatrix}0&0\\0&1\end{pmatrix}\right],
\end{equation}
with row differences $D(\Phi_V)=(1,0)$, $D(\Phi_{R,1})=(-1,1)$, and $D(\Phi_{R,2})=(0,-1)$. Hence
\begin{equation}\label{eq:DQA}
  D_{\act{H}}(Q)=-1+6\gamma,\qquad D_{\act{L}}(Q)=-1+2\gamma .
\end{equation}

\paragraph{Variant B.} The indicators are $\Phi_{\mathrm{ready}}=E_{HH}$ and $\Phi_{\mathrm{notReady}}=J-E_{HH}$. Against column $\act{L}$, $D_{\act{L}}(\Phi_{\mathrm{ready}})=D_{\act{L}}(\Phi_{\mathrm{notReady}})=0$: no successor is pivotal. \Cref{fig:decomposition} shows both decompositions.

\begin{figure}[ht]
\centering
\begin{tikzpicture}[x=1cm,y=1cm,
  cell/.style={draw=inkcol!45,line width=0.5pt,minimum height=0.6cm,inner sep=0pt,font=\scriptsize},
  op/.style={font=\small,inner sep=1pt},
  hd/.style={font=\scriptsize,align=center,anchor=south,inner sep=1pt},
  zero/.style={text=inkcol},
  piv/.style={draw=badcol,line width=1.2pt}]
  \newcommand{\mg}[9]{%
    \begin{scope}[shift={#1}]
      \node[cell,minimum width=\cw cm,fill=#6] at (0,0) {#2};
      \node[cell,minimum width=\cw cm,fill=#7] at (\cw,0) {#3};
      \node[cell,minimum width=\cw cm,fill=#8] at (0,-0.6) {#4};
      \node[cell,minimum width=\cw cm,fill=#9] at (\cw,-0.6) {#5};
    \end{scope}}
  \newcommand{\dg}[5]{%
    \begin{scope}[shift={#1}]
      \node[cell,minimum width=\cw cm,fill=softgray,#4] at (0,0) {#2};
      \node[cell,minimum width=\cw cm,fill=softgray,#5] at (\cw,0) {#3};
    \end{scope}}
  \def\fG{okcol!28}\def\fO{gamecol!30}\def\fR{badcol!22}\def\fW{white}

  \node[font=\footnotesize\bfseries,anchor=west] at (-1.2,1.55) {Variant A: graded routing};
  \def\cw{0.95}
  \mg{(0,0)}{$2$}{$-1$}{$3$}{$0$}{\fW}{\fW}{\fW}{\fW}
  \node[hd] at (0.475,0.42) {stage $R$};
  \node[font=\scriptsize,anchor=east,text=inkcol] at (-0.5,0) {$\act{H}$};
  \node[font=\scriptsize,anchor=east,text=inkcol] at (-0.5,-0.6) {$\act{L}$};
  \node[op] at (2.0,-0.3) {$+\ \gamma\,\Big[$};
  \node[op,anchor=east] at (2.85,-0.3) {$4$};
  \mg{(3.4,0)}{$1$}{$0$}{$0$}{$0$}{\fG}{\fW}{\fW}{\fW}
  \node[hd] at (3.875,0.42) {$\Phi_V$\\[-1pt]to $\G_V$, $V{=}4$};
  \node[op,anchor=east] at (5.6,-0.3) {$-\,2$};
  \mg{(6.15,0)}{$0$}{$1$}{$1$}{$0$}{\fW}{\fO}{\fO}{\fW}
  \node[hd] at (6.625,0.42) {$\Phi_{R,1}$\\[-1pt]to $\G_R$, $k{=}1$};
  \node[op,anchor=east] at (8.35,-0.3) {$-\,4$};
  \mg{(8.9,0)}{$0$}{$0$}{$0$}{$1$}{\fW}{\fW}{\fW}{\fR}
  \node[hd] at (9.375,0.42) {$\Phi_{R,2}$\\[-1pt]to $\G_R$, $k{=}2$};
  \node[op] at (10.55,-0.3) {$\Big]$};
  \node[op] at (10.95,-0.3) {$=$};
  \def\cw{1.3}
  \mg{(11.8,0)}{$2{+}4\gamma$}{$-1{-}2\gamma$}{$3{-}2\gamma$}{$-4\gamma$}{\fW}{\fW}{\fW}{\fW}
  \node[hd] at (12.45,0.42) {augmented $Q(\gamma)$};
  \node[font=\scriptsize,anchor=east] at (-0.5,-1.6) {$D(\cdot)$};
  \def\cw{0.95}
  \dg{(0,-1.6)}{$-1$}{$-1$}{}{}
  \dg{(3.4,-1.6)}{$1$}{$0$}{}{}
  \dg{(6.15,-1.6)}{$-1$}{$1$}{}{}
  \dg{(8.9,-1.6)}{$0$}{$-1$}{}{}
  \def\cw{1.3}
  \dg{(11.8,-1.6)}{$-1{+}6\gamma$}{$-1{+}2\gamma$}{draw=okcol!65!black,line width=1pt}{draw=okcol!65!black,line width=1pt}
  \node[font=\scriptsize,text=okcol!45!black,anchor=north] at (12.45,-1.95) {$\act{H}$ dominant $\iff\gamma>\tfrac12$};
  \foreach \x in {0.475,3.875,6.625,9.375,12.45} \draw[->,line width=0.6pt,draw=inkcol!50] (\x,-0.95) -- (\x,-1.28);

  \begin{scope}[yshift=-4.5cm]
  \node[font=\footnotesize\bfseries,anchor=west] at (-1.2,1.55) {Variant B: binary routing};
  \def\cw{0.95}
  \mg{(0,0)}{$2$}{$-1$}{$3$}{$0$}{\fW}{\fW}{\fW}{\fW}
  \node[hd] at (0.475,0.42) {stage $R$};
  \node[font=\scriptsize,anchor=east,text=inkcol] at (-0.5,0) {$\act{H}$};
  \node[font=\scriptsize,anchor=east,text=inkcol] at (-0.5,-0.6) {$\act{L}$};
  \node[op] at (1.95,-0.3) {$+\ \gamma\,\Big[$};
  \node[op,anchor=east] at (2.85,-0.3) {$v_A$};
  \mg{(3.4,0)}{$1$}{$0$}{$0$}{$0$}{\fG}{\fW}{\fW}{\fW}
  \node[hd] at (3.875,0.42) {$\Phi_{\mathrm{ready}}$\\[-1pt]to $\G_A$};
  \node[op,anchor=east] at (5.6,-0.3) {$+\,v_R$};
  \mg{(6.15,0)}{$0$}{$1$}{$1$}{$1$}{\fW}{\fR}{\fR}{\fR}
  \node[hd] at (6.625,0.42) {$\Phi_{\mathrm{notReady}}$\\[-1pt]to $\G_R$};
  \node[op] at (7.8,-0.3) {$\Big]$};
  \node[op] at (8.2,-0.3) {$=$};
  \def\cw{1.3}
  \mg{(9.05,0)}{$2{+}\gamma v_A$}{$-1{+}\gamma v_R$}{$3{+}\gamma v_R$}{$\gamma v_R$}{\fW}{\fW}{\fW}{\fW}
  \node[hd] at (9.7,0.42) {augmented $Q$};
  \node[font=\scriptsize,anchor=east] at (-0.5,-1.6) {$D(\cdot)$};
  \def\cw{0.95}
  \dg{(0,-1.6)}{$-1$}{$-1$}{}{}
  \dg{(3.4,-1.6)}{$1$}{$0$}{}{piv}
  \dg{(6.15,-1.6)}{$-1$}{$0$}{}{piv}
  \def\cw{1.3}
  \dg{(9.05,-1.6)}{$-1{+}\gamma\Delta v$}{$-1$}{}{piv}
  \foreach \x in {0.475,3.875,6.625,9.7} \draw[->,line width=0.6pt,draw=inkcol!50] (\x,-0.95) -- (\x,-1.28);
  \node[font=\scriptsize,text=badcol!85!black,align=left,anchor=west] at (11.3,-1.0) {column $\act{L}$: no outcome is\\pivotal, so $D_{\act{L}}(Q)=-1$\\for every $\gamma$, $v_A$, $v_R$};
  \end{scope}
\end{tikzpicture}
\caption{Matrix decomposition of the continuation-augmented game at $\G_P$ (row player's payoffs), as in \eqref{eq:Qmatrix}. Each outcome contributes an indicator matrix $\Phi_\omega$ (shaded cells mark the profiles that lead to that outcome), weighted by the value of the game it routes to. The row-difference operator $D(\cdot)$ is linear, so the incentive to choose $\act{H}$ over $\act{L}$ is the stage difference plus the weighted differences of the indicators (\cref{thm:dominance}). In variant A every column has a pivotal outcome, so the incentive can be reversed. In variant B both indicators have zero difference against $\act{L}$ (red), so no continuation value can change that comparison.}
\label{fig:decomposition}
\end{figure}

The three propositions below are now direct consequences of the general theory. Since the game is symmetric, the column player's analysis is identical.

\begin{proposition}[Equilibrium reversal in variant A]\label{prop:reversal}
At $\gamma=0.75$, $(\act{H},\act{H})$ is the unique equilibrium of the augmented game at $\G_P$, although $(\act{L},\act{L})$ is the unique equilibrium of the stage game. The equilibrium route therefore changes from recovery ($t_R$) to verification ($t_V$).
\end{proposition}
\begin{proof}
By \eqref{eq:DQA}, $D_{\act{H}}(Q)=3.5>0$ and $D_{\act{L}}(Q)=0.5>0$, so $\act{H}$ strictly dominates by \cref{thm:dominance}. At $\gamma=0$, both differences equal $-1$, so $\act{L}$ strictly dominates. By the guards in \cref{tab:routingA}, $(\act{L},\act{L})$ enables $t_R$ and $(\act{H},\act{H})$ enables $t_V$.
\end{proof}

\begin{proposition}[Phase diagram of variant A]\label{prop:phase}
The successors of $\G_P$ are terminal for the purposes of $\G_P$, so the continuation values are exogenous and \cref{cor:thresholds} is exact. For $\gamma\in[0,1]$, the equilibria of the augmented game at $\G_P$ are:
\begin{enumerate}[label=(\roman*),nosep]
  \item $\gamma<1/6$: $\act{L}$ is strictly dominant, and $(\act{L},\act{L})$ is the unique equilibrium (myopic regime);
  \item $\gamma=1/6$: $(\act{L},\act{L})$ is a strict equilibrium, and $(\act{H},\act{H})$ is a weak equilibrium, since $\act{H}$ and $\act{L}$ tie against $\act{H}$;
  \item $1/6<\gamma<1/2$: $(\act{H},\act{H})$ and $(\act{L},\act{L})$ are strict equilibria, and there is a third, symmetric mixed equilibrium in which each player chooses $\act{H}$ with probability
  \begin{equation}\label{eq:qstar}
    q^\star(\gamma)=\frac{1-2\gamma}{4\gamma}\in(0,1)
  \end{equation}
  (coordination regime);
  \item $\gamma=1/2$: $\act{H}$ is weakly dominant, $(\act{H},\act{H})$ is a strict equilibrium, and $(\act{L},\act{L})$ is a weak equilibrium;
  \item $\gamma>1/2$: $\act{H}$ is strictly dominant, and $(\act{H},\act{H})$ is the unique equilibrium (reversal regime).
\end{enumerate}
\end{proposition}
\begin{proof}
By \cref{cor:thresholds} and \eqref{eq:DQA}, the comparison against $\act{H}$ switches sign at $\gamma^\star_{\act{H}}=1/6$, and the comparison against $\act{L}$ at $\gamma^\star_{\act{L}}=1/2$. The game is symmetric. Cases (i), (v), and the pure equilibria in (ii)--(iv) follow from the signs of $D_{\act{H}}(Q)$ and $D_{\act{L}}(Q)$. For the mixed equilibrium, if the opponent plays $\act{H}$ with probability $q$, a player is indifferent if and only if $qD_{\act{H}}(Q)+(1-q)D_{\act{L}}(Q)=0$, that is, $q(-1+6\gamma)+(1-q)(-1+2\gamma)=0$, which gives \eqref{eq:qstar}. It lies strictly between $0$ and $1$ exactly when $1/6<\gamma<1/2$, and it tends to the pure profiles $(\act{H},\act{H})$ at $\gamma=1/6$ and $(\act{L},\act{L})$ at $\gamma=1/2$.
\end{proof}

\begin{proposition}[Binary routing and dominance]\label{prop:binary}
In variant B, with prescribed continuation values $v_A,v_R$ and $\Delta v:=v_A-v_R$, the augmented game at $\G_P$ satisfies:
\begin{enumerate}[label=(\roman*),nosep]
  \item $(\act{L},\act{L})$ is an equilibrium for every $\gamma\ge0$ and every $\Delta v$;
  \item $(\act{H},\act{H})$ is an equilibrium if and only if $\gamma\,\Delta v\ge 1$.
\end{enumerate}
\end{proposition}
\begin{proof}
(i) No successor is pivotal against $\act{L}$, so by \cref{thm:dominance} $D_{\act{L}}(Q)=D_{\act{L}}(R)=-1$ for all values: $\act{L}$ is the best reply to $\act{L}$. (ii) $D_{\act{H}}(Q)=-1+\gamma\big(v_A D_{\act{H}}(E_{HH})+v_R D_{\act{H}}(J-E_{HH})\big)=-1+\gamma\,\Delta v$, which is nonnegative if and only if $\gamma\,\Delta v\ge1$.
\end{proof}

As a local illustration, take the prescribed values $v_A=5$, $v_R=0$, and $\gamma=1$. Then both $(\act{H},\act{H})$ and $(\act{L},\act{L})$ are equilibria of the augmented game at $\G_P$. The downstream value creates a second, cooperative equilibrium but does not remove the inefficient one, so the solution layer must state a selection rule. This calculation is not yet an equilibrium of the cyclic global game, because $v_R$ is in fact endogenous; \cref{prop:variantB} corrects this. \Cref{fig:matrices} shows the three augmented games as best-reply diagrams, and \cref{fig:phase} shows the regimes.

\begin{figure}[ht]
\centering
\begin{tikzpicture}[eq/.style={draw=okcol!70!black,line width=1.6pt,rounded corners=3pt,inner sep=0pt},
  br/.style={->,line width=1.1pt,draw=#1}, br/.default=inkcol]
  \payoffgrid{(0,0)}{m}{$2,2$}{$-1,3$}{$3,-1$}{$0,0$}{(a) myopic stage game}
  \draw[br] ($(ma.west)+(-0.28,-0.2)$) -- ($(mc.west)+(-0.28,0.2)$);
  \draw[br] ($(mb.east)+(0.25,-0.2)$) -- ($(md.east)+(0.25,0.2)$);
  \draw[br=gamecol!80!black] ($(ma.north)+(0.35,0.12)$) -- ($(mb.north)+(-0.35,0.12)$);
  \draw[br=gamecol!80!black] ($(mc.south)+(0.35,-0.12)$) -- ($(md.south)+(-0.35,-0.12)$);
  \node[eq,fit=(md)] {};
  \payoffgrid{(5.45,0)}{a}{$5,5$}{$-2.5,1.5$}{$1.5,-2.5$}{$-3,-3$}{(b) Variant A, $\gamma=0.75$}
  \draw[br] ($(ac.west)+(-0.28,0.2)$) -- ($(aa.west)+(-0.28,-0.2)$);
  \draw[br] ($(ad.east)+(0.25,0.2)$) -- ($(ab.east)+(0.25,-0.2)$);
  \draw[br=gamecol!80!black] ($(ab.north)+(-0.35,0.12)$) -- ($(aa.north)+(0.35,0.12)$);
  \draw[br=gamecol!80!black] ($(ad.south)+(-0.35,-0.12)$) -- ($(ac.south)+(0.35,-0.12)$);
  \node[eq,fit=(aa)] {};
  \payoffgrid{(10.9,0)}{b}{$7,7$}{$-1,3$}{$3,-1$}{$0,0$}{(c) Variant B, $\gamma=1$}
  \draw[br] ($(bc.west)+(-0.28,0.2)$) -- ($(ba.west)+(-0.28,-0.2)$);
  \draw[br] ($(bb.east)+(0.25,-0.2)$) -- ($(bd.east)+(0.25,0.2)$);
  \draw[br=gamecol!80!black] ($(bb.north)+(-0.35,0.12)$) -- ($(ba.north)+(0.35,0.12)$);
  \draw[br=gamecol!80!black] ($(bc.south)+(0.35,-0.12)$) -- ($(bd.south)+(-0.35,-0.12)$);
  \node[eq,fit=(ba)] {}; \node[eq,fit=(bd)] {};
  \node[font=\footnotesize,align=center] at (6.7,-2.95) {%
    \tikz[baseline=-0.5ex]\draw[br] (0,0)--(0.6,0); operator's better reply (within a column) \qquad
    \tikz[baseline=-0.5ex]\draw[br=gamecol!80!black] (0,0)--(0.6,0); provider's better reply (within a row)\\[2pt]
    \tikz[baseline=-0.5ex]\node[eq,minimum width=5mm,minimum height=3mm]{}; pure Nash equilibrium};
\end{tikzpicture}
\caption{Best-reply (movement) diagrams. Rows are the operator's action and columns the provider's; cells show the payoffs to $(O,S)$. Arrows point toward the better reply, so a cell with only incoming arrows is a pure equilibrium. (a) The stage game alone has the unique equilibrium $(\act{L},\act{L})$. (b) Graded routing reverses every arrow: $(\act{H},\act{H})$ becomes unique. (c) Binary routing reverses only the arrows into $(\act{H},\act{H})$, so both equilibria survive.}
\label{fig:matrices}
\end{figure}
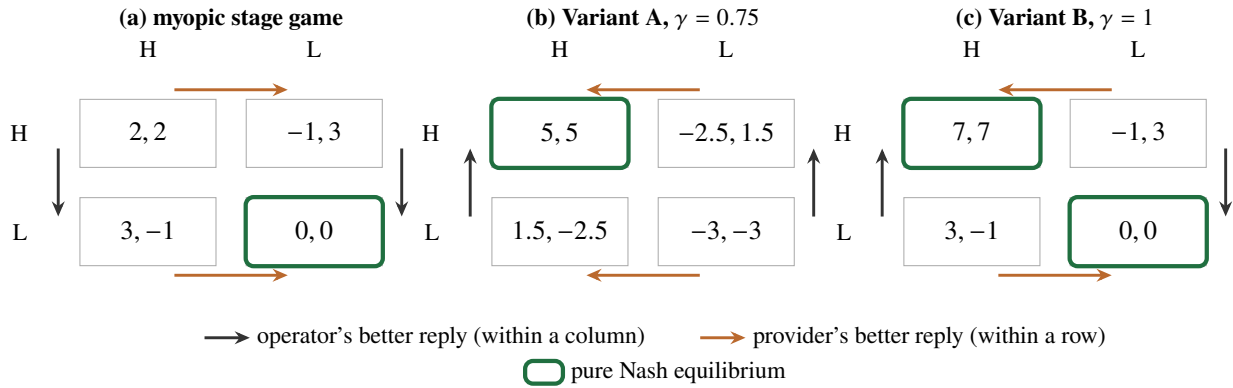

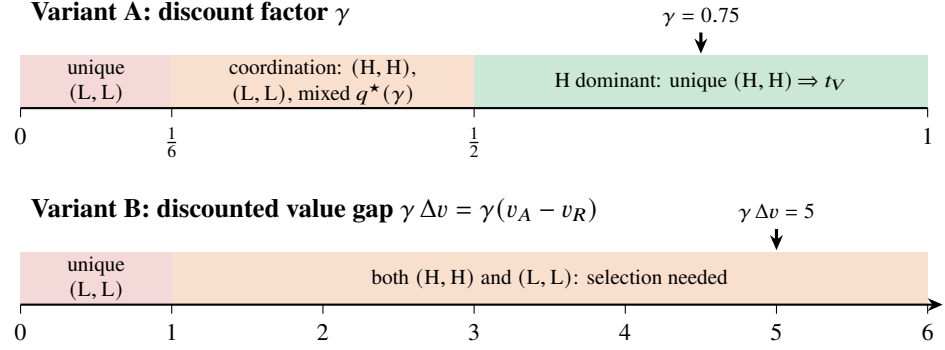
\begin{figure}[ht]
\centering
\begin{tikzpicture}[x=12cm,y=1cm]
  \node[font=\small\bfseries,anchor=west] at (0,1.25) {Variant A: discount factor $\gamma$};
  \fill[badcol!20] (0,0) rectangle (1/6,0.7);
  \fill[gamecol!25] (1/6,0) rectangle (1/2,0.7);
  \fill[okcol!25] (1/2,0) rectangle (1,0.7);
  \draw[thick] (0,0) -- (1,0);
  \foreach \x/\l in {0/0,0.16667/\tfrac16,0.5/\tfrac12,1/1} \draw (\x,0) -- (\x,-0.12) node[below,font=\footnotesize] {$\l$};
  \node[font=\scriptsize,align=center] at (1/12,0.35) {unique\\$(\act{L},\act{L})$};
  \node[font=\scriptsize,align=center] at (1/3,0.35) {coordination: $(\act{H},\act{H})$,\\$(\act{L},\act{L})$, mixed $q^\star(\gamma)$};
  \node[font=\scriptsize,align=center] at (0.75,0.35) {$\act{H}$ dominant: unique $(\act{H},\act{H})$\,$\Rightarrow t_V$};
  \draw[->,very thick] (0.75,1.0) -- (0.75,0.72); \node[font=\scriptsize,anchor=south] at (0.75,0.95) {$\gamma=0.75$};
  \begin{scope}[yshift=-2.6cm,x=2cm]
  \node[font=\small\bfseries,anchor=west] at (0,1.25) {Variant B: discounted value gap $\gamma\,\Delta v=\gamma(v_A-v_R)$};
  \fill[badcol!20] (0,0) rectangle (1,0.7);
  \fill[gamecol!25] (1,0) rectangle (6,0.7);
  \draw[thick,->] (0,0) -- (6.1,0);
  \foreach \x in {0,1,...,6} \draw (\x,0) -- (\x,-0.12) node[below,font=\footnotesize] {$\x$};
  \node[font=\scriptsize,align=center] at (0.5,0.35) {unique\\$(\act{L},\act{L})$};
  \node[font=\scriptsize,align=center] at (3.5,0.35) {both $(\act{H},\act{H})$ and $(\act{L},\act{L})$: selection needed};
  \draw[->,very thick] (5,1.0) -- (5,0.72); \node[font=\scriptsize,anchor=south] at (5,0.95) {$\gamma\,\Delta v=5$};
  \end{scope}
\end{tikzpicture}
\caption{Regime maps for the two variants, from \cref{prop:phase,prop:binary}. The stage game is the same in both, but the downstream \emph{structure} differs. Graded consequences (variant A) can make cooperation dominant. A coupling in which only joint success matters (variant B) can at most add a coordination equilibrium, however large the reward. Which structure applies is a property of the net, not of $\G_P$.}
\label{fig:phase}
\end{figure}

\begin{intuition}[Comparison of the two variants]
In variant A, a unilateral shirk ($k=1$) is punished downstream \emph{even if the other player also shirks}, because $V(\G_R^{(2)})<V(\G_R^{(1)})$. In the language of \cref{thm:dominance}, some successor is pivotal against $\act{L}$. In variant B, once one player shirks, the other's choice no longer affects which successor game is reached, so no successor is pivotal against $\act{L}$.
\end{intuition}

\subsection{Values and computation}
\label{sec:re1-values}

\paragraph{Variant A: backward induction.} The place graph is $\G_P\to\{\G_V,\G_R^{(1)},\G_R^{(2)}\}$, with no cycles. By \cref{prop:nilpotent}, a single backward pass solves it. The successor games are solved first, giving the values in \eqref{eq:downstreamA}. The augmented game \eqref{eq:decompA} is then solved at $\G_P$. At $\gamma=0.75$, the equilibrium $(\act{H},\act{H})$ gives each player $V(\G_P)=2+0.75\cdot4=5$, the top-left entry of \cref{fig:matrices}(b).

\paragraph{Variant B: global analysis.} The rework arc $\G_R\to\G_P$ merges $\G_P$ and $\G_R$ into one SCC, and $\G_A$ is downstream. For a dynamically consistent analysis, take $\gamma\in(0,1)$, treat $\G_A$ as terminal with value $v_A$ (the authorization path ends the episode), and give $\G_R$ a stage payoff $c_R<0$ to each player for the cost of renegotiation and downtime, with deterministic routing back to $\G_P$. The value of $\G_R$ is then endogenous: $V(\G_R)=c_R+\gamma V(\G_P)$.

\begin{proposition}[Global equilibria of variant B]\label{prop:variantB}
For $\gamma\in(0,1)$:
\begin{enumerate}[label=(\roman*),nosep]
  \item playing $(\act{L},\act{L})$ at $\G_P$ is always a global equilibrium, with values $V(\G_P)=\gamma c_R/(1-\gamma^2)$ and $V(\G_R)=c_R/(1-\gamma^2)$;
  \item playing $(\act{H},\act{H})$ at $\G_P$ is a global equilibrium if and only if
  \begin{equation}\label{eq:cstar}
    c_R\ \le\ c^\star(\gamma):=\frac{(2+\gamma v_A)(1-\gamma^2)-3}{\gamma},
  \end{equation}
  with values $V(\G_P)=2+\gamma v_A$ and $V(\G_R)=c_R+\gamma(2+\gamma v_A)$.
\end{enumerate}
\end{proposition}
\begin{proof}
By \cref{thm:consistency}, it suffices to check the augmented game at $\G_P$, with $V(\G_R)$ computed from the profile itself by \cref{lem:policy}. (i) Under $(\act{L},\act{L})$, $V(\G_P)=\gamma V(\G_R)$ and $V(\G_R)=c_R+\gamma V(\G_P)$, which give the stated values. Against $\act{L}$, no successor is pivotal, so $\act{L}$ is the best reply by \cref{prop:binary}(i). (ii) Under $(\act{H},\act{H})$, $V(\G_P)=2+\gamma v_A$ and $V(\G_R)=c_R+\gamma V(\G_P)$. By \cref{prop:binary}(ii), $(\act{H},\act{H})$ is an equilibrium if and only if $\gamma\big(v_A-V(\G_R)\big)\ge1$. Substituting $V(\G_R)$ and rearranging gives \eqref{eq:cstar}.
\end{proof}

With $\gamma=0.9$ and $v_A=5$, the threshold is $c^\star=-1.96$. If recovery is cheap ($c_R=0$), then $V(\G_R)=5.85$ under $(\act{H},\act{H})$, and a player who shirks loses little: the rework loop returns play to a planning game in which cooperation resumes. So $(\act{H},\act{H})$ is \emph{not} an equilibrium. This contrasts with the local illustration, which prescribed $v_R=0$. If recovery is costly ($c_R=-3$), then $V(\G_R)=2.85$, $\gamma(v_A-V(\G_R))=1.94\ge1$, and $(\act{H},\act{H})$ is a global equilibrium with $V(\G_P)=6.5$. Alongside it, the livelock equilibrium $(\act{L},\act{L})$ has $V(\G_P)=-14.21$. \emph{A forgiving rework loop weakens upstream incentives}; only a sufficiently costly loop sustains cooperation. Within the component $\{\G_P,\G_R\}$, the routing matrix depends on the selected equilibrium:
\begin{equation}
  P^\sigma_{\{\G_P,\G_R\}}=\begin{pmatrix}0&0\\1&0\end{pmatrix}\ \text{under }(\act{H},\act{H}),\qquad
  P^\sigma_{\{\G_P,\G_R\}}=\begin{pmatrix}0&1\\1&0\end{pmatrix}\ \text{under }(\act{L},\act{L}),
\end{equation}
with rows and columns ordered $(\G_P,\G_R)$. Under $(\act{H},\act{H})$, the first row is zero because all mass leaves the component toward $\G_A$, so $\rho=0$ and play terminates. Under $(\act{L},\act{L})$, the block is a permutation matrix, so $\rho=1$. By \cref{prop:termination}, if the repair does not change the planning game, the net livelocks on the realized T-invariant \eqref{eq:tinv}. The two trajectories are
\begin{equation}
  \G_P \xrightarrow{\ \mathrm{ready}\ } \G_A \xrightarrow{\ \mathrm{approved}\ } \G_E
  \qquad\text{and}\qquad
  \G_P \xrightarrow{\ \mathrm{notReady}\ } \G_R \xrightarrow{\ \mathrm{repaired}\ } \G_P \to\cdots .
\end{equation}
In practice, the repair token updates the parameters of the next planning game, such as the budget and the equipment condition (dynamic instantiation). This changes $R^P$ from one visit to the next, and so it can break the livelock.

\subsection{Composition with concurrent diagnostics}
\label{sec:re1-composition}

In a more realistic extension, the ready outcome forks into two concurrent diagnostic games, a mechanical inspection $\G_M$ and a cybersecurity inspection $\G_C$ (\cref{fig:forkjoin}). A join transition requires
\begin{equation}
  M(p_M^{\mathrm{pass}})\ge1,\qquad M(p_C^{\mathrm{pass}})\ge1,
\end{equation}
and consumes a matching pair of signed pass certificates with the same work-order ID. Only then is $\G_A$ activated. A failed diagnostic produces a repair token instead of a pass certificate.

\begin{figure}[ht]
\centering
\begin{tikzpicture}[x=1cm,y=1cm]
  \begin{scope}[on background layer]
    \fill[rescol!5,rounded corners=4pt] (2.2,-2.45) rectangle (4.75,2.45);
  \end{scope}
  \node[lbl,text=rescol!70!black,font=\footnotesize\itshape] at (3.47,0) {concurrent};
  \node[place,label={[lbl]below:ready}] (rd) at (0,0) {}; \node[ctok=okcol] at (rd) {};
  \node[transition,label={[lbl]below:fork}] (fk) at (1.25,0) {};
  \coordinate (fan0) at (1.8,0);
  \node[gplace,label={[glbl]above:$\G_M$ mechanical}] (gm) at (2.9,1.5) {};
  \node[gplace,label={[glbl]below:$\G_C$ cyber}] (gc) at (2.9,-1.5) {};
  \node[rtrans,label={[lbl]below:$r_M$}] (rm) at (4.15,1.5) {};
  \node[rtrans,label={[lbl]above:$r_C$}] (rc) at (4.15,-1.5) {};
  \coordinate (fanM) at (5.1,1.5); \coordinate (fanC) at (5.1,-1.5);
  \node[place,label={[lbl]above:$p_M^{\mathrm{fail}}$}] (fm) at (6.0,2.45) {};
  \node[place,label={[lbl]above:$p_M^{\mathrm{pass}}$}] (pm) at (6.0,0.85) {};
  \node[place,label={[lbl]below:$p_C^{\mathrm{pass}}$}] (pc) at (6.0,-0.85) {};
  \node[place,label={[lbl]below:$p_C^{\mathrm{fail}}$}] (fc) at (6.0,-2.45) {};
  \node[transition] (jn) at (7.5,0) {};
  \node[gplace,label={[glbl]right:$\G_A$}] (ga) at (9.9,0) {};
  \node[transition] (xm) at (7.5,2.45) {};
  \node[transition] (xc) at (7.5,-2.45) {};
  \node[place,label={[lbl,text=badcol]right:repair $\to\G_R$}] (rpM) at (8.8,2.45) {};
  \node[place,label={[lbl,text=badcol]right:repair $\to\G_R$}] (rpC) at (8.8,-2.45) {};
  \draw[arc] (rd)--(fk);
  \draw[arc] (fk.east)--(fan0)|-(gm); \draw[arc] (fk.east)--(fan0)|-(gc);
  \draw[arc] (gm)--(rm); \draw[arc] (gc)--(rc);
  \draw[arc,okcol!65!black] (rm.east)--(fanM)|-(pm); \draw[arc,badcol] (rm.east)--(fanM)|-(fm);
  \draw[arc,okcol!65!black] (rc.east)--(fanC)|-(pc); \draw[arc,badcol] (rc.east)--(fanC)|-(fc);
  \draw[arc] (pm)-|(jn); \draw[arc] (pc)-|(jn);
  \draw[arc] (jn) -- node[below,guard]{$[\mathrm{id}_M{=}\mathrm{id}_C]$} (ga);
  \draw[arc] (fm)--(xm); \draw[arc] (xm)--(rpM); \draw[arc] (fc)--(xc); \draw[arc] (xc)--(rpC);
\end{tikzpicture}
\caption{Fork--join diagnostics. The two inspection games run concurrently and may be played by different agents with different information. The join fires only when \emph{both} signed pass certificates for the same work order are present. A single failure diverts that token to repair.}
\label{fig:forkjoin}
\end{figure}
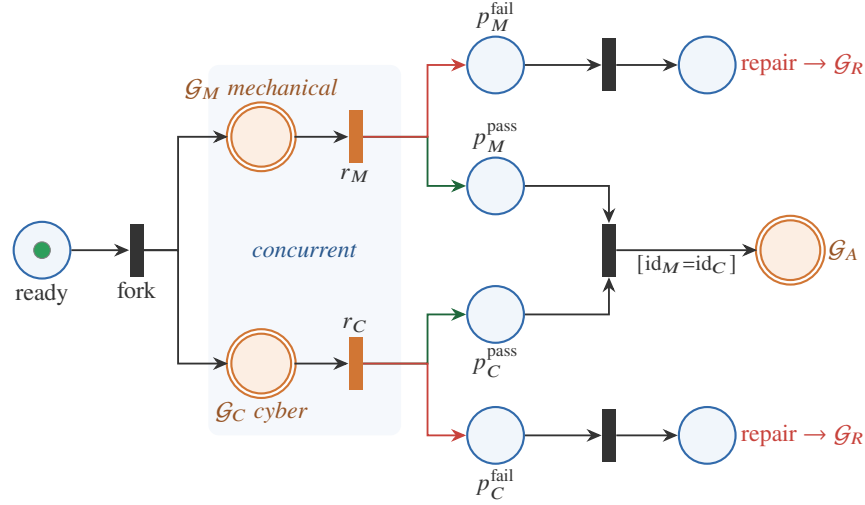

\Cref{thm:composition}(i) and (iv) give the structural analysis directly. Let $y$ be $1$ on the ready place and on $\G_A$, and $\tfrac12$ on every place of each branch, including the pass, fail, and repair places. The fork consumes weight $1$ and produces $\tfrac12+\tfrac12$. The join consumes $\tfrac12+\tfrac12$ and produces $1$. Every other column moves weight $\tfrac12$ within a branch. Hence $y^\top C=0$: \emph{a work order is split across the concurrent inspections and reassembled by the join}. The same calculation exposes a defect. If one inspection passes and the other fails, the pass certificate (weight $\tfrac12$) is never consumed, so it is stranded. A complete design adds a cleanup column that consumes the stranded certificate together with the repair token. A single-path game tree cannot express this synchronization, or its failure mode.

\section{Running example II: an intrusion campaign}
\label{sec:re2}

The second running example is adversarial: one player is an attacker whose objective opposes the system's. It shows that the theory of \cref{sec:theory} applies unchanged, and it stresses the two payoffs of the net structure, \emph{composition} and \emph{computation}. The example builds on our multi-resolution cyber-defense games~\cite{yang2025multires,yang2025general,yang2025crossechelon}. There, a campaign is a Markov game over network segments, and each segment hides a technique-level engagement, a \emph{micro base game} (MBG), between an attacker $\PA$ and a defender $\PD$.

\subsection{Model}
\label{sec:re2-model}

\paragraph{From multi-resolution games to an OCGPN.} In the multi-resolution framework, the enterprise network is partitioned into segments with a coupling graph. At each segment $s$, the attacker selects ATT\&CK techniques~\cite{strom2018mitre} and the defender selects D3FEND countermeasures~\cite{mitre_d3fend}, in an extensive-form game $\Gamma^s$ with imperfect information. Every outcome of $\Gamma^s$ is labeled by \emph{the segment the attacker occupies when the engagement ends}. \emph{Grounding} (zoom-in) writes the continuation value of that position into the leaves of $\Gamma^s$. \emph{Lifting} (zoom-out) carries the equilibrium outcome distribution back up as the attacker's strategic routing. Following~\cite{yang2025crossechelon}, segments are written $s$ and micro base games $\Gamma^s$; these symbols are local to this example and unrelated to the global state $s$ and the global game $\Gamma(\N)$. Each object has a direct OCGPN counterpart (\cref{tab:dictionary}). The attacker's foothold is a \emph{token}, and the attacker and defender are the \emph{players} of the game at the place where that token sits.

\begin{table}[ht]
\centering
\caption{Dictionary between multi-resolution cyber games~\cite{yang2025multires,yang2025crossechelon} and OCGPN objects.}
\label{tab:dictionary}
\small
\begin{tabular}{@{}p{0.42\textwidth}p{0.52\textwidth}@{}}
\toprule
\textbf{Multi-resolution cyber game} & \textbf{OCGPN counterpart}\\\midrule
segment $s$; MBG $\Gamma^s$ & game place $p_s\in P_G$; local game $\G_{p_s}$\\
attacker foothold in $s$ & active-instance token $\tok{\mathrm{cid},s,\mathrm{artifacts}}$ in $p_s$\\
outcome labeled by the segment $s'$ & $\Out_{p_s}$ emits a foothold token routed to $p_{s'}$\\
remaining in place (self-loop), stalling penalty $\beta$ & rework arc returning the token to $p_s$\\
absorbing exfiltration states $s^0$ (genuine), $s^6$ (decoy) & terminal places $z_{\mathrm{gen}}$, $z_{\mathrm{dec}}$\\
grounding of MBG leaves & continuation-augmented matrix $Q^{p_s}(V)$, \eqref{eq:Qmatrix}\\
lifting to strategic routing & routing matrix $P^\sigma$, \eqref{eq:routing}\\
resolution configuration (coarse $s$ or fine $\Gamma^s$) & choice of local matrices at each game place (\cref{cor:coarse})\\
durable artifacts (stolen credentials, implants) & colored tokens in side places, read by later games\\
\bottomrule
\end{tabular}
\end{table}

\paragraph{The campaign net.} We use a three-segment slice of the case-study network of~\cite{yang2025crossechelon} (\cref{fig:campaignnet}). An attacker holding a foothold in the DMZ enters the IT site, the lateral-movement place $p_L$. From there it tries to reach the developer server, the asset place $p_A$, which holds genuine intellectual property alongside decoy files. The parameters are those of~\cite{yang2025crossechelon}: attacker capability $\lambda_A=0.6$, securing effectiveness $\eta=0.85$, stalling penalty $\beta=-2$, discount factor $\gamma=0.85$, segment value $\nu(A)=10$, genuine exfiltration value $\nu(s^0)=15$, and decoy exfiltration value $\nu(s^6)=-1$. The game is zero-sum, and the attacker maximizes. We write $V_L:=V(p_L)$ and $V_A:=V(p_A)$ for the attacker's values at the IT-site and asset places, and $\nu(L)$, $\nu(A)$ for the values of these two segments.

\begin{figure}[ht]
\centering
\begin{tikzpicture}[x=0.96cm,y=1cm]
  \node[place,label={[lbl]above:DMZ foothold}] (dmz) at (0,0) {}; \node[ctok=badcol] at (dmz) {};
  \node[transition,label={[lbl]below:$t_{\mathrm{in}}$}] (tin) at (1.2,0) {};
  \node[gplace,label={[glbl]above:$p_L$\\[-1pt]IT site}] (gl) at (2.5,0) {};
  \node[rtrans,label={[lbl]below right:$r_L$}] (rl) at (3.8,0) {};
  \node[place,label={[lbl]below:$q_L$}] (ol) at (5.0,0) {};
  \coordinate (fanL) at (5.7,0);
  \node[transition,label={[guard]above:$[\mathrm{moved}]$}] (tmv) at (6.5,1.6) {};
  \node[gplace,label={[glbl]below:$p_A$\\[-1pt]asset}] (ga) at (7.8,1.6) {};
  \node[rtrans,label={[lbl]below:$r_A$}] (ra) at (9.1,1.6) {};
  \node[place,label={[lbl]below:$q_A$}] (oa) at (10.3,1.6) {};
  \coordinate (fanA) at (11.0,1.6);
  \node[transition,label={[guard]above:$[\mathrm{genuine}\wedge\mathrm{egress}]$},label={[lbl]below:$t_{\mathrm{gen}}$}] (tg) at (11.8,2.6) {};
  \node[place,label={[lbl,text=badcol]right:$z_{\mathrm{gen}}$ ($\nu{=}15$)}] (xg) at (13.1,2.6) {};
  \node[transition,label={[guard]below:$[\mathrm{decoy}]$}] (td) at (11.8,0.6) {};
  \coordinate (fanD) at (12.4,0.6);
  \node[place,label={[lbl]right:$z_{\mathrm{dec}}$ ($\nu{=}{-}1$)}] (xd) at (13.1,1.25) {};
  \node[place,label={[lbl,text=rescol!70!black,align=left]right:telemetry\\[-1pt]$\to$ attribution}] (tel) at (13.1,-0.05) {};
  \node[transition,label={[guard]above:$[\mathrm{stall}]$}] (tst) at (6.5,-1.6) {};
  \node[place,label={[lbl]left:credentials}] (cred) at (3.8,2.9) {};
  \node[place,tokens=2,label={[lbl]below:SOC capacity}] (soc) at (3.8,-1.7) {};
  \foreach \u/\v in {dmz/tin,tin/gl,gl/rl,rl/ol,tmv/ga,ga/ra,ra/oa,tg/xg} \draw[arc] (\u)--(\v);
  \draw[arc] (ol.east)--(fanL)|-(tmv); \draw[arc] (ol.east)--(fanL)|-(tst);
  \draw[arc] (oa.east)--(fanA)|-(tg); \draw[arc] (oa.east)--(fanA)|-(td);
  \draw[arc] (td.east)--(fanD)|-(xd); \draw[arc] (td.east)--(fanD)|-(tel);
  \draw[arc,badcol] (tst.east) -- ++(0.45,0) |- (2.5,-2.95) -- (gl.south);
  \node[lbl,text=badcol,anchor=north] at (5.2,-3.0) {stall: payoff $\beta$, the foothold token returns};
  \draw[darc] (rl.north) -- node[pos=0.55,left,lbl]{if T1078 succeeds} (cred.south);
  \draw[{Stealth[length=2mm]}-{Stealth[length=2mm]},line width=0.7pt,dashed,draw=inkcol!60] (cred.east) -| (ga.north);
  \node[lbl,anchor=south] at (6.2,2.95) {read arc: sets $\vartheta$ of $\G_{p_A}$};
  \draw[{Stealth[length=2mm]}-{Stealth[length=2mm]},line width=0.7pt,draw=rescol!70!black] (soc.north) -- (rl.south);
\end{tikzpicture}
\caption{The intrusion campaign as an OCGPN. The red token is the attacker's foothold. At each game place (orange), the attacker and defender play a technique-level game, and the outcome decides where the foothold token goes. Three things a segment-level Markov game cannot express directly are drawn explicitly: a \emph{credential} token that the lateral game may emit and the asset game reads; a \emph{telemetry} token that decoy exfiltration produces for the defender (a fork); and a shared place for \emph{security-operations-center (SOC) capacity}, for which concurrent engagements compete.}
\label{fig:campaignnet}
\end{figure}
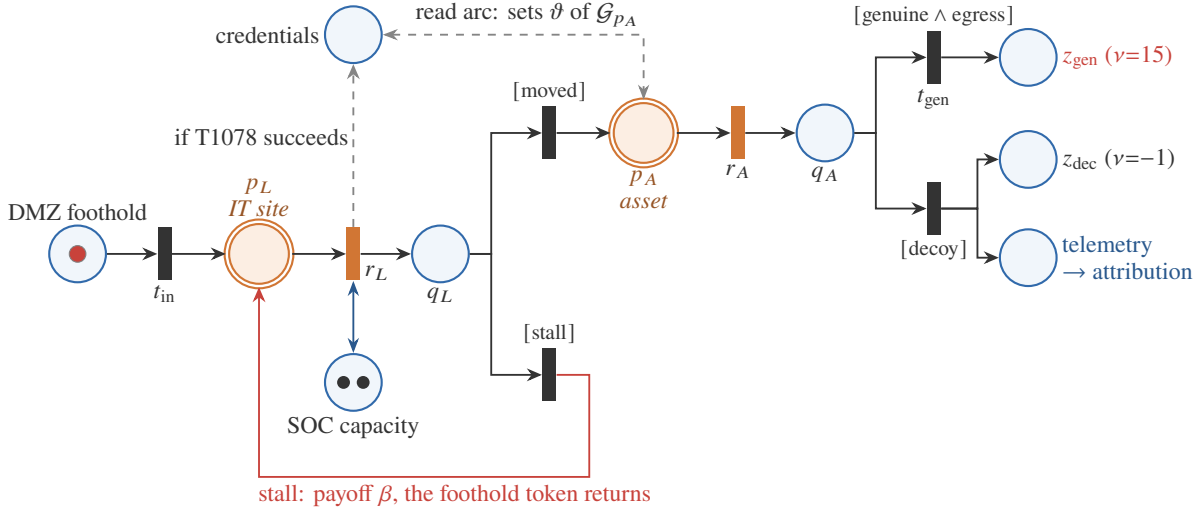

\paragraph{The asset game $\G_{p_A}$ (deception).} The defender chooses a posture, \emph{Plain} (no decoys) or \emph{Decoy} (seed the share with decoy files, at a cost of $1$; \cite{yang2025crossechelon} uses $0.6$, rounded here for hand calculation). The attacker cannot observe the posture. It either \emph{Collects} immediately or \emph{Enumerates} first, at a cost of $1$. Enumeration identifies the genuine files with probability $3/4$. Facing decoys, an immediate collection grabs a decoy with probability $1/2$. The outcome token is labeled \emph{genuine} or \emph{decoy}. Following~\cite{yang2025crossechelon}, the exfiltration rewards $\nu(s^0)=15$ and $\nu(s^6)=-1$ are stage payoffs of $\G_{p_A}$, and the terminal places $z_{\mathrm{gen}},z_{\mathrm{dec}}$ have value zero. Hence $Q^A=R^A$. In the zero-sum convention, the defender's decoy cost is added to the attacker's payoff. \Cref{tab:assetgame} gives the matrix $Q^A$.

\begin{table}[ht]
\centering
\caption{The augmented asset game $Q^A$, attacker payoffs (zero-sum). The boxed entry is a saddle point.}
\label{tab:assetgame}
\begin{tabular}{l|cc}
\toprule
Attacker $\backslash$ Defender & Plain & Decoy\\\midrule
Collect & $15$ & $\tfrac12\cdot15+\tfrac12\cdot(-1)+1=8$\\
Enumerate & $15-1=14$ & $\tfrac34\cdot15+\tfrac14\cdot(-1)-1+1=\boxed{11}$\\
\bottomrule
\end{tabular}
\end{table}

\paragraph{The lateral game $\G_{p_L}$.} The attacker chooses between two lateral-movement techniques, T1021 (remote services) and T1078 (valid accounts). The defender chooses between two countermeasures: network filtering, which covers T1021, and credential monitoring, which covers T1078. If the countermeasure matches the technique, the move succeeds with probability $\lambda_A(1-\eta)=0.09$. Otherwise it succeeds with probability $\lambda_A=0.6$. On success, the foothold token moves to $p_A$ and the attacker collects $\nu(A)$. On failure, the token returns to $p_L$ along the rework arc, and the attacker collects $\beta$. The outcome-probability matrices are therefore
\begin{equation}
  \Phi^L_{\mathrm{moved}}=\begin{pmatrix}0.09&0.6\\0.6&0.09\end{pmatrix},\qquad \Phi^L_{\mathrm{stall}}=J-\Phi^L_{\mathrm{moved}},
\end{equation}
with rows T1021 and T1078 and columns filtering and monitoring, and \eqref{eq:Qmatrix} reads
\begin{equation}\label{eq:lateral}
  Q^{L}(V)=\nu(A)\,\Phi^L_{\mathrm{moved}}+\beta\,\Phi^L_{\mathrm{stall}}+\gamma\big(V_A\,\Phi^L_{\mathrm{moved}}+V_L\,\Phi^L_{\mathrm{stall}}\big)
  =\begin{pmatrix} Q_{\mathrm{m}} & Q_{\mathrm{u}}\\ Q_{\mathrm{u}} & Q_{\mathrm{m}}\end{pmatrix},
\end{equation}
where $U_{\mathrm{mv}}:=\nu(A)+\gamma V_A$ is the value of moving, $U_{\mathrm{st}}:=\beta+\gamma V_L$ the value of stalling, $Q_{\mathrm{m}} = 0.09\,U_{\mathrm{mv}} + 0.91\,U_{\mathrm{st}}$ the entry when the countermeasure matches, and $Q_{\mathrm{u}} = 0.6\,U_{\mathrm{mv}} + 0.4\,U_{\mathrm{st}}$ the entry when it does not.

\subsection{Strategic analysis and values}
\label{sec:re2-values}

The place graph has edges $p_L\to p_L$ (the stall loop), $p_L\to p_A$, and $p_A\to\{z_{\mathrm{gen}},z_{\mathrm{dec}}\}$. Its SCCs in reverse topological order are $\{p_A\}$ and then $\{p_L\}$. By \cref{thm:scc}, Algorithm~\ref{alg:scc} solves the asset game once, then the lateral component, which has a self-loop and is therefore a fixed point.

\begin{proposition}[Propagation of deception value]\label{prop:deception}
In the campaign net of \cref{fig:campaignnet}:
\begin{enumerate}[label=(\roman*),nosep]
  \item $\G_{p_A}$ has the pure saddle point $(\text{Enumerate},\text{Decoy})$ with value $V_A=11$. Without the decoy option, $V_A=15$.
  \item $\G_{p_L}$ has a unique equilibrium in which both players mix $\tfrac12$--$\tfrac12$, and its value is
  \begin{equation}\label{eq:VL}
    V_L = \frac{\bar p\,U_{\mathrm{mv}} + (1-\bar p)\,\beta}{1-\gamma(1-\bar p)},\qquad \bar p = \frac{\lambda_A(1-\eta)+\lambda_A}{2}=0.345 .
  \end{equation}
  \item Hence the decoy posture at $p_A$ lowers the attacker's value at the upstream place $p_L$ from $14.75$ to $12.11$, although no action at $p_L$ has changed.
\end{enumerate}
\end{proposition}
\begin{proof}
(i) In \cref{tab:assetgame}, the row minima are $8$ and $11$, and the column maxima are $15$ and $11$. So the LP \eqref{eq:lp} has value $11$, attained at (Enumerate, Decoy). Without decoys, only the Plain column remains, and Collect yields $15$. (ii) We have $Q_{\mathrm{u}}-Q_{\mathrm{m}}=0.51\,(U_{\mathrm{mv}}-U_{\mathrm{st}})$, so $Q_{\mathrm{m}}<Q_{\mathrm{u}}$ if and only if moving is worth more than stalling. At the fixed point, $U_{\mathrm{mv}}=19.35>U_{\mathrm{st}}=8.29$. For any $Q_{\mathrm{m}}\neq Q_{\mathrm{u}}$, the symmetric matrix \eqref{eq:lateral} has no saddle point and the unique equilibrium $\sigma_L=\big((\tfrac12,\tfrac12),(\tfrac12,\tfrac12)\big)$; if $Q_{\mathrm{m}}=Q_{\mathrm{u}}$, every profile would be an equilibrium. By \eqref{eq:routing}, with places ordered $(p_L,p_A)$,
\begin{equation}
  P^\sigma=\begin{pmatrix}1-\bar p & \bar p\\ 0&0\end{pmatrix}=\begin{pmatrix}0.655&0.345\\0&0\end{pmatrix},\qquad
  \mathbf{r}^\sigma=\begin{pmatrix}\bar p\,\nu(A)+(1-\bar p)\beta\\ V_A\end{pmatrix}=\begin{pmatrix}2.14\\ 11\end{pmatrix}.
\end{equation}
The closed form \eqref{eq:neumann} gives $V^\sigma=(I-\gamma P^\sigma)^{-1}\mathbf{r}^\sigma$, whose first entry is \eqref{eq:VL}. (iii) Substitute $U_{\mathrm{mv}}=10+0.85\,V_A$ with $V_A\in\{15,11\}$.
\end{proof}

\paragraph{Termination and outcome probabilities.} Here $\rho(P^\sigma)=0.655<1$, so play terminates with probability one (\cref{prop:termination}). The absorbing-chain formula \eqref{eq:absorb} gives $\big[(I-P^\sigma)^{-1}\big]_{LL}=1/\bar p=2.90$ expected engagements at the IT site. The attacker reaches the asset server with probability one, and it exfiltrates genuine data with probability $\tfrac34$ under (Enumerate, Decoy).

\subsection{Structural analysis of guards and deception}
\label{sec:re2-structural}

Decoys change $\Phi^{A}$ but not the incidence matrix $C$. The genuine-exfiltration column remains, so no Farkas certificate (\cref{prop:farkas}) for $m_B=\mathbf{e}_{z_{\mathrm{gen}}}$ can exist, and indeed $z_{\mathrm{gen}}$ is reached with probability $\tfrac34$. The genuine-exfiltration transition $t_{\mathrm{gen}}$ carries the guard $[\mathrm{genuine}\wedge\mathrm{egress}]$, where $\mathrm{egress}$ is a binding-level flag for whether egress is allowed. An egress lockdown sets this flag to false and thereby deletes the column. The row of $C$ for $z_{\mathrm{gen}}$ then becomes zero, and $y=\mathbf{e}_{z_{\mathrm{gen}}}$ satisfies \eqref{eq:farkas}, since $C^\top y=0$ and $y^\top(m_B-M_0)=1>0$. Unreachability is then certified for every equilibrium, at an availability cost. This is the design rule of \cref{sec:structural} in an adversarial setting: decoys are an incentive instrument and lower the value; the guard is a structural instrument and removes the outcome.

\subsection{Composition}
\label{sec:re2-composition}

The segment-level Markov game of~\cite{yang2025multires} rests on a narrow interface: the only thing a segment reports is where the attacker ends up. Its authors list the main ways this assumption fails: stolen credentials that are valid elsewhere, persistent implants, harvested reconnaissance, and intelligence gained by the defender~\cite{yang2025crossechelon}. In a Markov-game formulation, each must be encoded in the global state. In an OCGPN, each is a typed token in a side place, read only by the games that need it. An exact analysis still needs the artifact in its Markov state, but only within the artifact window of \cref{prop:instantiation}, so the representation is structured and localized rather than global.

\begin{description}[leftmargin=1.2em,font=\normalfont\bfseries]
  \item[Artifacts (dynamic instantiation).] Suppose T1078 succeeds at $p_L$. Then $\Out_{p_L}$ also emits a credential token, which the asset game reads through a read arc (\cref{fig:campaignnet}). The token selects the template parameter $\vartheta=\mathrm{cred}$, under which enumeration is free. So $p_A$ unfolds into two instances, $p_A^{\mathrm{plain}}$ with the matrix of \cref{tab:assetgame} and $p_A^{\mathrm{cred}}$ with $\left(\begin{smallmatrix}15&8\\15&12\end{smallmatrix}\right)$ and saddle-point value $V_A^{\mathrm{cred}}=12$. The continuation value of a successful move now depends on the technique: $U_{\mathrm{mv}}^{0}=\nu(A)+\gamma V_A=19.35$ after T1021, and $U_{\mathrm{mv}}^{1}=\nu(A)+\gamma V_A^{\mathrm{cred}}=20.20$ after T1078. With $U_{\mathrm{st}}=\beta+\gamma V_L$, the augmented lateral game \eqref{eq:Qmatrix} becomes
\begin{equation}\label{eq:lateralcred}
  Q^L=\begin{pmatrix}0.09\,U_{\mathrm{mv}}^{0}+0.91\,U_{\mathrm{st}} & 0.6\,U_{\mathrm{mv}}^{0}+0.4\,U_{\mathrm{st}}\\ 0.6\,U_{\mathrm{mv}}^{1}+0.4\,U_{\mathrm{st}} & 0.09\,U_{\mathrm{mv}}^{1}+0.91\,U_{\mathrm{st}}\end{pmatrix}
\end{equation}
(rows T1021, T1078; columns filtering, monitoring). It is no longer symmetric. Iterating $V_L\gets\mathrm{val}(Q^L)$ to its fixed point gives $V_L=12.42$, $U_{\mathrm{st}}=8.56$, and $Q^L=\left(\begin{smallmatrix}9.53&15.03\\15.54&9.61\end{smallmatrix}\right)$. The equilibrium is fully mixed: the attacker plays T1078 with probability $0.48$, and the defender plays credential monitoring with probability $0.53$. Only the instance of $\G_{p_A}$ doubles; the rest of the net is untouched.
  \item[Telemetry from deception (fork).] Decoy exfiltration ends the attacker's campaign at $z_{\mathrm{dec}}$. It also emits a \emph{telemetry} token for the defender, which can enable an attribution or eviction game. By \cref{thm:composition}(ii), that downstream game affects the campaign only through its entry value. In the zero-sum Markov game, this gain can only be folded into the scalar $\nu(s^6)$.
  \item[Shared defensive capacity (parallel).] Countermeasures consume SOC-capacity tokens. If two campaigns run concurrently, the composed incidence matrix has the form \eqref{eq:parallelC}, with the capacity place as the shared row. A capacity-conservation invariant of each campaign glues into one for the composition by \cref{thm:composition}(i). Strategically, however, the defender's feasible actions at one game place now depend on the marking left by the other, so \cref{ass:sequential} fails and values must be computed over markings.
\end{description}

\subsection{Resolution and scale}
\label{sec:re2-resolution}

Keeping a place coarse means replacing its local matrices with a summary (\cref{fig:zoom}). The coarse summary of $p_A$ lets the attacker \emph{choose} the genuine adjacency and gives value $15$, so $\varepsilon_A=15-11=4$. \Cref{cor:coarse} bounds the resulting error by $4/(1-0.85)=26.7$. The actual error at $p_L$ is $14.75-12.11=2.65$. The bound is conservative, but it points to the right place: deception is where the coarse summary fails. By \cref{prop:complexity}, the cost is additive in the local game sizes. In the 39-segment enterprise study of~\cite{yang2025crossechelon}, the decomposed model has about $10^4$ tree nodes, against roughly $10^{25}$ for a monolithic extensive form of an eight-stage campaign.

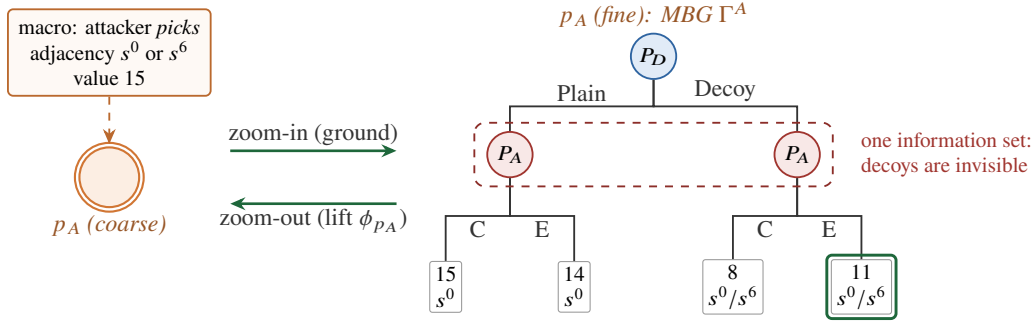
\begin{figure}[ht]
\centering
\begin{tikzpicture}[
  dn/.style={circle,draw=rescol!80!black,fill=rescol!15,line width=0.7pt,minimum size=6mm,inner sep=0pt,font=\scriptsize},
  an/.style={circle,draw=badcol!80!black,fill=badcol!12,line width=0.7pt,minimum size=6mm,inner sep=0pt,font=\scriptsize},
  lf/.style={draw=inkcol!50,fill=white,rounded corners=1pt,font=\scriptsize,inner sep=2pt,align=center},
  el/.style={line width=0.7pt,draw=inkcol}]
  \node[gplace,label={[glbl]below:$p_A$ (coarse)}] (pc) at (0,0) {};
  \node[cbox,font=\scriptsize,align=center] (cm) at (0,1.65) {macro: attacker \emph{picks}\\ adjacency $s^0$ or $s^6$\\ value $15$};
  \draw[darc,draw=gamecol!80!black] (cm.south) -- (pc.north);
  \draw[->,line width=1pt,draw=okcol!65!black] (1.6,0.35) -- node[above,lbl]{zoom-in (ground)} (3.8,0.35);
  \draw[<-,line width=1pt,draw=okcol!65!black] (1.6,-0.35) -- node[below,lbl]{zoom-out (lift $\phi_{p_A}$)} (3.8,-0.35);
  \begin{scope}[shift={(7.2,1.6)}]
    \node[dn] (r) at (0,0) {$\PD$};
    \node[an] (a1) at (-1.9,-1.3) {$\PA$};
    \node[an] (a2) at (1.9,-1.3) {$\PA$};
    \draw[el] (r.south) -- ++(0,-0.35) -| node[pos=0.25,above,lbl]{Plain} (a1.north);
    \draw[el] (r.south) -- ++(0,-0.35) -| node[pos=0.25,above,lbl]{Decoy} (a2.north);
    \node[lf] (l1) at (-2.75,-3.05) {$15$\\$s^0$};
    \node[lf] (l2) at (-1.05,-3.05) {$14$\\$s^0$};
    \node[lf] (l3) at (1.05,-3.05) {$8$\\$s^0/s^6$};
    \node[lf] (l4) at (2.75,-3.05) {$11$\\$s^0/s^6$};
    \draw[el] (a1.south) -- ++(0,-0.5) -| node[pos=0.25,below,lbl]{C} (l1.north);
    \draw[el] (a1.south) -- ++(0,-0.5) -| node[pos=0.25,below,lbl]{E} (l2.north);
    \draw[el] (a2.south) -- ++(0,-0.5) -| node[pos=0.25,below,lbl]{C} (l3.north);
    \draw[el] (a2.south) -- ++(0,-0.5) -| node[pos=0.25,below,lbl]{E} (l4.north);
    \draw[dashed,line width=0.7pt,draw=badcol!70!black,rounded corners=4pt] ($(a1.north west)+(-0.25,0.2)$) rectangle ($(a2.south east)+(0.25,-0.2)$);
    \node[font=\scriptsize,text=badcol!80!black,align=left,anchor=west] at ($(a2.east)+(0.4,0)$) {one information set:\\decoys are invisible};
    \node[draw=okcol!65!black,line width=1.1pt,rounded corners=2pt,fit=(l4),inner sep=1.5pt] {};
    \node[glbl] at (0,0.55) {$p_A$ (fine): MBG $\Gamma^A$};
  \end{scope}
\end{tikzpicture}
\caption{Resolution as place refinement. Zooming in replaces the coarse local game at $p_A$ with the extensive-form MBG, whose leaves are grounded in continuation values. Zooming out keeps only the outcome kernel $\phi_{p_A}$, which is what the routing matrix consumes. The coarse model overstates the attacker's value by $\varepsilon_A=4$. The dashed information set is the formal content of the deception. C means collect and E means enumerate.}
\label{fig:zoom}
\end{figure}

\subsection{Composing a newly discovered scenario}
\label{sec:re2-new}

Suppose threat intelligence now reports a phishing campaign against the user segment. Phishing gives the attacker a second route into the IT site, which the model of \cref{fig:campaignnet} does not contain. The defender's question is practical: \emph{does the whole campaign have to be re-modeled and re-solved?} The composition results of \cref{sec:composition} say no. This subsection works through the composition step by step.

\paragraph{The new subnet.} The new scenario is an open net $\N_{\mathrm{new}}$ (\cref{fig:newscenario}). Its entry place holds the attacker's token in a user mailbox. Its game place $p_U$ is a phishing game, and its exit place is a foothold in the IT site. At $p_U$, the attacker chooses \emph{Spearphish} (S) or \emph{Credential harvest} (H). The defender chooses \emph{Email filtering} (F) or \emph{multi-factor authentication} (M). Each profile leads to one of three outcomes: \emph{success} (the token reaches the IT site, and the attacker collects $\nu(L)=5$, the value of the IT site in~\cite{yang2025crossechelon}); \emph{retry} (the token returns to $p_U$, with stalling payoff $\beta=-2$); or \emph{evicted} (the token reaches a terminal place, and the attacker loses its tooling, payoff $-3$). \Cref{tab:phish} lists the outcome-probability matrices. Each countermeasure is effective against one technique and weaker against the other.

\begin{table}[ht]
\centering
\caption{Outcome indicators $\big(\Phi^U_{\mathrm{success}},\Phi^U_{\mathrm{retry}},\Phi^U_{\mathrm{evicted}}\big)$ of the phishing game $\G_{p_U}$. Rows are the attacker's techniques; columns are the defender's countermeasures.}
\label{tab:phish}
\begin{tabular}{l|cc}
\toprule
Attacker $\backslash$ Defender & Email filtering (F) & MFA (M)\\\midrule
Spearphish (S) & $(0.10,\ 0.50,\ 0.40)$ & $(0.60,\ 0.30,\ 0.10)$\\
Credential harvest (H) & $(0.50,\ 0.35,\ 0.15)$ & $(0.20,\ 0.50,\ 0.30)$\\
\bottomrule
\end{tabular}
\end{table}

\begin{figure}[ht]
\centering
\begin{tikzpicture}[x=1cm,y=1cm]
  \draw[line width=0.7pt,dashed,draw=badcol!70!black,rounded corners=4pt] (-0.55,-2.75) rectangle (8.15,2.45);
  \node[lbl,anchor=north west,text=badcol!70!black,align=left] at (-0.45,2.4) {new scenario $\N_{\mathrm{new}}$ (phishing)};
  \draw[line width=0.7pt,dashed,draw=rescol!70!black,rounded corners=4pt] (7.05,0.45) rectangle (12.9,2.45);
  \node[lbl,anchor=north east,text=rescol!70!black] at (12.8,2.4) {old campaign $\N_{\mathrm{old}}$ (\cref{fig:campaignnet})};
  \node[place,label={[lbl,align=center]below:user\\[-1pt]mailbox}] (mb) at (0,0) {}; \node[ctok=badcol] at (mb) {};
  \node[transition] (tp) at (1.1,0) {};
  \node[gplace,label={[glbl]above:$p_U$\\[-1pt]phishing}] (pu) at (2.3,0) {};
  \node[rtrans,label={[lbl]below:$r_U$}] (ru) at (3.5,0) {};
  \node[place,label={[lbl]below:$q_U$}] (ou) at (4.6,0) {};
  \coordinate (fan) at (5.3,0);
  \node[transition,label={[guard]below:$[\mathrm{success}]$}] (ts) at (6.1,1.45) {};
  \node[transition,label={[guard]below:$[\mathrm{evicted}]$}] (te) at (6.1,0) {};
  \node[transition,label={[guard]below:$[\mathrm{retry}]$}] (tr) at (6.1,-1.45) {};
  \node[place,label={[lbl,align=center]below:$z_{\mathrm{ev}}$\\[-1pt]reward $-3$}] (xe) at (7.45,0) {};
  \node[place,fill=okcol!18,draw=okcol!60!black,label={[lbl,text=okcol!45!black,fill=white,inner sep=1.5pt]above:IT foothold (fused)}] (it) at (7.6,1.45) {};
  \node[transition] (tin) at (8.7,1.45) {};
  \node[gplace,label={[glbl]below:$p_L$}] (pl) at (9.8,1.45) {};
  \node[gplace,label={[glbl]below:$p_A$}] (pa) at (11.0,1.45) {};
  \node[lbl,anchor=west] at (11.55,1.45) {$\cdots$};
  \foreach \u/\v in {mb/tp,tp/pu,pu/ru,ru/ou,ts/it,te/xe,it/tin,tin/pl,pl/pa} \draw[arc] (\u)--(\v);
  \draw[arc] (ou.east)--(fan)|-(ts); \draw[arc] (ou.east)--(fan)--(te); \draw[arc] (ou.east)--(fan)|-(tr);
  \draw[arc,badcol] (tr.east) -- ++(0.4,0) |- (2.3,-2.35) -- (pu.south);
  \node[lbl,text=badcol,anchor=north] at (4.6,-2.4) {retry: payoff $\beta=-2$};
  \node[draw=rescol!70!black,fill=rescol!6,rounded corners=2pt,line width=0.6pt,font=\footnotesize,align=center,anchor=north] at (10.75,0.2) {summarized by its entry value\\$V_L=12.11$ (\cref{thm:seqeq})};
\end{tikzpicture}
\caption{Composing a newly discovered phishing scenario with the existing campaign. The exit place of $\N_{\mathrm{new}}$ is fused with the entry place of $\N_{\mathrm{old}}$ (sequential composition, \cref{def:operators}(a)). The phishing game has a retry loop and an eviction exit. For the analysis of $\N_{\mathrm{new}}$, the entire old campaign is replaced by a single number, its entry value $V_L$.}
\label{fig:newscenario}
\end{figure}
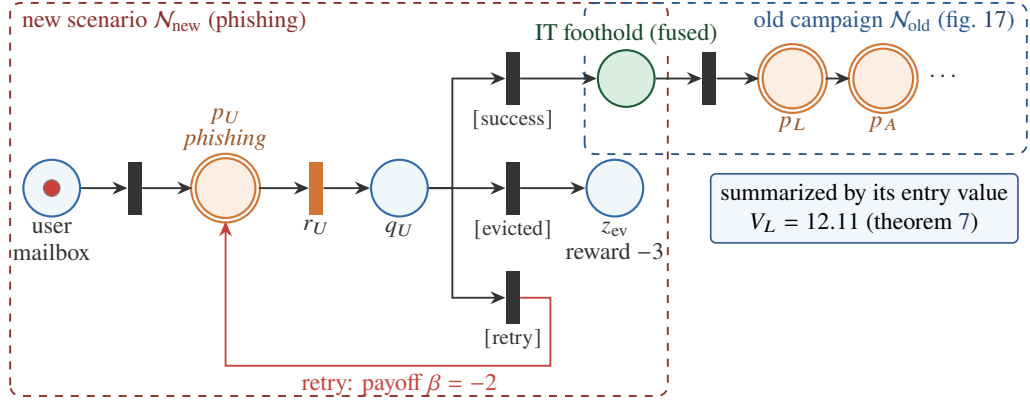

\paragraph{Step 1: structural analysis (\cref{lem:fusion}).} The composite is $\N_{\mathrm{new}};\N_{\mathrm{old}}$, obtained by fusing the IT-foothold exit of $\N_{\mathrm{new}}$ with the entry of $\N_{\mathrm{old}}$. In $\N_{\mathrm{new}}$, the vector $y_{\mathrm{new}}$ equal to $1$ on every place that can hold the attacker's token (mailbox, $p_U$, $q_U$, $z_{\mathrm{ev}}$, IT foothold) is a P-invariant. Every column moves that token from one such place to another. The analogous vector $y_{\mathrm{old}}$ is a P-invariant of $\N_{\mathrm{old}}$. Both equal $1$ on the fused place, so by \cref{lem:fusion}(i) they glue into a P-invariant of the composite: there is always exactly one attacker foothold. The retry loop is a T-invariant of $\N_{\mathrm{new}}$ alone. Composition creates no new cycle, because no arc leads from $\N_{\mathrm{old}}$ back to $\N_{\mathrm{new}}$.

\paragraph{Step 2: interface values (\cref{thm:seqeq}).} By \cref{thm:seqeq}, the equilibria of the composite are obtained by fixing an equilibrium of $\N_{\mathrm{old}}$ and solving $\N_{\mathrm{new}}$ with the old net's entry value as the terminal value of the IT foothold. From \cref{prop:deception}, this entry value is $V_L=12.11$ (with the decoy posture at $p_A$). Nothing else about $\N_{\mathrm{old}}$ is needed. By \eqref{eq:Qmatrix}, the augmented phishing game is
\begin{equation}
  Q^U(V)=\big(\nu(L)+\gamma V_L\big)\,\Phi^U_{\mathrm{success}}+\big(\beta+\gamma V_U\big)\,\Phi^U_{\mathrm{retry}}-3\,\Phi^U_{\mathrm{evicted}},
\end{equation}
where $V_U$ appears on the right because of the retry loop.

\paragraph{Step 3: the retry loop (\cref{prop:loop}).} Solving the fixed point (one $2\times2$ LP per iteration, or the closed form below) gives the equilibrium. The attacker plays Spearphish with probability $0.366$, and the defender plays Email filtering with probability $0.488$. At the fixed point,
\begin{equation}
  Q^U=\begin{pmatrix}1.84&9.78\\8.25&3.67\end{pmatrix},\qquad V_U=5.90 .
\end{equation}
Under these strategies, each attempt succeeds with probability $0.350$, is retried with probability $0.416$, and ends in eviction with probability $0.234$. In \cref{prop:loop}, $\tau^+=1$, so
\begin{equation}
  A=0.350\,(5+0.85\cdot12.11)+0.416\cdot(-2)+0.234\cdot(-3)=3.82,\qquad B=0.85\cdot0.416=0.354,
\end{equation}
and $V_U=A/(1-B)=3.82/0.646=5.90$, in agreement with the fixed point.

\paragraph{Step 4: verification and computational cost.} Solving the full three-place composite $(p_U,p_L,p_A)$ from scratch by \cref{alg:scc} gives $(V_U,V_L,V_A)=(5.90,\,12.11,\,11)$. This is exactly the compositional answer. The compositional solve needed only the new $2\times2$ fixed point. The old campaign, which in a realistic deployment would be the expensive part (in~\cite{yang2025crossechelon}, about $10^4$ tree nodes), was not touched. By \cref{prop:complexity}, the incremental cost of adding a scenario is the cost of the scenario itself.

\paragraph{Step 5: effect of downstream deception (\cref{prop:monotone}).} Suppose the defender had not deployed decoys at $p_A$. Then $V_L=14.75$ instead of $12.11$, and the same computation gives $V_U=7.12$ instead of $5.90$ (\cref{tab:phishcompare}). So a deception decision two hops downstream lowers the value of a threat that did not exist when the decision was made. The equilibrium mixtures at $p_U$ barely move, however (Spearphish $0.366$ versus $0.366$; filtering $0.488$ versus $0.487$). \Cref{prop:monotone} explains why. The probability of reaching the IT site has row differences of opposite signs across the defender's actions: Spearphish is $0.4$ less likely to succeed than Credential harvest against filtering, and $0.4$ more likely against MFA. So raising or lowering $V_L$ pushes neither technique uniformly. Downstream decoys change \emph{the stakes} of phishing, not \emph{how} it is played.

\paragraph{Step 6: outcome probabilities (\eqref{eq:absorb}).} The retry loop makes the attacker reach the IT site with probability $0.350/(1-0.416)=0.599$ and be evicted with probability $0.401$, after $1/(1-0.416)=1.71$ phishing attempts on average. From the IT site, the old net reaches the asset server with probability one, and it exfiltrates genuine data with probability $\tfrac34$ under the decoy posture. By composition, the probability that the phishing route ends in genuine exfiltration is $0.599\times0.75=0.449$, against $0.599$ without decoys.

\begin{table}[ht]
\centering
\caption{The phishing scenario composed with the old campaign, with and without decoys at the asset server.}
\label{tab:phishcompare}
\begin{tabular}{lcc}
\toprule
 & no decoys & decoys at $p_A$\\\midrule
entry value of the old net, $V_L$ & $14.75$ & $12.11$\\
value of the phishing game, $V_U$ & $7.12$ & $5.90$\\
attacker: Pr(Spearphish) & $0.366$ & $0.366$\\
defender: Pr(Email filtering) & $0.487$ & $0.488$\\
Pr(reach IT site) & $0.599$ & $0.599$\\
Pr(genuine exfiltration via phishing) & $0.599$ & $0.449$\\
\bottomrule
\end{tabular}
\end{table}

\paragraph{Discussion of the example.} A newly discovered scenario was added by one fusion. Its structural guarantee followed by gluing invariants. Its equilibrium was computed from one number supplied by the old model, and the answer agreed with a full re-solve. The effect of existing defenses on the new threat was read off from the same numbers. If the new scenario had shared resources with the old one, for example SOC capacity, \cref{prop:parallel} would no longer apply, and the composite would have to be solved over its joint state (\cref{rem:beyond}). That is the boundary of the compositional method.

\begin{intuition}[Summary of the running examples]
Running example~I shows that the outcome map shapes incentives, through pivotality. Running example~II adds three points. An adversary's game fits the same template unchanged. Failures of a narrow interface become typed tokens whose effect on the state is confined to an artifact window. And the local, place-by-place structure that makes the model readable is the same structure that makes it computable.
\end{intuition}

\section{Conclusion}
\label{sec:discussion}

This paper introduced game Petri nets, in the form of outcome-driven compositional game Petri nets (OCGPNs), as a framework for modeling and analyzing large-scale strategic systems. The central device is the game-to-token interface. A local game is attached to a place, its resolution produces typed tokens that depend on the binding, the action profile, and the outcome, and these tokens decide which interactions become possible next, with what resources and information. On this foundation we developed a theory at two levels. For general, concurrent nets, the incidence matrix supports invariants and Farkas certificates of unreachability that hold for every admissible profile. For sequential nets, the net induces a stochastic game on its game places. There, global equilibria are exactly collections of equilibria of continuation-augmented local games, incentives can be read off successor-level routing probabilities, and values are computed place by place. A composition calculus based on place fusion ties the two levels together: it glues structural certificates across subnets and composes equilibria through entry values.

\paragraph{Lessons.} The two running examples, read through the theory, yield five lessons that we expect to hold well beyond them.
\begin{enumerate}[label=(\arabic*),nosep]
  \item \emph{Continuation values determine upstream incentives.} Myopic equilibria are generally not globally consistent; the right local object is the augmented game (\cref{thm:consistency}). In variant A of running example~I, continuation values reverse a prisoner's dilemma (\cref{prop:reversal,prop:phase}).
  \item \emph{The structure of the interface outweighs the magnitude of rewards.} Downstream games change upstream behavior only through successors that a player's action makes more or less likely (\cref{thm:dominance,cor:myopic}). Binary routing, in which only joint success matters, can add a cooperative equilibrium but can never make cooperation dominant, whatever the reward (\cref{prop:binary}).
  \item \emph{Feedback loops can weaken incentives.} A forgiving rework loop returns play to a game in which cooperation resumes, so shirking is cheap. In variant B, cooperation survives as a global equilibrium only if rework is costly enough (\cref{prop:variantB}). This effect is invisible to an analysis that prescribes continuation values locally.
  \item \emph{Incentive instruments and guards act on different layers.} Deception lowered the attacker's value at every upstream place but left genuine exfiltration reachable. An egress guard removed it with a certificate valid for every equilibrium (\cref{sec:re2-structural}). The two instruments live on different layers and should be designed together, not substituted for one another.
  \item \emph{Composition enables analysis at scale.} A newly discovered phishing scenario was composed with an existing campaign through a single number, the campaign's entry value. The answer agreed with a full re-solve (\cref{sec:re2-new}). Computation is additive in local game sizes, and artifacts such as stolen credentials enlarge the state only within their window (\cref{prop:complexity,prop:instantiation}).
\end{enumerate}

\paragraph{Scope and novelty.} Petri nets, colored Petri nets, Petri games, Markov decision processes, and stochastic games are established fields~\cite{petri1962,jensen2009,finkbeiner2017,puterman1994,shapley1953}. Several of our supporting results specialize standard theory to the place-level game: local--global consistency, the place-level Shapley operator, and backward induction (\cref{thm:consistency,thm:contraction,thm:backward}). The contributions we regard as new are the game-to-token interface and its semantics; successor-level interface design; the fusion-based composition calculus, including the compositional-equilibrium theorem (\cref{thm:seqeq}); and the systematic separation of structural certificates from strategic values.

\paragraph{Limitations.} The equilibrium and computation results require a sequential OCGPN (\cref{ass:normalform,ass:sequential,ass:discount}) with a finite unfolding (\cref{ass:finite}). Concurrent strategic interaction is represented by the general semantics, but its analysis falls back on the marking-level stochastic game, whose state space can be large. Local games are in normal form with complete information, and the scheduler is exogenous. The cost of building the unfolded net lies outside the complexity bound of \cref{prop:complexity}. The numerical examples are illustrative and not calibrated to data.

\paragraph{Future directions.} Four directions follow from these limitations and from the lessons above.
\begin{enumerate}[label=(\alph*),nosep]
  \item \emph{Concurrent composition.} Characterize when equilibria of a parallel composition with shared resource places can be computed from the parts, for example through a price on the shared capacity, without solving the marking-level game.
  \item \emph{Information.} Extend local--global consistency and the compositional-equilibrium theorem to Bayesian local games, in which token colors carry private types and beliefs propagate along token flows.
  \item \emph{Adversarial scheduling.} Treat the scheduler as a player, and characterize the guarantees that hold against every scheduler.
  \item \emph{Mechanism design on the interface.} By \cref{thm:dominance}, choosing $\Out_p$ and the routing guards so that a desired profile becomes the unique equilibrium is a linear feasibility problem in the routing probabilities. Variant A satisfies it for $\gamma>1/2$. Variant B cannot, and \cref{prop:variantB} shows that its loop must even be made costly to sustain cooperation as one of several equilibria. A general design theory, which would combine this problem with structural guards, remains to be developed.
\end{enumerate}
On the structural side, we expect the gluing of linear invariants to extend to richer interface contracts:

\begin{conjecture}[Selection-independent safety under interface contracts]
\Cref{thm:composition}(i) glues linear invariants. If each subnet satisfies an interface contract (an assume--guarantee pair on its entry and exit token colors), then the composed net satisfies the conjunction of the guarantees, whatever admissible local profiles are played and whatever scheduler is used.
\end{conjecture}

Beyond these questions, two practical steps would make the framework usable at scale. The first is a toolchain that compiles OCGPNs from workflow or attack-graph descriptions, computes certificates and values, and reports where coarse summaries are inaccurate. The second is a calibration of the examples to operational data, for instance the multi-resolution cyber-defense case studies of~\cite{yang2025crossechelon}.

\paragraph{Concluding remarks.} Strategic systems are increasingly built from many interacting agents, human and artificial, whose decisions open and close each other's options. Game Petri nets give such systems a common language. Concurrency, resources, and safety are expressed in the net. Incentives are expressed in the games. The interface between the two is a typed, analyzable object. The results of this paper suggest that this separation is what makes it possible to model such systems compositionally, to certify what they can and cannot do, and to compute how rational players will behave in them.

\begin{singlespace}

\end{singlespace}

\end{document}